%% file: main.tex
\pdfoutput=1
\documentclass{article}

\usepackage{amsmath, amsthm, amssymb}
\usepackage{natbib}
\usepackage{booktabs}
\usepackage{graphicx}
\usepackage{hyperref}
\usepackage[margin=1.4in]{geometry}
\hypersetup{colorlinks,citecolor=blue,urlcolor=blue,linkcolor=blue,
  pdftitle={Higher-Order Asymptotics of Regression-Adjusted Estimators in Randomized Trials},
  pdfauthor={Anna Kelley and Stefan Wager}}
\newcommand{\RR}{\mathbb{R}}
\newcommand{\nn}{\mathcal{N}}
\newcommand{\EEE}{\mathbb{E}}
\newcommand{\EE}[1]{\mathbb{E}\left[#1\right]}
\newcommand{\Var}[1]{\operatorname{Var}\left[#1\right]}
\newcommand{\PP}[1]{\mathbb{P}\left[#1\right]}
\newcommand{\p}[1]{\left(#1\right)}
\newcommand{\cb}[1]{\left\{#1\right\}}
\newcommand{\abs}[1]{\left|#1\right|}
\newcommand{\cond}{\;\middle|\;}
\newcommand{\htau}{\hat{\tau}}
\newcommand{\hdelta}{\hat{\delta}}
\newcommand{\hbeta}{\hat{\beta}}
\DeclareMathOperator*{\argmin}{argmin}
\newcommand{\limn}{\lim_{n\to\infty}}
\DeclareMathOperator{\MSE}{MSE}
\DeclareMathOperator{\tr}{tr}
\DeclareMathOperator{\Cov}{Cov}
\newcommand{\Bern}{\mathrm{Bern}}
\newcommand{\CRT}{\mathrm{CRT}}
\newcommand{\IREG}{\mathrm{IREG}}
\newcommand{\REG}{\mathrm{REG}}

\newcommand{\MSEexp}{\MSE^{\mathrm{exp}}}
\input{tables/gamma0_values.tex}
\newcommand{\gz}[1]{\ifcsname gz:#1\endcsname\csname gz:#1\endcsname\else\errmessage{Missing generated value #1}\fi}


\theoremstyle{plain}
\newtheorem{prop}{Proposition}

\newtheorem{coro}[prop]{Corollary}
\newtheorem{lemm}[prop]{Lemma}
\newtheorem{theo}[prop]{Theorem}

\theoremstyle{definition}

\newtheorem{exam}{Example}
\newtheorem{defi}[exam]{Definition}

\theoremstyle{remark}
\newtheorem{rema}[prop]{Remark}

\newcommand\blfootnote[1]{%
  \begingroup
  \renewcommand\thefootnote{}
  \footnote{#1}%
  \addtocounter{footnote}{-1}
  \endgroup
}

\newcommand{\monthyear}{%
  \ifcase\month\or
  January\or February\or March\or April\or May\or June\or July\or
  August\or September\or October\or November\or December\fi
  \space \number\year
}

\author{Anna Kelley \\ Stanford University \and Stefan Wager \\ Stanford University}
\date{\monthyear}
\title{Higher-Order Asymptotics of \\ Regression-Adjusted Estimators in Randomized Trials}

\begin{document}

\maketitle

\begin{abstract}
Regression adjustment is widely used to improve the precision of average treatment effect estimates from randomized controlled trials. Interacted regression weakly dominates simple regression in terms of its first-order behavior but, in 50-50 randomized trials, the asymptotic variance of both methods is the same and so relative performance is determined by higher-order behavior. Here, we conduct a higher order analysis to compare simple vs.~interacted regression adjustments in the equal-probability case. The relative performance of the methods depends on a number of factors including feature skewness, treatment-effect heterogeneity, heteroskedasticity, and nonlinear misspecification---and can favor either estimator. However, in some practically important settings, including when features are skewed and outlying feature values are associated with higher outcome noise, simple regression can outperform interacted regression by a material amount.
\end{abstract}

\section{Introduction}

Randomized controlled trials (RCTs) are fundamental to how causal effects are ascertained across a number of fields. One valuable attribute of RCTs is that they enable unbiased estimation of average treatment effects using simple methods such as the difference in means, without reliance on parametric assumptions \citep{imbens2015causal}. However, while the difference-in-means estimator is unbiased, this does not automatically make it the preferred approach; and in particular regression adjustments are often used to improve power. For example, the Food and Drug Administration issued recent guidance encouraging analysts to ``adjust for baseline covariates that are strongly associated with an outcome in a clinical trial population'' \citep[accompanying guidance snapshot]{FDA2023}.\blfootnote{\hspace{-5.5mm}Code, simulation scripts and the computational certificates used in the proofs: \url{https://github.com/kelleyan-pixel/higher-order-regression-adjustment-for-rct}. Claude Opus 5.5 was used for code generation, including numerical simulation and computational proof design via SymPy \citep{sympy}; GPT-6 Astra was used for checking proofs and writing assistance throughout the appendix. }

The study of regression adjustments in RCTs has a long tradition in statistics. \citet{fisher1935design} used arguments rooted in parametric linear modeling to advocate for such regression adjustments. Fisher's analysis was later critiqued by \citet{Freedman2008}, who argued that reliance on linear models was misaligned with model-free guarantees sought in RCT inference. \citet{Lin2013}, however, responded to Freedman's critique by showing that regression adjustments do not require parametric assumptions to be near-unbiased for the average treatment effect. Furthermore, \citet{Lin2013} showed that regression adjustments with full treatment-control interaction weakly dominate both the difference in means estimator and Fisher's simple regression adjustment in terms of their leading-order (root-$n$ scale) errors. Lin's weak dominance result has led such treatment-interacted regressions to be often presented as a safe starting point for use of regression adjustments in RCTs \citep{li2017general,negi2021revisiting,Wager2025}.

A gap in this narrative, however, is that treatment-interacted regressions only strictly dominate simple regression in designs with randomization probability $\pi \neq 0.5$. In evenly randomized trials with $\pi = 0.5$, the first-order accuracy of both methods is equal for any data-generating distribution, i.e., first-order theory offers no guidance on which method to prefer. And, while $\pi = 0.5$ is of course a special case, this is an important special case in applications, and it is not a priori clear we should prefer interacted regression to simple regression here. For example, \citet{zhao2023covariate} argue based on simulation evidence and heuristics on the number of free parameters that simple regression may in practice be more stable than interacted regression with limited sample size---especially when there are multiple treatment levels.
Just to give a concrete example, suppose that
\begin{equation}
\label{eq:warmup}
\log(X_i) \sim \nn\p{0, \, 1.1^2}, \ \ \ Y_i(w) = 0.9 X_i + w + (1 + X_i) U_i, \ \ \ U_i \sim \nn\p{0, \, 1}.
\end{equation}
Then, at sample size $n = 500$, interacted regression has 7\% higher mean-squared error for the average treatment effect than simple regression.


Motivated by this gap, we here conduct a higher-order asymptotic analysis of both simple and interacted regression adjustments in RCTs to elucidate relative merits of both approaches in the case where $\pi = 0.5$. We frame our findings in terms of asymptotic deficiency \citep{HodgesLehmann1970}. When two estimators have matching first-order behavior but one dominates the other on the second-order scale then, asymptotically, switching to the better one gives an accuracy boost that is equivalent to a fixed boost in sample size, which is referred to as the deficiency. Our main finding is that, while either the simple or interacted regression adjustment may have better second-order performance, this trade-off appears to favor simple regression in some realistic settings. 

As a practical extension of this theory, we study feasible inference choices for both estimators with heteroskedasticity-robust standard error estimates. In particular, we consider both interacted and simple regression point estimates in tandem with HC1 and jackknife/HC3 standard error estimates \citep{EfronStein1981,MacKinnonWhite1985}. For both point estimates, HC1 can undercover substantially under skewed, tail-heteroskedastic covariates. The jackknife improves coverage in most designs, especially at lower nominal levels, but is still liable to undercover in cases of rare outliers or large HTE. In a typical low-HTE regime, simple regression combined with the jackknife achieves at least nominal 95\% coverage in all of our designs, with intervals no wider than those of interacted regression under either variance estimator.

\subsection{Related Work}
 
\paragraph{Regression adjustment in randomized experiments.}
The literature on regression adjustment in randomized experiments sits in two camps. The design-based literature treats the potential outcomes and covariates as fixed in a finite population, with randomness coming solely from treatment assignment. In this setting, \citet{Lin2013} shows that interacted regression weakly dominates both simple regression and difference-in-means in asymptotic efficiency. Meanwhile, \citet{zhao2023covariate} compare their finite sample behavior and allow for multi-armed and factorial experiments, postulating that interacted regression exhibits poor finite sample performance due to its additional model complexity. Another interesting but more distant line of work, including \citet{lei2021regression} and \citet{ChiangMatsushitaOtsu2026}, addresses bias and inference of the interacted estimator as $p$ grows with $n$.
 
By contrast, the superpopulation framework treats the experimental units as random draws from a larger population and studies repeated sampling over both the units and treatment assignment. Most past work here has focused on first-order comparisons. For instance, \citet{negi2021revisiting} show that the interacted regression estimator weakly dominates both difference in means and simple regression in this framework as well. We continue this line of work to formally investigate higher-order effects of both estimators, in particular to resolve the ambiguous optimality at $\pi=0.5$. Importantly, our goal is to characterize the higher-order behavior of the two estimators under IID sampling specifically, and it is possible that a design-based analysis yields different second-order results.
 
\paragraph{Higher-order asymptotic theory.}
First-order theory inherently cannot compare estimators with the same asymptotic variance, even though their finite-sample behavior may differ meaningfully. To address this, \citet{Nagar1959} introduced expansions of the higher-order moments of estimators, and \citet{Sargan1974} gave conditions for their validity; meanwhile Edgeworth expansions approximate their underlying distribution \citep{Sargan1976,BhattacharyaRao1976,Hall1992,Ghosh1994}. As we focus on the mean-squared error, we primarily build on the former tools, although the latter are useful for rigorous results on tail behavior and inference. \citet{HodgesLehmann1970} introduced asymptotic deficiency as a measure of the additional sample size required by the less accurate of two first-order equivalent procedures, and \citet{Rothenberg1984} emphasized that higher-order approximations should be consulted when competing estimators have the same limiting distribution. Parallel to our contribution, \citet{NeweySmith2004} present a higher-order comparison of the first-order equivalent general method of moments and generalized empirical likelihood estimators. \citet{Tan2014} studies the second-order properties of regression and likelihood calibration estimators in survey sampling and missing-data problems, finding that neither consistently dominates. \citet{ChiangMatsushitaOtsu2026} derive second-order stochastic expansions of regression-adjusted estimators to compare inference in a design-based, growing $p$ regime. However, to our knowledge we are the first to compare interacted and simple regression in a superpopulation setting using higher-order expansions.
 
\paragraph{Variance estimation.}
Under superpopulation sampling, the fixed-design based conventional sandwich variance for the interacted estimator ignores the variance of the sample covariate mean used for centering \citep[Section~5.3]{imbens2009recent}; see also \citet[Section~5]{Lin2013}. \citet{imbens2009recent} add an estimate of the omitted term, and \citet{negi2021revisiting} derive the asymptotic variance that accounts for it; we take this result for granted. General finite-sample behavior of heteroskedasticity-robust variance estimators is studied by \citet{EfronStein1981}, \citet{MacKinnonWhite1985}, \citet{ChesherJewitt1987}, \citet{LongErvin2000} and \citet{imbens2016robust}; the HC* nomenclature was introduced in \citet{MacKinnonWhite1985}. It is well known that despite their popularity in applied research, HC0 and HC1 are prone to undercovering in finite samples, and the statistics literature typically recommends switching to the jackknife (or HC3 in the \citet{MacKinnonWhite1985} terminology). \citet{lei2021regression} and \citet{ChiangMatsushitaOtsu2026} demonstrate this point specifically in the case of regression adjustment inference. Following this work, we present an empirical study of feasible inference using regression adjustment combined with robust standard errors, but we do not claim to contribute to this literature from a theoretical standpoint.

\section{Asymptotics of Regression Adjustment}
\label{sec:asymptotics}

We focus on completely randomized trials (CRTs) with causal effects formalized using the potential outcomes model \citep{imbens2015causal}. We assume the covariates $X_i\in \RR^p$ and potential outcomes $Y_i(0), \, Y_i(1) \in \RR$ are drawn IID from a distribution $P$, with average treatment effect $\tau = \EE{Y_i(1) - Y_i(0)}$. Both regression estimators below include an intercept, and interacted regression is evaluated at the sample mean of the covariates, so translating the covariates leaves both estimators unchanged; we therefore assume without loss of generality that $\EE{X}=0$. Treatment is assigned randomly conditional on the fixed treated group size $n_1=\sum_{i=1}^n W_i=\lfloor \pi n\rfloor$; treatment assignments are exchangeable but not independent.  The analyst observes $(X_i, \, Y_i, \, W_i)$ for $i = 1, \, \ldots, \, n$, where $Y_i = Y_i(W_i)$. We consider the following three estimators of the ATE:
\begin{equation}
\label{eq:estimators}
\begin{split}
&\htau_{DM} = \frac{1}{\abs{\cb{i : W_i = 1}}} \sum_{\cb{i : W_i = 1}} Y_i - \frac{1}{\abs{\cb{i : W_i = 0}}} \sum_{\cb{i : W_i = 0}} Y_i, \\
&\htau_{REG} = \htau,  \text{ with }  Y_i \sim \alpha + X_i \beta + W_i \tau, \\
&\htau_{IREG} = \htau + \frac{1}{n} \sum_{i = 1}^n X_i \hdelta,  \text{ with } Y_i \sim \alpha + X_i \beta + W_i \tau+  W_i X_i \delta,
\end{split}
\end{equation}
where, e.g., $Y_i \sim \alpha + X_i \beta + W_i \tau$ means we algorithmically run a regression of $Y_i$ against $X_i$ and $W_i$ with an intercept (without making any claims on well-specification of such a regression). To describe their first-order behavior it is helpful to define the best linear approximation to the conditional response function of the potential outcomes,
\begin{equation}
\label{eq:regr}
\cb{\alpha_{(w)}^*, \ \beta_{(w)}^*} =\argmin_{\alpha_{(w)},\,\beta_{(w)}} \cb{\EE{\p{Y_i(w) - X_i \beta_{(w)} - \alpha_{(w)}}^2}},
\end{equation}
together with the oracle residuals
\begin{equation}
\label{eq:epsi}
\varepsilon_i(w) = Y_i(w)-\alpha^*_{(w)}-X_i^\top\beta^*_{(w)},
\end{equation}
Writing $Z_i=\Sigma^{-1/2}X_i$ with $\Sigma=\Var{X_i}$, we have $\EEE[Z_i]=0$, $\EEE[Z_iZ_i^\top]=I_p$, $\EEE[Z_i\varepsilon_i(w)]=0$ and $\varepsilon_i = \varepsilon_i(W_i)$. The first-order behavior of the three estimators is well known \citep{negi2021revisiting}.

\begin{prop}
\label{prop:first_order}
Suppose that $X_i$ and $Y_i(w)$ have finite second moments and that $\Var{X_i}$ is invertible. Then all three estimators above are asymptotically normal, $\sqrt{n}(\htau_\cdot - \tau) \Rightarrow \nn(0, \, V_\cdot)$, with limiting variances
\begin{equation}
\begin{split}
&V_{DM} = V_{IREG} +  \Var{X_i\p{\sqrt{\frac{\pi}{1 - \pi}} \, \beta_{(0)}^* + \sqrt{\frac{1 - \pi}{\pi}} \, \beta_{(1)}^*}}, \\
&V_{REG} = V_{IREG} + \frac{\p{1 - 2\pi}^2}{\pi(1 - \pi)} \Var{X_i\p{\beta_{(1)}^* - \beta_{(0)}^*}}, \\
&V_{IREG} = \Var{X_i\p{\beta_{(1)}^* - \beta_{(0)}^*}} + \sum_{w = 0, \, 1} \frac{1}{\PP{W_i = w}} \EE{\p{Y_i(w) - X_i\beta_{(w)}^* - \alpha_{(w)}^*}^2}.
\end{split}
\end{equation}
\end{prop}

Thus $\htau_{IREG}$ weakly dominates both $\htau_{DM}$ and $\htau_{REG}$ at every $\pi$; its advantage over $\htau_{REG}$ is strict for generic distributions $P$ when $\pi\neq0.5$, while at $\pi = 0.5$ the two regression estimators have the same asymptotic variance
\begin{equation}
\label{eq:Vstar}
V^* = \Var{X_i\p{\beta_{(1)}^* - \beta_{(0)}^*}} + 2\EE{\varepsilon^2_i(1)} + 2\EE{\varepsilon^2_i(0)}.
\end{equation}
This is the setting we study. Because first-order theory is indifferent between the two methods, we turn to second-order asymptotics.

Suppose that we are comparing two estimators, $\htau_A$ and $\htau_B$, of a shared quantity $\tau$, whose mean-squared error rates are equivalent on the (first-order) $1/n$ but not on the (second-order) $1/n^2$ scale,
\begin{equation}
\label{eq:MSE_scaling}
\EE{\p{\htau_A - \tau}^2} = \frac{V}{n} + \frac{c_A}{n^2} + o\p{\frac{1}{n^2}}, \ \
\EE{\p{\htau_B - \tau}^2} = \frac{V}{n} + \frac{c_B}{n^2} + o\p{\frac{1}{n^2}}.
\end{equation}
\citet{HodgesLehmann1970} propose quantifying this difference in performance in terms of the number of extra samples required to compensate for using the less accurate estimator, i.e., they ask to solve for $k_n$ such that $\htau_B$ run on $k_n$ samples matches the accuracy of $\htau_A$ on $n$ samples as closely as possible. When this gap admits a limit $\limn (k_n - n) = \Delta \in \RR$, \citet{HodgesLehmann1970} refer to this limit $\Delta$ as the asymptotic deficiency of $\htau_B$ relative to $\htau_A$. Under the scaling \eqref{eq:MSE_scaling}, it is immediate that this limit exists, and that the asymptotic deficiency is:
\begin{equation}
\label{eq:deficiency}
\Delta = \p{c_B - c_A}/\,V.
\end{equation}
Following this program, our goal is to establish scaling results of the form \eqref{eq:MSE_scaling} for both $\htau_{REG}$ and $\htau_{IREG}$, thus enabling us to measure the asymptotic deficiency of $\htau_{IREG}$ relative to $\htau_{REG}$.

To focus on the typical estimator behavior rather than risk dominated by rare error blow-up, we consider a capped mean-squared error. Define
\begin{equation}
\label{eq:capped_mse}
\overline{\operatorname{MSE}}\p{\htau} = \EE{\p{\htau - \tau}^2 \land 1}
\end{equation}
This updated loss agrees with the regular squared error whenever $|\htau-\tau|\le 1$ while bounding the contribution of unstable regression fits. We use the same fixed cutoff for both estimators. When a fit is not uniquely defined, set $\hat{\tau}$ to any fixed finite value; the capped loss makes the choice immaterial.

The proof of the following result is given in Appendix \ref{sec:proofs}. Appendix \ref{sec:numerical-verification} gives simulation-based comparisons of realized, uncapped MSE with the formula in a variety of settings.

\begin{theo}
\label{theo:deficiency}
Consider a CRT on $n$ units with $n_1 = \lfloor n/2\rfloor$, and suppose that
$(X_i, Y_i(0), Y_i(1))$ are IID from a distribution $P$, $\Sigma = \Var{X_i} \succ 0$, and
\[
\EEE\big[\|X_i\|^{12}\big], \ \EEE\big[|Y_i(w)|^{12}\big] < \infty \text{ for } w = 0, \, 1.
\]
Define the whitened features and corresponding effect-heterogeneity coefficients as
\begin{equation}
\label{eq:whitening}
Z_i = \Sigma^{-1/2}X_i, \quad \Sigma = \Var{X_i},
\qquad
\gamma = \Sigma^{1/2}\p{\beta^*_{(1)}-\beta^*_{(0)}},
\end{equation}
and let
\[
K=\EEE[ZZ^\top\|Z\|^2],\qquad
S=\EEE[Z\|Z\|^2],\qquad
S_\gamma=\EEE[ZZ^\top(Z^\top\gamma)],
\]
\[
\mathcal E=\EEE[ZZ^\top\varepsilon],
\qquad
\mathcal E_{W}=\tfrac14\big(\EEE[ZZ^\top\varepsilon(1)]-\EEE[ZZ^\top\varepsilon(0)]\big).
\]
Then
\begin{equation}
\begin{split}
\label{eq:Dp_general}
&\overline{\operatorname{MSE}}\p{\htau_{IREG}}-\overline{\operatorname{MSE}}\p{\htau_{REG}} = \frac{D_p}{n^2}+O(n^{-3}), \\
&D_p
=
-(2p+3)\|\gamma\|^2
+\gamma^\top K\gamma
-(\gamma^\top S)^2
-3\|S_\gamma\|_F^2\\
&\quad\quad\quad\quad-16\tr(S_\gamma \mathcal E_W)
-8(\gamma^\top S)\tr(\mathcal E_W)
+8\gamma^\top \mathcal E_WS
+8\tr(\mathcal E^2)+8S^\top \EEE[Z\varepsilon^2].
\end{split}
\end{equation}
Consequently, when $V^*>0$, the asymptotic deficiency of IREG relative to REG, computed from the capped MSE, is $\Delta = D_p/V^*$ with $V^*$ as in \eqref{eq:Vstar}.
\end{theo}

\begin{rema}[Bias and variance contributions]
\label{rema:bias}
One may naturally be interested in how the overall MSE difference decomposes into bias and variance. Writing $e_{j}=\max\{-1,\min\{\htau_j-\tau,1\}\}$, define
\[
\mathrm{Bias}_j=\EEE[e_{j}],
\qquad
\mathrm{Var}_j=\Var{e_{j}},
\qquad
\overline{\operatorname{MSE}}(\htau_j)=\mathrm{Var}_j+\mathrm{Bias}_j^2.
\]
Evaluating the leading bias term for the two estimators (Lemma \ref{lemm:bias} in Appendix \ref{sec:proofs}) gives
\[
\mathrm{Bias}_{\REG}
= -\frac{4\tr(\mathcal E_W)+\gamma^\top S}{n}+O(n^{-2}),
\qquad
\mathrm{Bias}_{\IREG}
= -\frac{4\tr(\mathcal E_W)}{n}+O(n^{-2}),
\]
so that
$n^2\{\mathrm{Bias}^2_{\IREG}-\mathrm{Bias}^2_{\REG}\}
=-(\gamma^\top S)^2-8(\gamma^\top S)\tr(\mathcal E_W)$
to leading order, indeed directly reproducing two terms;
the remaining terms are due to variance.
\end{rema}

The expression for $D_p$ combines four key blocks: a covariate shape block scaled by the linear HTE $\gamma$; terms crossing $\gamma$ with residual asymmetry; residual heteroskedasticity; and nonlinear misspecification. No single term determines the sign and either method can come up on top---even when the population linear slopes coincide (i.e., when $\gamma = 0$). When $\gamma = 0$, the expression in \eqref{eq:Dp_general} reduces to
\[
D_p=8\tr(\mathcal E^2)+8S^\top \EEE[Z\varepsilon^2],
\]
and while $\tr(\mathcal E^2) = \|\mathcal E\|_F^2 \geq 0$, the term $S^\top \EEE[Z\varepsilon^2]$ can have either sign, so $\Delta$ can be positive or negative even when $\gamma = 0$.
To interpret the scale of $\Delta$ write
\[
R^2_\tau
=
\frac{\|\gamma\|^2}{\|\gamma\|^2+4\EEE[\varepsilon_i^2]}
= \frac{\|\gamma\|^2}{V^*},
\]
which, when $V^*>0$, quantifies the share of the first-order variance $V^*$ attributable to linear treatment-effect heterogeneity.

\begin{coro}
\label{coro:homosk}
Under the conditions of Theorem \ref{theo:deficiency}, suppose furthermore that the potential outcomes satisfy a well-specified linear model with homoskedastic errors, $\EE{\varepsilon_i(w)\cond X_i}=0$ and $\EE{\varepsilon_i^2(w)\cond X_i}=\sigma_w^2$ for $w\in\{0,1\}$. Then $\mathcal E = \mathcal E_W = 0$ and $\EEE[Z\varepsilon^2]=0$. For $\gamma\neq0$,
\begin{equation}
\label{eq:homosk_deficiency}
\Delta
=
R_\tau^2
\left[
-(2p+3)
+
\frac{\gamma^\top K\gamma-(\gamma^\top S)^2-3\|S_\gamma\|_F^2}{\|\gamma\|^2}
\right].
\end{equation}
In the Gaussian sub-case $K = (p+2)I_p$, $S=0$ and $S_\gamma = 0$, so that $\Delta = -(p+1)R_\tau^2$.
If $\gamma=0$, the second-order MSE difference vanishes for any covariate distribution; when $V^*>0$, $R_\tau^2=0$ and $\Delta=0$.
\end{coro}

In the Gaussian homoskedastic case, the second-order coefficient favors IREG and grows linearly in the number of covariates, with $p$ fixed as $n\to\infty$. This complicates the parameter-count heuristic of \citet{zhao2023covariate}, showing that the finite sample advantage of either estimator cannot be reduced to any simple feature of the DGP.

However, the following corollary sheds some light on why REG can outperform IREG in settings such as \eqref{eq:warmup}.

\begin{coro}
\label{coro:mult_normproportional}
Under the conditions of Theorem \ref{theo:deficiency}, suppose furthermore that the potential outcomes follow
\[
Y_i(w) = \alpha^*_{(w)} + X_i^\top\beta^*_{(w)} + \eta_i\|Z_i\|,
\]
where $\EE{\eta_i}=0$, $\EE{\eta_i^2}=\nu^2$ and $\eta_i$ is independent of $(Z_i,W_i)$. Then $\EE{\varepsilon_i(w)\cond Z_i}=0$ and $\EE{\varepsilon_i^2(w)\cond Z_i}=\nu^2\|Z_i\|^2$, where $\|\cdot\|$ is the Euclidean norm, so that $\EE{Z_i\varepsilon_i^2(w)}=\nu^2S$ and $\EE{\varepsilon_i^2}=\nu^2p$. For $\gamma\neq0$,
\begin{equation}
\label{eq:mult}
\Delta
=
R_\tau^2
\left[
-(2p+3)
+
\frac{\gamma^\top K\gamma-(\gamma^\top S)^2-3\|S_\gamma\|_F^2}{\|\gamma\|^2}
\right]
+
\frac{8\nu^2\,S^\top S}{\|\gamma\|^2+4\nu^2 p}.
\end{equation}
When $\gamma=0$ and $\nu^2>0$, the deficiency is
$\Delta=2\|S\|^2/p$, which is nonnegative and is strictly
positive whenever $S\neq0$.
\end{coro}

It is useful to consider how the terms in \eqref{eq:mult} compare in practice. In many applications the linear HTE signal is small relative to the residual departure from the regression model, so $R_\tau^2$ is small and the first summand rarely dominates unless its bracket is large. The bracketed component is of order $p$ in the leading cases---it equals $-(p+1)$ under Gaussian covariates---and its moment terms grow with the kurtosis and skewness of the underlying covariates, so it is sensitive to departures from Gaussianity in either direction.

The remaining term is nonnegative, since $8\nu^2\|S\|^2/(\|\gamma\|^2+4\nu^2p)\geq0$, and for $\nu>0$ it vanishes if and only if $S=\EEE[Z\|Z\|^2]=0$. Right-skewed covariates with a long upper tail, such as measures of user activity, engagement or spending, can produce a large third-moment vector $S$ after centering and whitening, in which case this term can dominate the first summand. The scalar case makes this transparent: with $p=1$, Corollary \ref{coro:mult_normproportional} gives
\[
\Delta
=
R_\tau^2\cb{\EE{Z_i^4}-4\EE{Z_i^3}^2-5}
+
\frac{8\nu^2\EE{Z_i^3}^2}{\gamma^2+4\nu^2}.
\]
Using Pearson's inequality $\EE{Z_i^4}\geq \EE{Z_i^3}^2+1$ \citep[see, e.g.,][]{rohatgi1989sharp} and $8\nu^2/(\gamma^2+4\nu^2) = 2(1-R^2_\tau)$,
\[
\Delta
\geq
-4R_\tau^2 + (2-5R_\tau^2)\EE{Z_i^3}^2 .
\]
So when $R_\tau^2 < 2/5$ the deficiency is positive as soon as $\EE{Z_i^3}^2 > 4R^2_\tau/(2-5R^2_\tau)$, and the lower bound grows linearly in $\EE{Z_i^3}^2$. 

\section{Feasible Inference}
\label{sec:feasible}

While the above results concern the precision of the point estimates provided by REG and IREG, in practice we are also interested in the validity of the resulting confidence intervals---which now depends jointly on how we estimate the standard error. As before, we operate in an assumption-lean framework, only considering estimators that are robust to misspecification of the underlying linear model.


We fit REG and IREG as in the R package \texttt{estimatr}, version 1.0.2 \citep{estimatr}: REG is \texttt{lm\char`\_robust(y \textasciitilde{} w + x)}, and IREG is \texttt{lm\char`\_lin(y \textasciitilde{} w, covariates = \textasciitilde{} x)}, with \texttt{se\char`\_type = "HC1"} (the package default is HC2). Because we operate in a superpopulation framework while these implementations are for fixed design sampling, we add a plug-in estimate of the variance of the sample covariate mean used for centering (\citealp[Section~5.3]{imbens2009recent}; \citealp{negi2021revisiting}):
\begin{equation}
\label{eq:vsp}
\hat G=\hdelta^\top\hat\Sigma_X\hdelta,
\qquad
\widehat{\mathrm{Var}}_{\mathrm{SP}}(\htau_\IREG)=\widehat{\mathrm{Var}}_{\mathrm{HC1}}(\htau_\IREG)+\frac{\hat G}{n},
\end{equation}
where $\hdelta$ is the vector of interaction coefficients and $\hat\Sigma_X$ the sample covariance of the covariates; $\hat G$ estimates $G=\|\gamma\|^2$. Estimation noise in the interaction coefficients can inflate the
plug-in correction, particularly when the true heterogeneity is small.

We choose to investigate HC1 as a popular benchmark: it is Stata's default and a recent survey \citep{Kranz2024} found it used in $98.1\%$ of a selected sample of 40,571 robust Stata regressions in reproduction packages of leading economics journals. We compare it with a robust alternative, the delete-one jackknife applied to the final estimate \citep{EfronStein1981,MacKinnonWhite1985}. Let $\htau_{(-i)}$ be the estimator recomputed on the sample with unit $i$ removed---both arm-specific fits are refitted and, for IREG, the covariates are recentred at the leave-one-out mean $\bar X_{(-i)}$---the jackknife variance estimate is then given by
\begin{equation}
\label{eq:jk}
\hat V_{\mathrm{JK}}
=
\frac{n-1}{n}\sum_{i=1}^n\big(\htau_{(-i)}-\bar\htau\big)^2,
\qquad
\bar\htau = \frac1n\sum_{i=1}^n \htau_{(-i)} .
\end{equation}
The jackknife offers two attractive attributes over HC1: conservativeness in expectation \citep{EfronStein1981}, and naturally accounting for superpopulation sampling variation by applying to the final packaged estimate. Thus for IREG we report HC1 with the plug-in correction \eqref{eq:vsp} alongside the jackknife; for REG, we directly report HC1 and jackknife. When removing unit $i$ leaves a rank-deficient design (unit $i$ has leverage one in its fit), $\htau_{(-i)}$ is undefined; we then omit it and apply \eqref{eq:jk} to the remaining $n'$ leave-one-out estimates, with $n'$ in place of $n$. All intervals use $t$ critical values with $n-k$ degrees of freedom, where $k$ is the number of regression coefficients.


Table \ref{tab:battery} compares all four joint procedure possibilities across \gz{hte:ba:ndesigns} DGPs. These DGPs exhibit a range of values of $R^2_\tau$ and covariate distributions, including Gaussian, skewed, heavy-tailed, and rare-value covariates. Appendix \ref{sec:designs} defines each design. All have $n=500$ and \gz{hte:reps} replications, and every procedure is applied to the same simulated datasets.

\begin{table}[tp]
\centering
\footnotesize
\setlength{\tabcolsep}{2.5pt}
\resizebox{\textwidth}{!}{%
\begin{tabular}{lr|c|cc|cc|c|cc|cc}
\toprule
& & \multicolumn{5}{c|}{\textbf{IREG}} & \multicolumn{5}{c}{\textbf{REG}} \\
& & & \multicolumn{2}{c|}{HC1 + plug-in} & \multicolumn{2}{c|}{jackknife} & & \multicolumn{2}{c|}{HC1} & \multicolumn{2}{c}{jackknife} \\
Design & $R^2_\tau$ & RMSE & width & cov. & width & cov. & RMSE & width & cov. & width & cov. \\
\midrule
\input{tables/jk_battery_body.tex} 
\bottomrule
\end{tabular}}
\caption{Root-mean-squared error of the point estimators, and median width and coverage (\%) of nominal $95\%$ intervals for the population ATE: IREG with HC1 plus the plug-in correction and with the direct jackknife \eqref{eq:jk}, and REG with HC1 and with the direct jackknife ($n=500$, \gz{hte:reps} replications; ``het.''\ denotes heteroskedastic noise). Designs are defined in Appendix \ref{sec:designs}; $g$ is the treatment-effect slope of the scalar lognormal designs. RMSE is computed over replications where both point estimators exist; coverage counts unavailable intervals as failures, and median width is computed over available intervals. In the rare-value design, \gz{jk:rare:IREG:degfrac}\% of replications have a leverage-one unit in an IREG arm, whose leave-one-out estimate is omitted from the jackknife (Section \ref{sec:feasible}).}
\label{tab:battery}
\end{table}

\begin{table}[tp]
\centering
\footnotesize
\setlength{\tabcolsep}{3pt}
\resizebox{\textwidth}{!}{%
\begin{tabular}{lr|ccc|ccc|ccc|ccc}
\toprule
& & \multicolumn{3}{c|}{IREG, HC1 + plug-in} & \multicolumn{3}{c|}{IREG, jackknife} & \multicolumn{3}{c|}{REG, HC1} & \multicolumn{3}{c}{REG, jackknife} \\
Design & $R^2_\tau$ & $80$ & $95$ & $99$ & $80$ & $95$ & $99$ & $80$ & $95$ & $99$ & $80$ & $95$ & $99$ \\
\midrule
\input{tables/jk_calibration_body.tex} 
\bottomrule
\end{tabular}}
\caption{Coverage (\%) at nominal $80\%$, $95\%$ and $99\%$ in the designs of Table \ref{tab:battery} with small or no heterogeneity, for IREG with HC1 plus the plug-in and with the jackknife, and for REG with HC1 and with the jackknife.}
\label{tab:plugin-cal}
\end{table}

When $R^2_\tau$ is small ($<0.05$), the combination of REG and jackknife is a clear winner across our selected designs. Its $95\%$ coverage ranges from \gz{jk:small:REG_JK:min}\% to \gz{jk:small:REG_JK:max}\%, its median interval width is between $\gz{jk:small:wRJK_IJK:min}$ and $\gz{jk:small:wRJK_IJK:max}$ times that of IREG with the jackknife and between $\gz{jk:small:wRJK_IHC1c:min}$ and $\gz{jk:small:wRJK_IHC1c:max}$ times that of IREG with HC1 plus the plug-in, which covers \gz{jk:small:IREG_HC1c:min}--\gz{jk:small:IREG_HC1c:max}\%. REG is also at least as precise a point estimator in these designs, with MSE ratios IREG/REG between \gz{jk:small:mseratio:min} and \gz{jk:small:mseratio:max}. At the $95\%$ level both estimators achieve only slightly anticonservative coverage with HC1; for IREG the upward bias of the plug-in partly offsets HC1's understatement of variance, and the undercoverage shows up more at other nominal levels (Table \ref{tab:plugin-cal}). The price of the jackknife is width: under the heaviest tails it is conservative, and REG+jackknife intervals are up to $\gz{jk:small:wRJK_RHC1:max}$ times as wide as REG+HC1's.

This pattern does not hold across all DGPs. As $R^2_\tau$ grows with skewed covariates, REG undercovers with either variance estimator, down to \gz{jk:ln15_large:REG_JK:cov}\% with jackknife at $R^2_\tau=\gz{jk:ln15_large:R2}$, while IREG with jackknife is the most reliable procedure, although it also reaches only \gz{jk:ln15_large:IREG_JK:cov}\% there. In the rare-value design an IREG arm contains a single rare unit in \gz{jk:rare:IREG:degfrac}\% of replications; removing that unit leaves a rank-deficient fit, so its LOO estimate is omitted, and IREG+jackknife covers \gz{jk:rare:IREG:degcov}\% of those refits against \gz{jk:rare:IREG:othercov}\% of the others. Treating any undefined leave-one-out estimate as making the entire interval unavailable would instead give IREG+jackknife coverage of \gz{jk:rare:IREG:strict}\%. REG is essentially unaffected: removing a unit leaves its pooled fit rank deficient in only \gz{jk:rare:REG:degfrac}\% of replications.

Table \ref{tab:plugin-cal} tests this pattern in the low-$R^2_\tau$ DGPs across nominal levels. HC1 intervals are close to nominal at $95\%$ and $99\%$ but undercover at $80\%$ in the skewed designs, by up to \gz{jk:small:IREG_HC1c:maxshort} percentage points for IREG with the plug-in and \gz{jk:small:REG_HC1:maxshort} for REG, so their $95\%$ coverage is partly a crossing point of a miscalibrated interval. With the jackknife, neither estimator is more than \gz{jk:small:REG_JK:maxshort} points below nominal at any of the three levels in any of these designs; under the heaviest tails it is conservative, with REG+jackknife covering up to \gz{jk:small:REG_JK:0.80:max}\% at nominal $80\%$.

\bibliographystyle{plainnat}
\bibliography{references}

\appendix

\section{Proofs}
\label{sec:proofs}

The goal of this section is to derive the expression for $D_p$ in Theorem \ref{theo:deficiency}. The argument has three main parts. First, we obtain an MSE Taylor expansion for regression estimates by expressing the estimator as a function of IID samples. Next, we calculate this second-order constant under a Bernoulli trial. Third, we remove the effect of random treatment imbalance introduced by Bernoulli assignment to relate this back to the desired CRT constant. 

Throughout this appendix, $j\in\{\REG,\IREG\}$ identifies the estimator.
The subscripts $\Bern$ and $(\CRT,r)$ indicate, respectively,
independent Bernoulli$(1/2)$ assignment and complete randomization
with $m=rn$ treated units. Expectations under these designs are
written $\EEE_\Bern$ and $\EEE_{\CRT,r}$. In particular,
\[
\overline{\operatorname{MSE}}_{\mathcal D}(\htau_j)
=\EEE_{\mathcal D}[\min\{(\htau_j-\tau)^2,1\}],
\qquad \mathcal D\in\{\Bern,(\CRT,r)\}.
\]
Unless otherwise specified, the $\CRT$ subscript defaults to
$r=\lfloor n/2\rfloor/n$. Estimators use the finite-valued
completion specified after \eqref{eq:capped_mse}.

\begin{defi}[Taylor-polynomial MSE]
\label{def:msexp}
Let $\bar{\mathcal V}=n^{-1}\sum_{i=1}^n\mathcal V_i$ and
$\hat\tau=g(\bar{\mathcal V})$ locally, where $g$ is smooth
near $\mu=\EEE[\mathcal V_i]$ and $g(\mu)=\tau$. Put
$d=\bar{\mathcal V}-\mu$, $a_i=\partial_i g(\mu)$,
$H_{ij}=\partial_{ij}g(\mu)$ and $A_{3,ijk}=\partial_{ijk}g(\mu)$. Repeated coordinate indices in this definition are summed from $1$ to the dimension of $\mathcal V_i$. The third-order Taylor expansion is
\begin{equation}
\label{eq:hall_expansion}
\tilde\tau_n
=\tau+a_i d_i+\tfrac12H_{ij}d_i d_j
+\tfrac16A_{3,ijk}d_i d_j d_k,
\end{equation}
and its mean-squared error is 
\[
\MSEexp(\hat\tau):=\EEE[(\tilde\tau_n-\tau)^2].
\]
The superscript $\mathrm{exp}$ indicates that this is the MSE of the third-order Taylor expansion of $\hat{\tau}$ (that is, $\Tilde{\tau}_n$), rather than $\hat{\tau}$ itself. 

For each arm $w\in\{0,1\}$, collect the moments needed for the
regression into a vector $U_i^{(w)}$. Its entries are $X_i$,
$Y_i(w)$, the distinct entries of $X_iX_i^\top$, and
$X_iY_i(w)$. Write
\[
\mu_w=\EEE[U_i^{(w)}],
\]
and let $\bar U_w$ denote the average of this vector over the
observations assigned to arm $w$.

\emph{Complete randomization with fixed arm sizes.}
Fix the treated fraction $r=m/n\in(0,1)$. Both regression estimators
can be computed from the two arm averages:
\[
\hat\tau_j=h_{j,r}(\bar U_1,\bar U_0),
\qquad j\in\{\REG,\IREG\},
\]
where Lemma \ref{lemm:two_sample} gives the function $h_{j,r}$.

To approximate the estimator, expand $h_{j,r}$ around the
population point $(\mu_1,\mu_0)$, keeping terms through degree
three in the deviations
\[
d=(\bar U_1-\mu_1,\bar U_0-\mu_0).
\]
Denote the resulting Taylor approximation by $\mathcal P_{j,r}$.
Its mean squared error is
\[
\MSEexp_{\CRT,r}(\hat\tau_j)
=
\EEE_{\CRT,r}\big[(\mathcal P_{j,r}-\tau)^2\big].
\]
When considering a single estimator $j$, we shorten
$\mathcal P_{j,r}$ to $\mathcal P_r$.

\emph{Bernoulli assignment.}
Here the treatment count is random, so we include it among the
sample moments. Define the IID vectors
\[
\mathcal V_i
=
\big(W_i,\;W_iU_i^{(1)},\;(1-W_i)U_i^{(0)}\big),
\]
whose population mean is
\[
\mu=\EEE[\mathcal V_i]
=
\big(1/2,\;\mu_1/2,\;\mu_0/2\big).
\]
Their sample average records the treated fraction and the two
arm-specific moment totals, with both totals divided by $n$.

For a vector $v=(v_W,v_1,v_0)$ of these raw moments, division by
the arm fractions recovers the within-arm averages. The estimator
can therefore be expressed as
\[
g_j(v)
=
h_{j,v_W}
\left(\frac{v_1}{v_W},\frac{v_0}{1-v_W}\right).
\]
Thus $g_j$ first converts the raw moments into within-arm averages,
then applies $h_{j,r}$ with $r=v_W$.

Let $\tilde\tau_{n,j}$ be the cubic Taylor approximation of $g_j$
around $\mu$, as in \eqref{eq:hall_expansion}. Define its
mean squared error by
\[
\MSEexp_\Bern(\hat\tau_j)
=
\EEE_\Bern\big[(\tilde\tau_{n,j}-\tau)^2\big].
\]
\end{defi}

\emph{Relation to capped and ordinary MSE.}
Lemma \ref{lemm:localization} connects these polynomial
approximations to the capped risk. Under the sampling and moment
conditions of Theorem \ref{theo:deficiency}, the capped MSE and
the Taylor-polynomial MSE differ by $O(n^{-3})$. They therefore
have the same coefficients through order $n^{-2}$. 

Transferring the expansion to ordinary, uncapped MSE requires
additional tail control. For example, under Bernoulli assignment,
suppose the completed estimator has finite second moment and
\[
\|\htau_j-\tilde\tau_{n,j}\|_{L^2}=o(n^{-3/2}).
\]
Because
$\|\tilde\tau_{n,j}-\tau\|_{L^2}=O(n^{-1/2})$,
Cauchy--Schwarz then gives an $o(n^{-2})$ difference between
ordinary MSE and $\MSEexp$. This is a sufficient additional
condition; it does not follow from the population moment
assumptions. For fixed-fraction complete randomization, the
corresponding condition uses $\mathcal P_{j,r}$ in place of
$\tilde\tau_{n,j}$.

We use $G=\|\gamma\|^2$. In covariate-moment tensor formulas, indices such as
$i,j,k,l,a,b,c,d$ range from $1$ to $p$ and refer to covariate
coordinates, rather than experimental units. Repeated indices are
summed over unless stated otherwise. In particular, $\delta_{jk}$ is
the Kronecker delta. Indices of the raw-moment vector $\mathcal V$
instead range over the coordinates of that vector; an estimator
subscript $j\in\{\REG,\IREG\}$ is never summed.

\subsection{Regression estimation expansion}

The next three lemmas explain how we obtain the MSE expansion. First,
conditioning on the assignment turns the experiment into two independent samples. Second, we show that a cubic Taylor polynomial approximates the capped risk. Finally, we identify the terms in that polynomial that contribute at the relevant orders $n^{-1}$ and $n^{-2}$.
\begin{lemm}[Two-sample representation]
\label{lemm:two_sample}
Assume the sampling and moment conditions of Theorem
\ref{theo:deficiency}. Let the assignment vector $W$ be independent of the sampled units, with $\sum_i W_i=m$ and $0<m<n$. Conditionally on $W$, the treated observations $\{(X_i,Y_i(1)):W_i=1\}$ and the control observations $\{(X_i,Y_i(0)):W_i=0\}$ are independent IID samples of sizes $m$ and $n-m$, from the laws of $(X,Y(1))$ and $(X,Y(0))$, respectively.

Wherever the required fits are uniquely defined, both $\htau_\REG$ and $\htau_\IREG$ can be written as $h_{j,r}(\bar U_1,\bar U_0)$, where $r=m/n$. Each map $h_{j,r}$ is infinitely differentiable in a
neighborhood of $(\mu_1,\mu_0)$, and its derivatives at
$(\mu_1,\mu_0)$ are rational functions of $r$ without poles in
$(0,1)$.
\end{lemm}

\begin{proof}
Conditioning on $W$ fixes the two disjoint sets of units assigned to
treatment and control. Since $W$ is independent of the sampled units,
this conditioning does not change their distributions. Independence
across units then gives independence within each arm and between the
two arms.

For the second claim, the coordinates of $\bar U_w$ contain all the
means and cross-moments needed for least squares. In particular, let
\[
C_w=\overline{XX^\top}_w-\bar X_w\bar X_w^\top,
\qquad
s_w=\overline{XY}_w-\bar X_w\bar Y_w.
\]
The arm-specific slopes and the pooled within-arm slope are
\[
\hat\beta_w=C_w^{-1}s_w,
\qquad
\hat\beta_{\rm pool}(r)
=\{rC_1+(1-r)C_0\}^{-1}\{rs_1+(1-r)s_0\}.
\]
Thus, writing $d_X=\bar X_1-\bar X_0$, we have
\[
\begin{aligned}
h_{\REG,r}(\bar U_1,\bar U_0)
&=\bar Y_1-\bar Y_0-d_X^\top\hat\beta_{\rm pool}(r),\\
h_{\IREG,r}(\bar U_1,\bar U_0)
&=\bar Y_1-\bar Y_0
-d_X^\top\{(1-r)\hat\beta_1+r\hat\beta_0\}.
\end{aligned}
\]
The second expression evaluates both fitted regressions at the overall covariate mean $\bar X=r\bar X_1+(1-r)\bar X_0$.

These expressions are rational functions of the arm moments. At
$(\mu_1,\mu_0)$, both arm covariance matrices, and hence their pooled
version, equal $\Sigma\succ0$. Also $d_X=0$, so both coefficient maps
equal $\tau$. The matrix inverses therefore remain well defined in a
neighborhood of the population moments. This proves smoothness, and
the displayed formulas give the claimed rational dependence of the
derivatives on $r$, with no poles in $(0,1)$.
\end{proof}

\begin{lemm}[Localization for capped loss and clipped errors]
\label{lemm:localization}
Under the sampling and moment conditions of Theorem
\ref{theo:deficiency}, let $\mathcal P_r$ be the third-order Taylor
polynomial of $h_{j,r}$ at $(\mu_1,\mu_0)$. For any 
$\eta\in(0,1/2)$,
\begin{equation}
\label{eq:capped_transfer}
\overline{\operatorname{MSE}}_{\CRT,r}(\htau_j)
=\MSEexp_{\CRT,r}(\htau_j)+O(n^{-3}),
\end{equation}
uniformly for $r=m/n\in[1/2-\eta,1/2+\eta]$. Also, with
$e_{j}=\max\{-1,\min\{\htau_j-\tau,1\}\}$,
\begin{equation}
\label{eq:clipped_bias_transfer}
\EEE[e_{j}]=\EEE[\mathcal P_r-\tau]+O(n^{-2})
\end{equation}
uniformly on the same interval. The analogous statements hold under
independent Bernoulli$(1/2)$ assignment with the unit-level cubic
Taylor polynomial $\tilde\tau_{n,j}$.
\end{lemm}

\begin{proof}
The idea is that on the event that the sample moments are close
to their expected values, the regression maps are
smooth and the cap is irrelevant. Outside this event, the bounded loss and a sixth-moment bound control the error.

\emph{Moment bounds.}
Each coordinate of $U^{(w)}$ is one of $X_k$, $Y(w)$, $X_kX_l$, or
$X_kY(w)$. The assumed twelfth moments give finite sixth moments for
all these coordinates. For the product coordinates, Cauchy--Schwarz
gives
\[
\begin{aligned}
\EEE|X_kX_l|^6
&\le\{\EEE|X_k|^{12}\,\EEE|X_l|^{12}\}^{1/2},\\
\EEE|X_kY(w)|^6
&\le\{\EEE|X_k|^{12}\,\EEE|Y(w)|^{12}\}^{1/2}.
\end{aligned}
\]
In particular, $\EEE\|U^{(w)}\|^6<\infty$.

Set
\[
\mu=(\mu_1,\mu_0),
\qquad
d=(\bar U_1-\mu_1,\bar U_0-\mu_0).
\]
We now bound moments of $d$. First, expand the products of its coordinate sample means into sums of products of centered unit-level observations. Any summand in which a unit appears only once has expectation zero: its sole centered coordinate is independent of the factors from all other units and has mean zero. Thus, a nonzero term in a product of six coordinates can involve at most three distinct units while a product of four can involve at most two. This tells us that for instance $\EEE\|d\|^6=O(n^{-3})$.

\emph{Reduce to inactive cap event.}
Lemma \ref{lemm:two_sample} and compactness of the fraction interval
allow us to choose a fixed $\delta>0$ such that, on
\[
\mathcal A=\{\|d\|\le\delta\},
\]
both regression maps are smooth, their derivatives through order six
are uniformly bounded, and $|h_{j,r}(\mu+d)-\tau|\le 1$.
The moment bound gives $\Pr(\mathcal A^c)=O(n^{-3})$. More generally,
for every $0\le k\le6$,
\[
\EEE[\|d\|^k\mathbf 1_{\mathcal A^c}]
\le\delta^{k-6}\EEE\|d\|^6
=O(n^{-3}).
\]
Thus restricting a polynomial term of degree at most six to
$\mathcal A$ changes its expectation by only $O(n^{-3})$, provided
its coefficients are uniformly bounded, as the Taylor coefficients
are here.

\emph{Compare the squared errors.}
Let $q_k$ denote the degree-$k$ Taylor term of
$h_{j,r}-\tau$, so that $\mathcal P_r-\tau=q_1+q_2+q_3$.
On $\mathcal A$, expanding the squared error through degree five gives
\[
\begin{aligned}
(h_{j,r}(\mu+d)-\tau)^2
={}&q_1^2+2q_1q_2+q_2^2+2q_1q_3\\
&+2q_1q_4+2q_2q_3+O(\|d\|^6).
\end{aligned}
\]
The preceding bound lets us replace an expectation restricted to
$\mathcal A$ by an expectation over the whole sample space, with
an $O(n^{-3})$ error for each displayed polynomial term. The two
degree-five terms, $2q_1q_4$ and $2q_2q_3$, have expectations
$O(n^{-3})$ by the centering argument above. The Taylor remainder
also has expected absolute value $O(n^{-3})$ on $\mathcal A$.
Outside $\mathcal A$, the capped loss contributes at most
$\Pr(\mathcal A^c)=O(n^{-3})$. Consequently,
\[
\overline{\operatorname{MSE}}_{\CRT,r}(\htau_j)
=\EEE[q_1^2+2q_1q_2+q_2^2+2q_1q_3]+O(n^{-3}).
\]

Next consider the squared error of the cubic approximation,
$(\mathcal P_r-\tau)^2=(q_1+q_2+q_3)^2$. Its terms of degrees
two through four are exactly the four terms inside the expectation
above. Its remaining terms are $2q_2q_3$ and $q_3^2$. The first
has expectation $O(n^{-3})$ by the signed fifth-moment bound;
the second has expectation $O(n^{-3})$ by the sixth-moment bound.
Thus
\[
\MSEexp_{\CRT,r}(\htau_j)
=\EEE[q_1^2+2q_1q_2+q_2^2+2q_1q_3]+O(n^{-3}).
\]
proving \eqref{eq:capped_transfer}.

\emph{Compare the clipped biases.}
Clipping is inactive on $\mathcal A$. Then the difference
between $h_{j,r}-\tau$ and $\mathcal P_r-\tau$ is $O(\|d\|^4)$,
whose expected absolute value is $O(n^{-2})$. Outside $\mathcal A$,
the clipped error is bounded by 1, and the expectation of each
term of the cubic polynomial is $O(n^{-3})$ after restriction to
that event. These bounds prove \eqref{eq:clipped_bias_transfer}.

\emph{Bernoulli assignment.}
For independent Bernoulli$(1/2)$ assignment, use the unit-level
raw-moment vectors
\[
\mathcal V_i=(W_i,W_iU_i^{(1)},(1-W_i)U_i^{(0)}).
\]
These vectors are IID and have finite sixth moments. At their
population mean, each arm has probability $1/2$ and the population
Gram matrices are invertible. The estimators are therefore smooth
functions of these raw moments near that mean. The same argument,
with $d=\bar{\mathcal V}-\EEE\mathcal V$ and cubic polynomial
$\tilde\tau_{n,j}$, proves both Bernoulli statements.
\end{proof}

\begin{lemm}[Conditional expansion]
\label{lemm:conditional}
Under the sampling and moment conditions of Theorem
\ref{theo:deficiency}, let $\mathcal P_r$ be the third-order Taylor
polynomial of $h_{j,r}$ at $(\mu_1,\mu_0)$ and define
$\MSEexp_{\CRT,r}(\htau_j)=\EEE[(\mathcal P_r-\tau)^2]$.
For every $\eta\in(0,\frac12)$,
\begin{equation}
\label{eq:conditional_expansion}
\MSEexp_{\CRT,r}(\hat{\tau}_{j})
=\frac{V_j(r)}{n}+\frac{c_j(r)}{n^2}+O(n^{-3}),
\qquad j\in\{\REG,\IREG\},
\end{equation}
uniformly in $r=m/n\in[\frac12-\eta,\frac12+\eta]$, with $V_j$ and
$c_j$ infinitely differentiable in $r$. The same expansion holds
with $\overline{\operatorname{MSE}}_{\CRT,r}$ in place of
$\MSEexp_{\CRT,r}$.
\end{lemm}

\begin{proof}
Fix $j\in\{\REG,\IREG\}$ and
$I_\eta=[\frac12-\eta,\frac12+\eta]$, where $0<\eta<1/2$.
The arm sizes are $m=rn$ and $n-m=(1-r)n$. We will first expand the
squared cubic polynomial, then count the powers of $1/n$ in the
expectation of each term.

\emph{Expand the squared polynomial.}
Write $d_1=\bar U_1-\mu_1$ and $d_0=\bar U_0-\mu_0$.
Recall
\[
\mathcal P_r-\tau=q_1+q_2+q_3.
\]
Lemma \ref{lemm:two_sample} shows that its
coefficients are rational functions of $r$ without poles in $(0,1)$.
Expanding the square gives
\[
\begin{aligned}
\MSEexp_{\CRT,r}(\htau_j)
={}&\EEE[q_1^2]+2\EEE[q_1q_2]
+\EEE[q_2^2]+2\EEE[q_1q_3]\\
&+2\EEE[q_2q_3]+\EEE[q_3^2].
\end{aligned}
\]

\emph{Separate the two arms.}
After multiplying out these terms, we need mixed moments of the form
\begin{equation}
\label{eq:mixed_moment}
\EEE\Big[\prod_{u=1}^{a}d_{1,k_u}
           \prod_{u=1}^{b}d_{0,l_u}\Big],
\qquad 2\le a+b\le6.
\end{equation}
Each is multiplied by a product of two Taylor coefficients, which
is again rational in $r$ without poles in $(0,1)$. The sixth-moment
calculation in the proof of Lemma \ref{lemm:localization} gives
$\EEE\|U^{(w)}\|^6<\infty$. H\"older's inequality therefore ensures
that all the joint coordinate moments needed here are finite.

By Lemma \ref{lemm:two_sample}, the two arms are independent
conditionally on $W$. Their conditional laws depend only on the
fixed arm sizes, so $d_1$ and $d_0$ are also independent under
complete randomization. Hence \eqref{eq:mixed_moment} factors as
\[
\EEE\Big[\prod_{u=1}^{a}d_{1,k_u}\Big]
\;\EEE\Big[\prod_{u=1}^{b}d_{0,l_u}\Big].
\]
It remains to determine the powers of the sample size in a moment
from a single arm.

\emph{Compute moments within one arm.}
Consider the treated arm; the control calculation is identical, with
$U^{(0)}$ and $n-m$ replacing $U^{(1)}$ and $m$.
Let $\kappa_{k_1\cdots k_s}$ be the joint cumulant of
$(U^{(1)}_{k_1},\ldots,U^{(1)}_{k_s})$. Cumulants are useful here
because they add over independent observations. Thus
\[
\begin{aligned}
\operatorname{cum}(d_{1,k_1},\ldots,d_{1,k_s})
&=m^{-s}\sum_{i:W_i=1}
  \operatorname{cum}(U^{(1)}_{i,k_1},\ldots,U^{(1)}_{i,k_s})\\
&=m^{1-s}\kappa_{k_1\cdots k_s},
\qquad 2\le s\le6.
\end{aligned}
\]
The first-order cumulants are zero, since $d_1$ is centered.

The moment--cumulant formula expresses a moment as a sum over
partitions $\pi$ of the positions $\{1,\ldots,a\}$:
\[
\EEE\Big[\prod_{u=1}^{a}d_{1,k_u}\Big]
=\sum_{\pi}\prod_{B\in\pi}
  \operatorname{cum}(d_{1,k_u}:u\in B).
\]
Each block $B$ groups together coordinates whose joint cumulant
appears in the product. If a block contains only one position, its
cumulant is the mean of a centered coordinate, so that entire
product is zero. This does not require independence among the
coordinates: their dependence is retained in the cumulants for
blocks of size two or more.

A surviving partition has at most $\lfloor a/2\rfloor$ blocks.
A block of size $|B|$ contributes a factor $m^{1-|B|}$, so the
partition as a whole contributes
\[
m^{\sum_{B\in\pi}(1-|B|)}=m^{-(a-|\pi|)}.
\]
There is at least one block when $a\ge2$. It follows that, for
$2\le a\le6$,
\begin{equation}
\label{eq:mean_moment}
\EEE\Big[\prod_{u=1}^{a}d_{1,k_u}\Big]
=\sum_{s=\lceil a/2\rceil}^{a-1}\frac{M^{(1)}_s}{m^s}.
\end{equation}
Here $M^{(1)}_s$ is a finite sum of products of cumulants of
$U^{(1)}$. It may depend on the coordinate indices, but not on $m$.
For $a=1$ the moment is zero; for $a=0$ the empty product is one.
For example,
\[
\begin{aligned}
\EEE[d_{1,k}d_{1,l}]
&=\frac{\kappa_{kl}}{m},\\
\EEE[d_{1,k}d_{1,l}d_{1,q}]
&=\frac{\kappa_{klq}}{m^2},\\
\EEE[d_{1,k}d_{1,l}d_{1,q}d_{1,v}]
&=\frac{\kappa_{kl}\kappa_{qv}
       +\kappa_{kq}\kappa_{lv}
       +\kappa_{kv}\kappa_{lq}}{m^2}
  +\frac{\kappa_{klqv}}{m^3}.
\end{aligned}
\]
These identities are exact for every $m$.

\emph{Collect powers of $1/n$.}
Apply \eqref{eq:mean_moment} to both arms and substitute $m=rn$
and $n-m=(1-r)n$. A mixed moment in \eqref{eq:mixed_moment}
vanishes if $a=1$ or $b=1$. Otherwise it is a finite sum of terms
\[
\frac{M^{(1)}_sM^{(0)}_t}{r^s(1-r)^t}\,n^{-(s+t)},
\qquad
s\ge\lceil a/2\rceil,\quad
t\ge\lceil b/2\rceil.
\]
When $a=0$, take $s=0$ and $M^{(1)}_0=1$; when $b=0$, take
$t=0$ and $M^{(0)}_0=1$. Thus the power of $1/n$ is at least
\[
\lceil a/2\rceil+\lceil b/2\rceil
\ge\lceil(a+b)/2\rceil,
\]
and at most five, since the total degree is at most six.

Combining these identities with the squared polynomial gives the
exact finite expansion
\[
\MSEexp_{\CRT,r}(\htau_j)
=\sum_{\ell=1}^{5}\frac{e_\ell(r)}{n^\ell}.
\]
Each $e_\ell$ is a finite sum of products of Taylor coefficients,
cumulants of $U^{(1)}$ and $U^{(0)}$, and powers of $r^{-1}$ and
$(1-r)^{-1}$. It is therefore rational in $r$ without poles in
$(0,1)$, and hence infinitely differentiable there and bounded on
$I_\eta$. This is an exact polynomial in $1/n$, before any terms
are put into a remainder.

The degree count also tells us which terms determine the first two
coefficients. The degree-two term contributes only to $n^{-1}$.
Indeed, if $g_w$ is the gradient of $h_{j,r}$ with respect to
$\bar U_w$ at $(\mu_1,\mu_0)$, then
\[
\begin{aligned}
\EEE[q_1^2]
&=\frac{g_1^\top\Var{U^{(1)}}g_1}{m}
  +\frac{g_0^\top\Var{U^{(0)}}g_0}{n-m}\\
&=\frac1n\left\{
  \frac{g_1^\top\Var{U^{(1)}}g_1}{r}
  +\frac{g_0^\top\Var{U^{(0)}}g_0}{1-r}\right\}.
\end{aligned}
\]
The degree-three term $2\EEE[q_1q_2]$ contributes only to
$n^{-2}$, through third cumulants. The degree-four terms
$\EEE[q_2^2]+2\EEE[q_1q_3]$ contribute to $n^{-2}$ and $n^{-3}$.
The degree-five and degree-six terms contribute only to $n^{-3}$
and smaller orders. Consequently,
\[
V_j(r)=e_1(r),\qquad c_j(r)=e_2(r),
\]
where $c_j(r)$ collects the $n^{-2}$ parts of $2\EEE[q_1q_2]$,
$\EEE[q_2^2]$, and $2\EEE[q_1q_3]$.

\emph{Bound the remainder uniformly.}
For every $n\ge1$ and admissible $r\in I_\eta$,
\[
\begin{aligned}
\left|\MSEexp_{\CRT,r}(\htau_j)
      -\frac{V_j(r)}n-\frac{c_j(r)}{n^2}\right|
&=\left|\sum_{\ell=3}^{5}\frac{e_\ell(r)}{n^\ell}\right|\\
&\le\frac1{n^3}
  \sum_{\ell=3}^{5}\sup_{r'\in I_\eta}|e_\ell(r')|.
\end{aligned}
\]
The sum on the right is finite because the coefficients are
continuous on the compact interval $I_\eta$, which stays away from
$0$ and $1$. This proves the uniform $O(n^{-3})$ remainder in
\eqref{eq:conditional_expansion}. Lemma \ref{lemm:localization}
then gives the same expansion for the capped risk.
\end{proof}

\subsection{From Bernoulli assignment to balanced complete randomization}

Bernoulli assignment makes the unit-level moment calculation convenient, but it also makes the treatment fraction random. The next lemma allows us to link the two trial design MSEs.

\begin{lemm}[Bernoulli and conditional expansions]
\label{lemm:bern_conditional}
Under the sampling and moment conditions of Theorem
\ref{theo:deficiency}, with
$W_i\stackrel{\mathrm{iid}}{\sim}\mathrm{Bernoulli}(1/2)$
independent of the units, let $\bar W=n^{-1}\sum_iW_i$.
For every $\eta\in(0,1/2)$,
\[
\overline{\operatorname{MSE}}_\Bern(\htau_j)
=\EEE\Big[
\overline{\operatorname{MSE}}_{\CRT,\bar W}(\htau_j)
\mathbf 1\{|\bar W-\tfrac12|\le\eta\}\Big]
+O(e^{-2n\eta^2}).
\]
Consequently,
\[
\MSEexp_\Bern(\htau_j)
=\EEE\Big[\MSEexp_{\CRT,\bar W}(\htau_j)
\mathbf 1\{|\bar W-\tfrac12|\le\eta\}\Big]+O(n^{-3}).
\]
\end{lemm}

\begin{proof}
Conditional on $n\bar W=m$, Bernoulli assignment is uniform over
assignments with $m$ treated units, exactly as in a CRT. The law of
total expectation therefore identifies the conditional capped risk
exactly. The omitted event contributes at most
\[
\Pr(|\bar W-1/2|>\eta)\le2e^{-2n\eta^2}.
\]
Applying Lemma \ref{lemm:localization} to the Bernoulli risk and
uniformly to the conditional risks proves the polynomial identity
with $O(n^{-3})$ error.
\end{proof}

The proof of Theorem \ref{theo:deficiency} combines these lemmas, a
reduction from the CRT to the Bernoulli design, and the Bernoulli
calculation.

\begin{theo}
\label{theo:crt_reduction}
Suppose the Bernoulli design has second-order difference
\begin{equation}
\label{eq:bern_second_order}
\MSEexp_\Bern(\hat{\tau}_\IREG)-\MSEexp_\Bern(\hat{\tau}_\REG)
=
\frac{D_\Bern}{n^2}+O(n^{-3})
\end{equation}
at $\pi=1/2$, and let the sampling and moment conditions of Theorem \ref{theo:deficiency} hold. Then, under
complete randomization with $n_1=\lfloor n/2\rfloor$ treated units,
\begin{equation}
\label{eq:crt_second_order}
\MSEexp_\CRT(\hat{\tau}_\IREG)-\MSEexp_\CRT(\hat{\tau}_\REG)
=
\frac{D_\Bern+4G}{n^2}+O(n^{-3}).
\end{equation}
The same identity holds for the capped MSE.
\end{theo}

\begin{proof}
\emph{Step 1: average the first-order variance difference.}
Let $n_1=\sum_{i=1}^n W_i$ under the Bernoulli design and write
$r=n_1/n$. Conditional on $n_1=m$, the treatment assignment has the same distribution as a CRT with treatment fraction $r$. Fix
$\eta\in(0,\frac12)$ and write
$\mathcal A_\eta=\{|r-1/2|\le\eta\}$. All expectations of
functions containing $r^{-1}$ or $(1-r)^{-1}$ below are restricted to
$\mathcal A_\eta$; their products with the indicator are defined
as zero off that event. By Lemma \ref{lemm:bern_conditional}, $\MSEexp_\Bern$ is, up to $O(n^{-3})$, the average over $n_1$ of
$\MSEexp_{\CRT,r}$ restricted to $|r-\frac12|\le\eta$, and by Lemma \ref{lemm:conditional} the latter admits the expansion \eqref{eq:conditional_expansion} uniformly on that set, with $V_j$ and $c_j$ infinitely differentiable. At treatment fraction $r$ the first-order variance difference is, by Proposition
\ref{prop:first_order},
\[
V_{\IREG}(r)-V_{\REG}(r)
=
-\frac{(1-2r)^2}{r(1-r)}G,
\]
and
\[
\frac{(1-2r)^2}{r(1-r)}
=
16u^2+O(u^4),
\qquad
u=r-\tfrac12.
\]
Under Bernoulli randomization with $\pi=1/2$,
$\EEE[u^2]=1/(4n)$ and $\EEE[u^4]=O(n^{-2})$.
Restricting these bounded moments to $\mathcal A_\eta$ changes
them by an exponentially small amount. The event is symmetric under
$u\mapsto-u$, so $\EEE[u\mathbf 1_{\mathcal A_\eta}]=0$.
It follows that
\[
\frac{1}{n}
\EEE\big[\{V_{\IREG}(r)-V_{\REG}(r)\}\mathbf 1_{\mathcal A_\eta}\big]
=
-\frac{4G}{n^2}+O(n^{-3}).
\]
\emph{Step 2: average the second-order coefficient.}
For the second-order coefficients, a second-order Taylor expansion of
$c_\IREG-c_\REG$ around $1/2$ together with the restricted moment identities just established gives
\[
\EEE[\{c_{\IREG}(r)-c_{\REG}(r)\}\mathbf 1_{\mathcal A_\eta}]
=
c_{\IREG}(1/2)-c_{\REG}(1/2)+O(n^{-1}).
\]
Averaging \eqref{eq:conditional_expansion} over $n_1$ therefore yields
\[
\MSEexp_{\Bern}(\hat{\tau}_\IREG)
-
\MSEexp_{\Bern}(\hat{\tau}_\REG)
=
\frac{
c_{\IREG}(1/2)-c_{\REG}(1/2)-4G
}{n^2}
+O(n^{-3}),
\]
so
\[
D_{\Bern}
=
c_{\IREG}(1/2)-c_{\REG}(1/2)-4G.
\]
\emph{Step 3: fix the arm sizes.}
Under balanced complete randomization, $r=1/2$ for even $n$, giving
\[
D_{\CRT}
=
c_{\IREG}(1/2)-c_{\REG}(1/2)
=
D_{\Bern}+4G.
\]
For even $n$, Lemma \ref{lemm:conditional} therefore gives directly
\[
\MSEexp_\CRT(\hat{\tau}_\IREG)-\MSEexp_\CRT(\hat{\tau}_\REG)
=
\frac{D_{\Bern}+4G}{n^2}+O(n^{-3}).
\]
For odd $n$, $r-1/2=O(n^{-1})$. The first-order variance difference is
quadratic in $r-1/2$, so its contribution to the MSE is $O(n^{-3})$;
the smoothness of $c_\IREG-c_\REG$ gives an $O(n^{-1})$ change in the
second-order coefficient, which is also $O(n^{-3})$ after division by
$n^2$. Thus the same expansion holds for odd $n$. Lemma \ref{lemm:localization} gives the capped-risk conclusion.
\end{proof}

\subsection{The second-order constant under Bernoulli assignment}

We now calculate $D_\Bern$. The tensor notation below is used only to
organize the contractions in the Taylor expansion; Lemma
\ref{lemm:matrix} converts the result to the matrix notation of the main
theorem. Here $F$ captures a fourth-moment contraction, $Q$ and
$\mathcal R$ capture third-moment contractions, and $A_P,B_P,C_P$
combine third moments with arm-specific residual structure.

\begin{theo}
\label{theo:bern_deficiency}
Let the sampling and moment conditions of Theorem \ref{theo:deficiency} hold, except that
$W_i\stackrel{\mathrm{iid}}{\sim}\text{Bernoulli}(1/2)$ independently of $(X_i,Y_i(0),Y_i(1))$. Define
\begin{align}
T_{3,ijk} &= \EEE[Z_iZ_jZ_k],
&
T_{4,ijkl} &= \EEE[Z_iZ_jZ_kZ_l],
&
q_i &= T_{3,ijk}\delta_{jk},
&
L_{ij} &= T_{3,iab}T_{3,jab},
\end{align}
write $S=q=\EEE[Z\|Z\|^2]$, and set $\mathcal E=\EEE[ZZ^\top\varepsilon]$ and $P=\mathcal E_W=\EEE[ZZ^\top(W-\tfrac12)\varepsilon]$. Here $P$ denotes this residual-moment matrix, not the superpopulation law. For the scalar contractions
\begin{align}
G&=\|\gamma\|^2,
&
F&=\gamma_i\gamma_jT_{4,ijkk},
&
Q&=(\gamma^\top q)^2,
&
\mathcal R&=\gamma^\top L\gamma,
\\
A_P&=\gamma_iT_{3,ijk}P_{jk},
&
B_P&=(\gamma^\top q)\tr(P),
&
C_P&=\gamma^\top Pq,
\end{align}
we have
\begin{equation}
\label{eq:bern_mse_diff}
\MSEexp_{\Bern}(\hat{\tau}_\IREG)
-
\MSEexp_{\Bern}(\hat{\tau}_\REG)
=
\frac{D_{\Bern}}{n^2}+O(n^{-3}),
\end{equation}
where
\begin{equation}
\label{eq:Dbern}
D_{\Bern}
=
-(2p+7)G
+F-Q-3\mathcal R
-16A_P-8B_P+8C_P
+8\tr(\mathcal E^2)
+8S^\top \EEE[Z\varepsilon^2].
\end{equation}
The same expansion holds for the capped MSE.
\end{theo}

\begin{proof}
\emph{Step 1: remove population main effects and whiten the covariates.}
On the neighborhood used in Lemma \ref{lemm:localization}, both REG and IREG contain the regressors $1$, $W$ and $X$ and their fits are uniquely defined. Subtracting any linear combination of these regressors from the outcome changes the fitted coefficient on $W$ by the corresponding amount while leaving the estimation error unchanged, so we may work with
\[
Y_i
=
W_i X_i^\top\{\beta^*_{(1)}-\beta^*_{(0)}\}
+\varepsilon_i.
\]
Whitening as in \eqref{eq:whitening} gives the canonical  representation
$Y_i=W_i\gamma^\top Z_i+\varepsilon_i$, whose population ATE is
zero. We write
\[
W_i=\tfrac12+\omega_i,
\qquad
\omega_i\in\{-\tfrac12,\tfrac12\},
\qquad
\EEE[\omega_i]=0,
\qquad
\omega_i^2=\tfrac14.
\]

\emph{Step 2: identify the three second-order contributions.}
Use the Bernoulli raw-moment vector $\mathcal V_i$ in
Definition \ref{def:msexp}, formed with the transformed outcome
and whitened covariates from Step 1, and the corresponding map
$g_j$. Let
\[
\Sigma^{\mathcal V}_{ij}
=
\Cov(\mathcal V_{1,i},\mathcal V_{1,j}),
\qquad
\mu^{\mathcal V}_{3,ijk}
=
\EEE[
(\mathcal V_{1,i}-\mu_i)
(\mathcal V_{1,j}-\mu_j)
(\mathcal V_{1,k}-\mu_k)
]
\]
be the covariance and third central-moment tensors of one unit's
raw-moment vector; the first subscript in $\mathcal V_{1,i}$ denotes
the unit and the second denotes its coordinate. In the next three
contractions all repeated indices range over raw-moment coordinates.
Expanding the square of \eqref{eq:hall_expansion} gives
\[
n\MSEexp_\Bern(\htau_j)
=V^*+\frac{\mathcal T_{1,j}+\mathcal T_{2,j}+\mathcal T_{3,j}}{n}
+O(n^{-2}),
\]
where, suppressing $j$ again,
\begin{equation}
\label{eq:hall_T1}
\mathcal T_1
=
a_iH_{jk}\mu^{\mathcal V}_{3,ijk},
\end{equation}
\begin{equation}
\label{eq:hall_T2_bern}
\mathcal T_2
=
\frac14 H_{ij}H_{kl}
\p{
\Sigma^{\mathcal V}_{ik}\Sigma^{\mathcal V}_{jl}
+
\Sigma^{\mathcal V}_{il}\Sigma^{\mathcal V}_{jk}
+
\Sigma^{\mathcal V}_{ij}\Sigma^{\mathcal V}_{kl}
},
\end{equation}
\begin{equation}
\label{eq:hall_T3_bern}
\mathcal T_3
=
\frac13 a_iA_{3,jkl}\p{
\Sigma^{\mathcal V}_{ij}\Sigma^{\mathcal V}_{kl}
+
\Sigma^{\mathcal V}_{ik}\Sigma^{\mathcal V}_{jl}
+
\Sigma^{\mathcal V}_{il}\Sigma^{\mathcal V}_{jk}
},
\end{equation}
\emph{Step 3: compute derivatives from the least-squares normal equations.}
For REG, order the regressors as $(1,Z,W)$; for IREG, use
$(1,Z,W,WZ)$. In these orders, the population Gram matrices and their
inverses in the whitened coordinates are
\[
A_R=
\begin{pmatrix}
1&0&1/2\\
0&I_p&0\\
1/2&0&1/2
\end{pmatrix},
\qquad
A_R^{-1}=
\begin{pmatrix}
2&0&-2\\
0&I_p&0\\
-2&0&4
\end{pmatrix},
\qquad
A_I=
\begin{pmatrix}
1&0&1/2&0\\
0&I_p&0&I_p/2\\
1/2&0&1/2&0\\
0&I_p/2&0&I_p/2
\end{pmatrix},
\]
\[
A_I^{-1}
=
\begin{pmatrix}
2&0&-2&0\\
0&2I_p&0&-2I_p\\
-2&0&4&0\\
0&-2I_p&0&4I_p
\end{pmatrix},
\qquad
\theta_R=(0,\gamma/2,0),
\qquad
\theta_I=(0,0,0,\gamma).
\]

For either regressor vector, let $A(v)$ and $b(v)$ be its Gram
matrix and regressor--outcome cross-moment vector as functions
of the raw moments $v$, and write $A=A(\mu)$ and $b=b(\mu)$.
These functions are affine in the chosen raw-moment coordinates,
so their derivatives of order two and higher vanish.
For a perturbation direction $h_s$ of the raw moments, write $dA_s$
and $db_s$ for the corresponding perturbations of $A$ and $b$.
Writing $\Omega=A^{-1}$ and $\theta=\Omega b$, differentiation of
$A\theta=b$ gives the exact derivative identities
\[
D\theta[h]
=
\Omega(db-dA\,\theta),
\]
\[
D^2\theta[h_1,h_2]
=
-\Omega dA_1\Omega(db_2-dA_2\theta)
-\Omega dA_2\Omega(db_1-dA_1\theta),
\]
\[
D^3\theta[h_1,h_2,h_3]
=
\sum_{\varsigma\in\mathfrak S_3}
\Omega dA_{\varsigma_1}\Omega dA_{\varsigma_2}
\Omega(db_{\varsigma_3}-dA_{\varsigma_3}\theta).
\]
Here $\mathfrak S_3$ is the set of all six permutations
$(\varsigma_1,\varsigma_2,\varsigma_3)$ of $(1,2,3)$;
$dA$ and $db$ in the first identity are the perturbations in direction $h$.

\emph{Step 4: account for the centering in IREG.}
REG's estimator is the treatment coefficient,
$g_\REG(v)=\theta_W(v)$. IREG's is
$g_\IREG(v)=\theta_W(v)+\bar z(v)^\top\theta_{WZ}(v)$, where
$\bar z(v)$ is the full-sample covariate mean, a linear coordinate of
$v$ that vanishes at $\mu$ in whitened coordinates, and $\theta_{WZ}=\gamma$ at $\mu$. The product rule gives
\[
a_\IREG[h]
=
D\theta_W[h]+D\bar z[h]^\top\gamma,
\]
\[
H_\IREG[h_1,h_2]
=
D^2\theta_W[h_1,h_2]
+
D\bar z[h_1]^\top D\theta_{WZ}[h_2]
+
D\bar z[h_2]^\top D\theta_{WZ}[h_1],
\]
\[
A_{3,\IREG}[h_1,h_2,h_3]
=
D^3\theta_W[h_1,h_2,h_3]
+
\sum_{\{i,j,k\}=\{1,2,3\}}
D\bar z[h_i]^\top D^2\theta_{WZ}[h_j,h_k],
\]
where the sum means the three terms with
$(i,j,k)=(1,2,3),(2,1,3),(3,1,2)$. The covariance and
third-moment blocks of $\mathcal V_i$ are products of moments of
$(Z_i,\omega_i,\varepsilon_i)$ under the canonical representation in
Step 1, with $W_i$ marginally Bernoulli$(1/2)$.

\emph{Step 5: evaluate the contractions and take IREG minus REG.}
Write $\Delta\mathcal T_k=\mathcal T_{k,\IREG}-\mathcal T_{k,\REG}$.
Substituting these derivative tensors and the covariance and third-moment blocks of $\mathcal V_i$ into \eqref{eq:hall_T1}--\eqref{eq:hall_T3_bern}, we reduce the resulting indexed expressions using tensor symmetries, dummy-index relabeling, and the identities implied by $\EEE[ZZ^\top]=I_p$ and $\omega_i^2=1/4$. This gives
\begin{align}
\Delta\mathcal T_1
&=
2G+2F+8\Xi_\gamma,
\label{eq:dT1}
\\
\Delta\mathcal T_2
&=
-3G-F-Q-\mathcal R-8\Xi_\gamma-8A_P-8B_P,
\label{eq:dT2}
\\
\Delta\mathcal T_3
&=
-(2p+6)G-2\mathcal R-8A_P+8C_P
+8\tr(\mathcal E^2)
+8S^\top \EEE[Z\varepsilon^2],
\label{eq:dT3}
\end{align}
where
\[
\Xi_\gamma
=
\EEE[(\gamma^\top Z)\|Z\|^2\omega\varepsilon]
\]
with $\omega=W-\frac12$. The accompanying verification script
enumerates the full Cartesian products of the expansion directions and
reduces the resulting indexed expressions using only tensor symmetries, dummy-index relabeling, and the identities implied by
$\EEE[ZZ^\top]=I_p$ and $\omega_i^2=1/4$, with the dimension $p$ kept
symbolic. The symbolic remainder in each of the three rows is
identically zero; the verification is implemented in
\texttt{verification/general\char`\_p\char`\_symbolic\char`\_certificate.py} in the replication repository and runs in under a minute.

\emph{Step 6: collect the terms.}
Adding the three rows gives the coefficient of $G$ as
$2-3-(2p+6)=-(2p+7)$ and the coefficient of $F$ as $2-1=1$.
The $\Xi_\gamma$ terms cancel. The remaining coefficients are $-1$
for $Q$, $-3$ for $\mathcal R$, and $(-16,-8,8)$ for
$(A_P,B_P,C_P)$, together with the two residual terms in
\eqref{eq:Dbern}. This proves the stated formula.

One contraction can be checked directly without the full tensor
calculation. The third covariance pairing in
\eqref{eq:hall_T2_bern} is
\[
\frac14(H_{ij}\Sigma^{\mathcal V}_{ij})^2.
\]
It is the square of the leading bias coefficient. By the bias
calculation in Lemma \ref{lemm:bias}, those coefficients at balanced
allocation are $b_I=-4\tr(P)$ and
$b_R=-4\tr(P)-\gamma^\top q$. They are also the Bernoulli
coefficients: the conditional bias coefficient is smooth in $r$,
and Lemma \ref{lemm:localization}, applied to clipped errors, changes the mean by only $O(n^{-2})$. Conditioning on the treatment count and using its symmetry then preserves the leading bias coefficient. Thus this pairing contributes
\[
b_I^2-b_R^2
=-(\gamma^\top q)^2
-8(\gamma^\top q)\tr(P)
=-Q-8B_P.
\]
These are exactly the two bias-squared terms in the second row.
The other pairings and the linear--cubic contraction account for the
remaining variance contributions.

The computational certificate checks the remaining algebra. Derived
by hand and displayed above: the expansion terms \eqref{eq:hall_T1}--\eqref{eq:hall_T3_bern}, the Gram matrices and coefficients at $\mu$, the derivative identities, the product-rule structure of the IREG centering functional, and the reduction rules. Verified symbolically by the accompanying implementation: the substitution of these objects into \eqref{eq:hall_T1}--\eqref{eq:hall_T3_bern} and the reduction of the resulting indexed sums to the rows \eqref{eq:dT1}--\eqref{eq:dT3}, with symbolic $p$ and zero remainder. The resulting constant is checked independently by Monte Carlo in Appendix
\ref{sec:numerical-verification}.
\end{proof}

\subsection{Recovering the matrix formula in the main theorem}

Adding the $4G$ correction of Theorem \ref{theo:crt_reduction} to
the Bernoulli constant in Theorem \ref{theo:bern_deficiency} changes
only the coefficient of $G$, from $-(2p+7)$ to $-(2p+3)$. Hence
\begin{equation}
\label{eq:unsimp}
D_p
=
-(2p+3)G
+F-Q-3\mathcal R
-16A_P-8B_P+8C_P
+8\tr(\mathcal E^2)
+8S^\top \EEE[Z\varepsilon^2],
\end{equation}
which Lemma \ref{lemm:matrix} rewrites in the matrix form of
Theorem \ref{theo:deficiency}. Together with Lemma \ref{lemm:localization}, this proves its capped-MSE identity.

\begin{lemm}[Equivalent matrix form]
\label{lemm:matrix}
With $K$, $S$, $S_\gamma$, $\mathcal E$ and $\mathcal E_W$ as in
Theorem \ref{theo:deficiency},
\[
G=\|\gamma\|^2,
\qquad
F=\gamma^\top K\gamma,
\qquad
Q=(\gamma^\top S)^2,
\qquad
\mathcal R=\|S_\gamma\|_F^2,
\]
\[
A_P=\tr(S_\gamma \mathcal E_W),
\qquad
B_P=(\gamma^\top S)\tr(\mathcal E_W),
\qquad
C_P=\gamma^\top \mathcal E_WS.
\]
\end{lemm}

\begin{proof}
First,
\[
F
=
\gamma_i\gamma_jT_{4,ijkl}\delta_{kl}
=
\gamma^\top \EEE[ZZ^\top\|Z\|^2]\gamma
=
\gamma^\top K\gamma,
\]
and
\[
S=\EEE[Z_iZ_jZ_j]_{i=1}^p=q
\]
gives
\[
Q=(\gamma^\top q)^2=(\gamma^\top S)^2.
\]
Next,
\[
(S_\gamma)_{ab}
=
\EEE[Z_aZ_b(Z^\top\gamma)]
=
T_{3,abk}\gamma_k,
\]
so
\[
\|S_\gamma\|_F^2
=
\sum_{a,b}(S_\gamma)_{ab}^2
=
\gamma_k\gamma_\ell T_{3,abk}T_{3,ab\ell}
=
\gamma_k\gamma_\ell L_{k\ell}
=
\mathcal R.
\]
Since $S_\gamma$ and $\mathcal E_W$ are symmetric,
\[
\tr(S_\gamma \mathcal E_W)
=
(S_\gamma)_{jk}(\mathcal E_W)_{jk}
=
\gamma_iT_{3,ijk}P_{jk}
=
A_P;
\]
finally,
\[
B_P
=
(\gamma^\top q)\tr(P)
=
(\gamma^\top S)\tr(\mathcal E_W)
\]
and
\[
C_P
=
\gamma^\top Pq
=
\gamma^\top \mathcal E_WS.
\]
\end{proof}

\subsection{The exact estimator difference and the leading biases}

The following lemma allows for variance-reduced
verification of Appendix \ref{sec:numerical-verification}.

\begin{lemm}[Exact difference identity]
\label{lemm:exact_diff}
Let $\bar X_w$ and
\[
\mathcal S_w
=
\sum_{i:W_i=w}
(X_i-\bar X_w)(X_i-\bar X_w)^\top
\]
denote the within-arm means and scatter matrices, let
\[
\hbeta_w
=
\mathcal S_w^{-1}
\sum_{i:W_i=w}
(X_i-\bar X_w)(Y_i-\bar Y_w),
\]
and write $d_X=\bar X_1-\bar X_0$. Then, whenever
$\mathcal S_0$, $\mathcal S_1$ and $\mathcal S_0+\mathcal S_1$ are
invertible and $n_1=n_0$,
\[
\htau_\REG
=
\bar Y_1-\bar Y_0
-
d_X^\top
(\mathcal S_1+\mathcal S_0)^{-1}
(\mathcal S_1\hbeta_1+\mathcal S_0\hbeta_0),
\]
\[
\htau_\IREG
=
\bar Y_1-\bar Y_0
-
d_X^\top
\frac{\hbeta_1+\hbeta_0}{2},
\]
and consequently
\[
\htau_\IREG-\htau_\REG
=
\tfrac12\,d_X^\top
(\mathcal S_1+\mathcal S_0)^{-1}
(\mathcal S_1-\mathcal S_0)
(\hbeta_1-\hbeta_0)
=O_P(n^{-1}).
\]
\end{lemm}

\begin{proof}
Both expressions for the estimators are the Frisch--Waugh--Lovell
representations of \eqref{eq:estimators}: the coefficient on $X$ in the REG regression is the pooled within-arm slope
$(\mathcal S_1+\mathcal S_0)^{-1}
(\mathcal S_1\hbeta_1+\mathcal S_0\hbeta_0)$, and IREG fits the two arms separately, so with $n_1=n_0$ the overall covariate mean satisfies
\[
\bar X_1-\bar X
=
d_X/2
=
-(\bar X_0-\bar X).
\]
Subtracting,
\[
\htau_\IREG-\htau_\REG
=
d_X^\top
\left[
(\mathcal S_1+\mathcal S_0)^{-1}
(\mathcal S_1\hbeta_1+\mathcal S_0\hbeta_0)
-
\tfrac12(\hbeta_1+\hbeta_0)
\right]
\]
\[
=
\tfrac12d_X^\top
(\mathcal S_1+\mathcal S_0)^{-1}
(\mathcal S_1-\mathcal S_0)
(\hbeta_1-\hbeta_0).
\]
The stochastic order follows from
\[
d_X=O_P(n^{-1/2}),
\qquad
(\mathcal S_1+\mathcal S_0)^{-1}
(\mathcal S_1-\mathcal S_0)=O_P(n^{-1/2}),
\qquad
\hbeta_1-\hbeta_0=O_P(1).
\]
\end{proof}

In whitened coordinates, the same Taylor expansion gives
\[
\htau_\IREG-\htau_\REG
=\tfrac14d_Z^\top(A_1-A_0)\gamma+O_P(n^{-3/2}),
\]
where $d_Z=\bar Z_1-\bar Z_0$ and $A_w$ is the arm-$w$ mean
of $ZZ^\top-I_p$; the arm sizes here remain $n_1=n_0=n/2$.

\begin{lemm}[Estimator bias]
\label{lemm:bias}
Under the sampling and moment conditions of Theorem
\ref{theo:deficiency} and balanced complete randomization, for
$e_{j}$ as in Remark \ref{rema:bias},
\[
\begin{aligned}
\EEE[e_{\IREG}]
&=-\frac{4\tr(\mathcal E_W)}{n}+O(n^{-2}),\\
\EEE[e_{\REG}]
&=-\frac{4\tr(\mathcal E_W)+\gamma^\top S}{n}+O(n^{-2}).
\end{aligned}
\]
\end{lemm}

\begin{proof}
By Lemma \ref{lemm:localization}, we may work with the cubic Taylor
polynomials. Their linear terms have mean zero, and their cubic
terms have expectations $O(n^{-2})$. Thus only the quadratic terms
determine the leading bias.

First take even $n$ and set $m=n/2$. In whitened coordinates, define
\[
d_Z=\bar Z_1-\bar Z_0,\qquad
t_w=\frac1m\sum_{i:W_i=w}Z_i\varepsilon_i(w),\qquad
A_w=\frac1m\sum_{i:W_i=w}(Z_iZ_i^\top-I_p).
\]
The first-order perturbations of the arm-specific slopes are $t_w$;
the pooled slope has perturbation
$(t_1+t_0)/2+(A_1-A_0)\gamma/4$.
Consequently, the quadratic estimator terms are
\[
q_{2,I}=-\frac12d_Z^\top(t_1+t_0),
\qquad
q_{2,R}=q_{2,I}-\frac14d_Z^\top(A_1-A_0)\gamma.
\]

Independence of the arms and centering eliminate the cross-arm
expectations. The remaining within-arm terms give
\[
\begin{aligned}
\EEE[d_Z^\top(t_1+t_0)]
&=\frac1m\EEE[\|Z\|^2\{\varepsilon(1)-\varepsilon(0)\}]
=\frac{4\tr(\mathcal E_W)}m,\\
\EEE[d_Z^\top(A_1-A_0)\gamma]
&=\frac2m\EEE[Z^\top(ZZ^\top-I_p)\gamma]
=\frac{2\gamma^\top S}m.
\end{aligned}
\]
Substituting these identities into the quadratic terms and using
$m=n/2$ proves both formulas.

For odd $n$, the treated fraction differs from $1/2$ by $O(n^{-1})$.
The leading bias is $n^{-1}$ times a smooth function of that
fraction, so this imbalance changes it by only $O(n^{-2})$.
\end{proof}

\section{Numerical illustration of Theorem \ref{theo:deficiency}}
\label{sec:numerical-verification}

The symbolic certificate verifies the algebra of Theorem
\ref{theo:bern_deficiency} exactly. Here we demonstrate empirically that the simulated difference in uncapped MSE lines up with our formula over a range of distributions.

\paragraph{Exact targets.}
We use data-generating processes with whitened covariates having
independent coordinates with known moments. The oracle residuals take
the form
\[
\varepsilon(w)=g_w(Z)+s_w(Z)\eta,
\]
where $\eta$ is independent of $Z$ with mean zero and unit variance,
$g_w$ is a polynomial orthogonal to $(1,Z)$, and $s_w^2$ is a polynomial. Each term in $D_p$ is therefore a finite combination of population moments and can be computed in exact rational arithmetic. 

We use six designs, chosen so that every nontrivial block of $D_p$ is
active in at least one design: Gaussian homoskedastic covariates;
centered-exponential covariates with heterogeneity; centered-exponential covariates with heteroskedasticity and $\gamma=0$; centered-exponential covariates with heterogeneity, arm-specific misspecification, and arm-specific heteroskedasticity; Gaussian covariates with arm-specific misspecification and heteroskedasticity; and a mixed skewed/Gaussian design. These designs are analytically tractable and their exact definitions are in \texttt{simulations/dgps.py} of the GitHub repository.

\paragraph{Variance reduction for numerical verification.}
We estimate the scaled difference in ordinary, uncapped MSE,
\[
D_n=n^2\left\{
\EEE[(\htau_\IREG-\tau)^2]
-\EEE[(\htau_\REG-\tau)^2]
\right\},
\]
which converges to $D_p$ under the additional mean-square remainder
conditions in Appendix \ref{sec:proofs}.

We checked both raw and variance-reduced estimates. The raw checks
gave compatible conclusions: at the largest sample size in each
of the 18 designs, all raw estimates lie within $2.3$ Monte Carlo
standard errors of $D_p$. Variance reduction improves precision
and hence power to distinguish nearby candidate constants.

On replication $b$, let $e_{j,b}=\htau_{j,b}-\tau$ and
$T_b=n^2(e_{\IREG,b}^2-e_{\REG,b}^2)$. Averaging $T_b$ estimates
$D_n$ directly. To reduce its noise, define
\[
d_1=\tfrac14d_Z^\top(A_1-A_0)\gamma,
\qquad
L=\gamma^\top\bar Z+\bar\varepsilon_1-\bar\varepsilon_0,
\]
where $d_Z=\bar Z_1-\bar Z_0$, $A_w$ averages $ZZ^\top-I_p$
within arm $w$, and $\bar\varepsilon_w$ averages its oracle
residuals. Lemma \ref{lemm:exact_diff} gives
$e_\IREG^2-e_\REG^2=2d_1L+O_P(n^{-2})$.
Thus $C_b=2n^2d_{1,b}L_b$ captures the leading fluctuation.
We report
\[
\hat D_p=\frac1B\sum_{b=1}^B(T_b-C_b)+\EEE[C_b].
\]
Subtracting the control variate and restoring its mean preserves
$\EEE[\hat D_p]=D_n$, provided the fitted estimators are defined
and have finite second moments. The control mean follows from
\begin{equation}
\label{eq:cv_identity}
\EEE[d_1L]
=\frac1{n^2}\EEE\Big[
(\|Z\|^2-1)(\gamma^\top Z)
\{\gamma^\top Z+\varepsilon(1)-\varepsilon(0)\}
\Big].
\end{equation}

The adjustment leaves $T_b-C_b=O_P(1)$; uniformly bounded variance
additionally requires $\sup_n\EEE[(T_b-C_b)^2]<\infty$.
Both estimates are reported by the code. In the Gaussian design
at $n=3200$, the adjustment reduces the Monte Carlo standard
error tenfold.

\paragraph{Results.}
Table \ref{tab:validation} reports the exact targets and the
control-variate estimates at $n\in\{200,800,3200\}$, using
$B=400{,}000$ replications for the two smaller sample sizes and
$B=120{,}000$ for $n=3200$.

The largest discrepancy occurs in the all-blocks design (fourth row).
There the estimate is $32.40$ at $n=200$, compared with the target
$44.52$, a difference of $12.1$ (about $22$ standard errors). The
difference falls to $5.7$ at $n=800$ and $2.55$ at $n=3200$. The
Bernoulli constant $D_\Bern=D_p-4G=41.32$ is within one
standard error of the $n=3200$ estimate, so that sample size alone does
not distinguish $D_\Bern$ from $D_p$ in this design. The successive
increments, $6.4$ and $3.1$, shrink by a
factor of about two as $n$ quadruples. A leading remainder proportional to $n^{-3}$ in MSE would shrink these increments by a factor of four. The observed increments do not establish a slower asymptotic rate: these sample sizes may be outside that regime, and the simulations use uncapped MSE. We do not extrapolate. The ablation in Table
\ref{tab:ablation} separates $D_p$ from $D_\Bern$ more clearly in other designs.

At $n=3200$, all six estimates are within $2.6$ standard errors of their targets. The two largest deviations are $z=-2.6$ in the all-blocks design and $z=+2.3$ in the $\gamma=0$ heteroskedastic design. In the latter, the estimate moves from $29.82$ at $n=200$ and $31.73$ at $n=800$ to $33.25$ at $n=3200$; with six designs, a deviation of this size is compatible with Monte Carlo error. In the all-blocks design the estimates are still moving toward the target, so that deviation also reflects finite-$n$ bias. The results are consistent with a remainder that is smaller than the $n^{-2}$ leading term, but three sample sizes are not enough to estimate its rate. 

\begin{table}[htbp]
\centering
\caption{Simulation comparison with Theorem \ref{theo:deficiency}. Targets are exact; estimates are the control-variate estimator with Monte Carlo standard errors in parentheses. $B=400{,}000$ for $n\in\{200,800\}$ and $B=120{,}000$ for $n=3200$. Parentheses after a design name list the blocks of $D_p$ that are active, in the notation of Theorem \ref{theo:bern_deficiency} and Lemma \ref{lemm:matrix} ($F=\gamma^\top K\gamma$, $Q=(\gamma^\top S)^2$, $\mathcal R=\|S_\gamma\|_F^2$, $E=\mathcal E$, $E_W=\mathcal E_W$); the mixed design has heterogeneity along its skewed coordinate and heteroskedastic residuals, so the covariate-shape blocks and $8S^\top\EEE[Z\varepsilon^2]$ are active.}
\label{tab:validation}
\small
\resizebox{\textwidth}{!}{%
\begin{tabular}{lrrrr}
\toprule
DGP (active blocks) & $D_p$ (exact) & $n=200$ & $n=800$ & $n=3200$ \\
\midrule
\input{tables/validation_body.tex} 
\bottomrule
\end{tabular}}
\end{table}

\paragraph{Power: what the verification would reject.}
Agreement with $D_p$ is more informative when the simulation can also
detect errors in the formula. We therefore compare the estimates with
several deliberately perturbed candidate constants. Table \ref{tab:ablation} reports $z$-statistics at $n=3200$ for the exact $D_p$ and several perturbed formulas: the Bernoulli constant $D_\Bern=D_p-4G$, which omits the CRT correction; the Gaussian special case $-(p+1)G$; and versions of $D_p$ with individual blocks removed.

The exact formula has $|z|\leq2.6$ in every design, whereas each
alternative has $|z|>3$ in at least one design, in some cases by tens of standard errors. The check has limited power when an alternative differs from $D_p$ by less than the residual finite-$n$ bias. In the all-blocks design, for example, $D_\Bern=D_p-3.2$ is not separated from $D_p$ at $n=3200$, which is why multiple designs are useful.

Without the control variate, the standard errors at $n=3200$ would be
$1.5$ to $3.8$, four to ten times larger in the heterogeneous designs. The Bernoulli constant would then not be separated from $D_p$ in any design ($|z|\leq3$ throughout), so the simulation would have little power to detect an omitted CRT correction.

\begin{table}[htbp]
\centering
\caption{Ablation: $z=(\hat D_p-D^{\mathrm{alt}})/\widehat{\mathrm{SE}}$ at $n=3200$ for each design and candidate constant. Values with $|z|>3$ indicate rejection of that candidate. Each ``drop'' column removes one block from $D_p$: the term $-3\mathcal R$, the term $-Q$, the residual terms $8\tr(\mathcal E^2)+8S^\top\EEE[Z\varepsilon^2]$ (``resid.''), or the three terms involving $\mathcal E_W$.}
\label{tab:ablation}
\small
\setlength{\tabcolsep}{4pt}
\resizebox{\textwidth}{!}{%
\begin{tabular}{lrrrrrrr}
\toprule
DGP & $D_p$ & $D_p-4G$ & $-(p+1)G$ & drop $-3\mathcal R$ & drop $-Q$ & drop resid.\ & drop $\mathcal E_W$ \\
\midrule
\input{tables/ablation_body.tex} 
\bottomrule
\end{tabular}}
\end{table}

\paragraph{Further covariate families and dimensions.}
The six designs above have $p\le3$ and independent coordinates. We also consider four covariate families at $p\in\{2,5,8\}$, including dependent coordinates (Table \ref{tab:families}).

For Gaussian covariates,
$D_p=-(p+1)G$. For Rademacher covariates,
$K=pI_p$ and $D_p=-(p+3)G$. For a symmetric scale mixture of normals,
where $2\%$ of units have six times the standard deviation, write
$s$ for the mixing scale ($s=6$ with probability $0.02$ and
$s=1$ otherwise). After standardizing the mixture,
\[
K=(p+2)\tilde\kappa I_p,
\qquad
\tilde\kappa=\frac{\EEE[s^4]}{\EEE[s^2]^2},
\]
and
\[
D_p
=
G\{(p+2)\tilde\kappa-(2p+3)\}.
\]
For Gaussian covariates with quadratic misspecification
$c(X^\top X-\tr\Sigma)$, where $c=0.3$ and
$\Sigma_{jk}=0.3^{|j-k|}$,
\[
\mathcal E=2c\Sigma,
\qquad
D_p=-(p+1)G+32c^2\tr(\Sigma^2).
\]

Every family uses homoskedastic $N(0,1)$ noise shared by the potential
outcomes, with $G=0.49$ ($1.96$ for the Rademacher family). Because both estimators are invariant to invertible linear transformations of the covariates, we simulate whitened covariates; dependence enters through $\gamma$ and, for the misspecified family, through $\tr(\Sigma^2)$.

The targets are exact. Recomputing them from the matrix form of Theorem \ref{theo:deficiency} using $10^6$ Monte Carlo draws reproduces them to within \gz{fam:matrixcheck:maxrel}\%, with the largest discrepancy in the scale mixture, where fourth moments are harder to estimate. Using exact targets avoids this additional Monte Carlo error.

With $n=\gz{fam:n}$ and $B=\gz{fam:B}$ replications per design, every
estimate is within \gz{fam:maxabsz} standard errors of its target, and
the control variate reduces the standard errors by factors of
\gz{fam:seratio:min} to \gz{fam:seratio:max} relative to the raw estimator
on the same replications.

Figure \ref{fig:validation-z} shows the $z$-statistics for all
$\gz{valall:n}$ checks: the designs of Table \ref{tab:validation} at
$n=3200$ and those of Table \ref{tab:families}. Of these,
$\gz{valall:within196}$ are within $\pm1.96$; the two exceptions are
among the deviations discussed above.

Figure \ref{fig:validation-rates} shows the estimates as a function of
$n$. For the four families at $p=5$, the standard errors remain roughly constant as $n$ grows. One of the twelve points below
$n=\gz{fam:n}$ lies outside $\pm1.96$ standard errors (the scale mixture at $n=1{,}000$, $z=\gz{fam:rate:heavytail1000:z}$), while the neighboring points at $n=500$ and $n=2{,}000$ are within
$\gz{fam:rate:heavytail:nbr}$ standard errors. With twenty-four family
checks, an isolated deviation of this size is compatible with Monte
Carlo variation.

\begin{table}[htbp]
\centering
\caption{Further simulation comparison with Theorem \ref{theo:deficiency}: four covariate families at $p\in\{2,5,8\}$, $n=\gz{fam:n}$, $B=\gz{fam:B}$. Targets are exact; estimates are the control-variate estimator with Monte Carlo standard errors in parentheses; the last column is the standard error of the raw estimator on the same replications.}
\label{tab:families}
\small
\begin{tabular}{lrrrrr}
\toprule
Family & $p$ & $D_p$ (exact) & Estimate & $z$ & Raw SE \\
\midrule
\input{tables/validation_families_body.tex} 
\bottomrule
\end{tabular}
\end{table}

\begin{figure}[htbp]
\centering
\includegraphics[width=\textwidth]{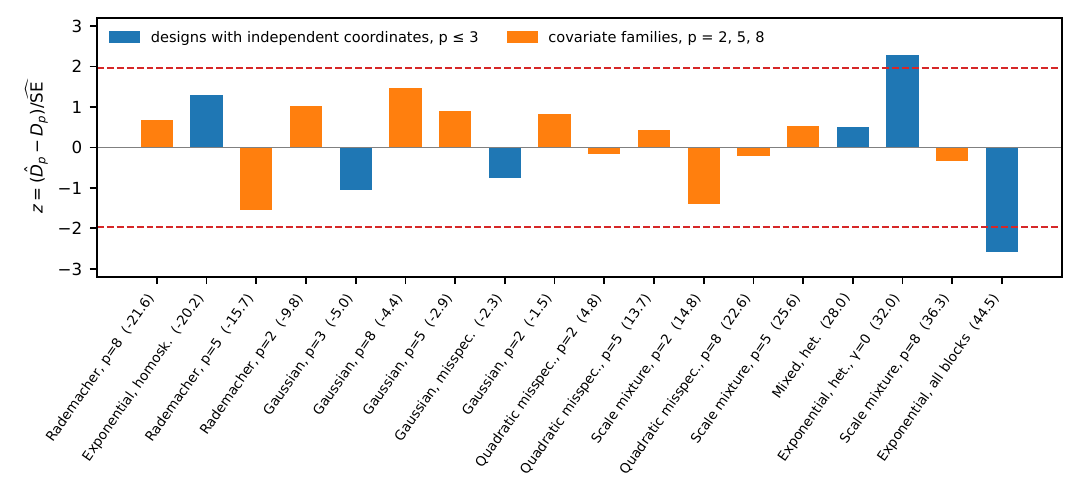}
\caption{$z$-statistics $(\hat D_p-D_p)/\widehat{\mathrm{SE}}$ of the control-variate estimates for the designs of Table \ref{tab:validation} ($n=3{,}200$) and Table \ref{tab:families} ($n=\gz{fam:n}$), sorted by $D_p$ (in parentheses). Dashed lines mark $\pm1.96$.}
\label{fig:validation-z}
\end{figure}

\begin{figure}[htbp]
\centering
\includegraphics[width=\textwidth]{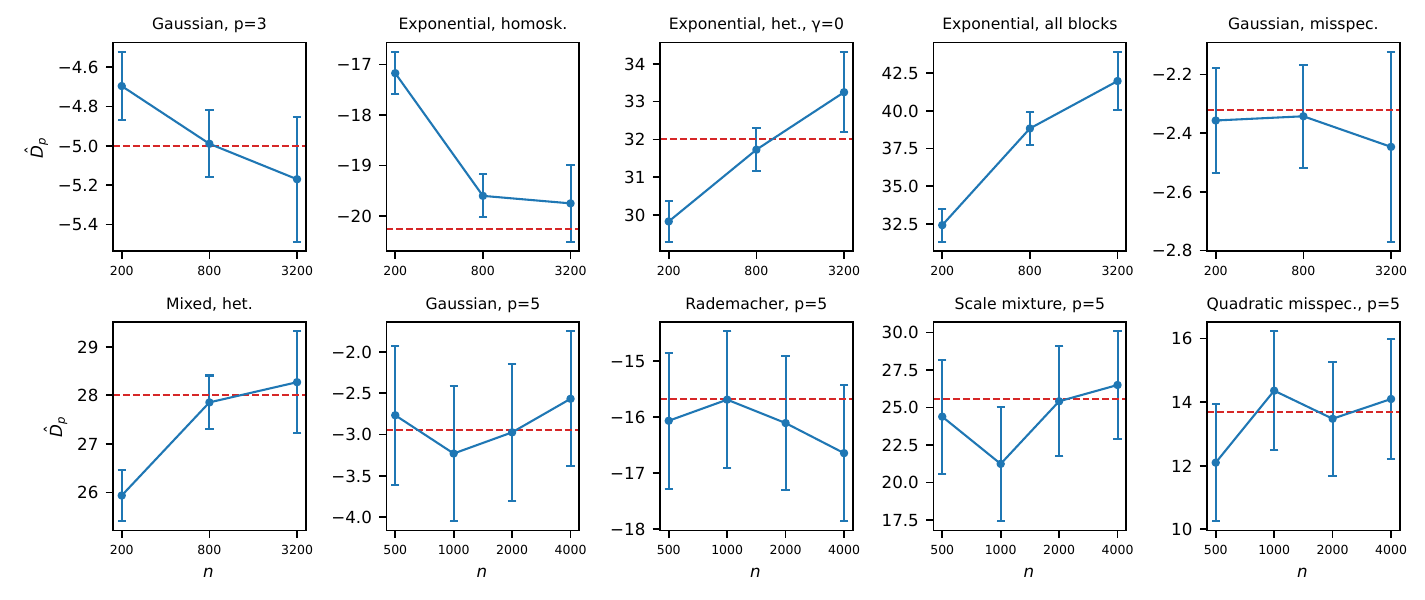}
\caption{Control-variate estimates of $D_p$ against $n$ (log scale) with $\pm1.96$ standard-error bars; the dashed line is the exact $D_p$. The first six panels are the designs of Table \ref{tab:validation} ($n=200$, $800$, $3{,}200$); the last four are the families of Table \ref{tab:families} at $p=5$ ($n=500$ to $\gz{fam:n}$).}
\label{fig:validation-rates}
\end{figure}

\paragraph{The CRT correction and the bias decomposition.}
We check Theorem \ref{theo:crt_reduction} separately by estimating the $n^{-2}$ constant under both balanced CRT and Bernoulli$(1/2)$ assignment on the same designs (\texttt{verification/crt\char`\_correction\char`\_certificate.py}).
In the Gaussian design, the estimated gap between the two designs is
$4.3\pm0.6$ at $n=200$ and $6.3\pm1.6$ at $n=800$, compared with the
predicted $4G=5$. In the exponential-covariate design, the corresponding gaps are $3.0\pm1.1$ and $8.6\pm3.0$. All four comparisons have $|z|<2$.

This check uses the raw statistic because the control variate is derived for the two-independent-samples representation of the CRT, so its power is limited. The ablation in Table \ref{tab:ablation} provides a sharper test: there, the Bernoulli constant is rejected by as many as thirty standard errors.

We also compare the leading coefficients in Remark \ref{rema:bias} with simulations of the ordinary, unclipped biases of both estimators
(\texttt{simulations/run\char`\_bias.py}). With $10^6$ replications, at $n=800$ all six estimated biases (REG and IREG in three designs) are within $0.8$ Monte Carlo standard errors of their leading-order
predictions. At $n=200$, the all-blocks design differs by about $11\%$
of the predicted REG bias ($z=9.3$) and by $5\%$ of the predicted IREG
bias ($z=3.5$); by $n=800$, both differences are negligible at the
reported Monte Carlo precision ($z=0.4$ and $z=-0.5$), consistent with the leading bias coefficients. Lemma \ref{lemm:bias} proves the $O(n^{-2})$ remainder for clipped errors; the ordinary-bias simulation does not establish that bound for unclipped errors.

\section{Range of validity of the second-order approximation}
\label{sec:validity}

Theorem \ref{theo:deficiency} is an asymptotic result. A separate question is when the approximation
\[
1+\Delta/n
\]
is accurate for the finite-sample MSE ratio. We examine this in the
scalar lognormal family given in the introduction example, with $\sigma(X)=1+X$, now with a treatment-effect slope of $0.4$ fixed while $X\sim\mathrm{LogNormal}(0,s^2)$ varies. The ratio compares IREG with REG, so values above one favor REG.
As $s$ increases, the covariate tails become heavier and the higher
moments entering $\Delta$ grow rapidly. 

\begin{table}[htbp]
\centering
\caption{Simulated MSE ratio $\widehat\MSE_\IREG/\widehat\MSE_\REG$ versus the second-order prediction $1+\Delta/n$, for $X\sim\mathrm{LogNormal}(0,s^2)$, $\sigma(X)=1+X$, treatment-effect slope $0.4$. Monte Carlo standard errors of the simulated ratios are given in parentheses in units of the last reported digit. The introductory example is this family at $s=1.1$ with no treatment-effect heterogeneity; the heavy-tailed example ($s=1.8$, $\sigma(X)=1+1.5X$, slope $0.3$) is a separate design.}
\label{tab:validity}
\small
\begin{tabular}{rrrrrrrrr}
\toprule
 & & & \multicolumn{2}{c}{$n=500$} & \multicolumn{2}{c}{$n=2000$} & \multicolumn{2}{c}{$n=8000$} \\
$s$ & skew $Z$ & $\Delta$ & sim. & pred. & sim. & pred. & sim. & pred. \\
\midrule
\input{tables/validity_body.tex} 
\bottomrule
\end{tabular}
\end{table}

For $s=0.8$ (skewness $3.7$), the simulated excess MSE is
$\gz{valid:s0.8:n500:frac}\%$ of the predicted excess $\Delta/n$ at
$n=500$ and $\gz{valid:s0.8:n2000:frac}\%$ at $n=2{,}000$; for $s=0.6$, the corresponding figures are $\gz{valid:s0.6:n500:frac}\%$ and $\gz{valid:s0.6:n2000:frac}\%$. Generally, the approximation accuracy deteriorates as the tails become heavier.

The heavy-tailed example illustrates the
same phenomenon more sharply. With $s=1.8$, $\sigma(X)=1+1.5X$, and
treatment-effect slope $0.3$, we have $\Delta\approx 3.9\times10^4$. The second-order prediction is then $78.5$, $20.4$, and $5.8$ at $n=500$, $2{,}000$, and $8{,}000$, compared with simulated MSE ratios of $1.68$, $1.34$, and $1.19$. The sign of the comparison is correct, but $1+\Delta/n$ is no longer a useful numerical approximation to the MSE ratio. 

\paragraph{Tail sensitivity.}
Deficiency depends on population moments that can be dominated by
extreme values rarely observed at the simulated sample sizes.
Using closed-form truncated lognormal moments, Table \ref{tab:tail}
recomputes $1+\Delta/n$ at $n=500$. Truncation at the $1-10^{-4}$
quantile reduces $\Delta$ by one to three orders of magnitude,
yet even truncation at $1-10^{-8}$ leaves predictions several
times larger than the simulated ratios.

The discrepancies begin at first order. For the same family with
$\gamma=0$, simulated versus first-order standard deviations
$\sqrt{V^*/n}$ are $\gz{samp:s0.6:sdR}$ versus
$\gz{samp:s0.6:sdtheory}$ at $s=0.6$,
$\gz{samp:s1.5:sdR}$ versus $\gz{samp:s1.5:sdtheory}$ at $s=1.5$,
and $\gz{samp:s2.0:sdR}$ versus $\gz{samp:s2.0:sdtheory}$ at $s=2$.
Thus the expansion loses quantitative accuracy under heavy tails.
Rare draws also increase Monte Carlo uncertainty, although the
second-order calculation correctly predicts the sign of the
REG--IREG comparison in every design.

\begin{table}[htbp]
\centering
\caption{Second-order prediction $1+\Delta/n$ at $n=500$ for the full lognormal law and for its truncation at the $q$ quantile, against the simulated MSE ratio (Monte Carlo standard error in parentheses, in units of the last digit). The last row is the family of Table \ref{tab:validity} at $s=2$ without treatment-effect heterogeneity ($\gamma=0$).}
\label{tab:tail}
\small
\setlength{\tabcolsep}{4pt}
\begin{tabular}{lccccc}
\toprule
& Simulated & \multicolumn{4}{c}{$1+\Delta/n$} \\
Design & ratio & full law & $q=1-10^{-4}$ & $q=1-10^{-6}$ & $q=1-10^{-8}$ \\
\midrule
\input{tables/tail_body.tex} 
\bottomrule
\end{tabular}
\end{table}

\section{Feasible Inference DGPs}
\label{sec:designs}

This section specifies the data-generating processes used for the
feasible inference simulations. The designs of Tables \ref{tab:battery} and \ref{tab:plugin-cal} all have $n=500$, balanced complete randomization and \gz{hte:reps} replications, and every interval is computed on the same simulated datasets. The scalar lognormal designs extend the family of the introduction:
\[
\begin{gathered}
p=1,\qquad
X_i\sim\mathrm{LogNormal}(0,s^2),\qquad
u_i\sim N(0,1),\\
Y_i(0)=0.9X_i+\sigma(X_i)u_i,\qquad
Y_i(1)=Y_i(0)+1+gX_i,
\end{gathered}
\]
with $\sigma(X)=1+X$ except in the $s=1.8$ design, where
$\sigma(X)=1+1.5X$, and
\[
(s,g)\in
\{(1.2,0.3),(1.5,0.1),(1.5,0.5),(1.8,0.3),
(2,0.1),(1.5,1),(1.5,3)\}.
\]
The remaining designs have $p=5$, potential outcomes that share their noise, and
\[
Y_i(0)=m(X_i)+\sigma(X_i)u_i,
\qquad
Y_i(1)=Y_i(0)+1+\gamma^*X_{i1},
\qquad
u_i\sim N(0,1),
\]
with $m(x)=0.9\sum_jx_j$ unless stated. They differ in the covariate law,
the noise scale and the slope $\gamma^*$:
\begin{itemize}
\item \emph{Gaussian}: $X_i\sim N(0,I_5)$, $\sigma\equiv1$ and
$\gamma^*=0.3$; the large-heterogeneity Gaussian row instead has
$X_i\sim N(0,\Sigma)$ with $\Sigma_{jk}=0.5^{|j-k|}$ and $\gamma^*=3$.
\item \emph{Misspecified, $G=0$}: $X_i\sim N(0,\Sigma)$ with
$\Sigma_{jk}=0.5^{|j-k|}$, $\sigma\equiv1$, $\gamma^*=0$, and
$m(x)=0.9\sum_jx_j+0.20(x_1^2-1)+0.15x_1x_2$, so that the working model is
misspecified in both arms.
\item \emph{Skewed, het., $G=0$}: independent coordinates, the first a
centered $\mathrm{LogNormal}(0,1.1^2)$ and the others standard normal, mixed
by the Cholesky factor of the matrix $0.35^{|j-k|}$ and then standardized to
unit variances; $\sigma(X)=1+0.5|X_{i1}|+0.25\|X_i\|$ and $\gamma^*=0$.
\item \emph{Heavy-tailed, het.}: the same construction with
$\mathrm{LogNormal}(0,1.5^2)$ and
$\sigma(X)=1+0.75|X_{i1}|+0.25\|X_i\|$; $\gamma^*=0.5$ in the
small-heterogeneity row and $\gamma^*=3$ in the large-heterogeneity row.
\item \emph{Rare extreme values}: $X_{i1}=\sqrt{(1-\rho)/\rho}$ with probability $\rho=0.02$ and $-\sqrt{\rho/(1-\rho)}$ otherwise, $X_{i2},\ldots,X_{i5}\sim N(0,1)$,
$\sigma\equiv1$ and $\gamma^*=1$. With $250$ units per arm, an arm contains no rare unit with probability about $0.6\%$, in which case the interacted design is rank deficient and IREG gives no interval (\gz{jk:rare:rankdef}\% of replications), and exactly one rare unit, which then has leverage one in its arm's fit, with probability about $3\%$.
\end{itemize}

\end{document}

%% file: tables/gamma0_values.tex
\expandafter\def\csname gz:all:REG_HC1:max\endcsname{95.4}
\expandafter\def\csname gz:all:REG_HC1:min\endcsname{94.7}
\expandafter\def\csname gz:analytic:cov:0.01\endcsname{94.9}
\expandafter\def\csname gz:analytic:cov:0.1\endcsname{93.7}
\expandafter\def\csname gz:bias:gauss:maxrelerr\endcsname{0.8}
\expandafter\def\csname gz:bias:gauss:ndesigns\endcsname{5}
\expandafter\def\csname gz:corr:heavy_tail_example:IREG_HC1:0.80\endcsname{67.2}
\expandafter\def\csname gz:corr:heavy_tail_example:IREG_HC1:0.90\endcsname{78.5}
\expandafter\def\csname gz:corr:heavy_tail_example:IREG_HC1:0.95\endcsname{85.6}
\expandafter\def\csname gz:corr:heavy_tail_example:IREG_HC1:0.98\endcsname{91.2}
\expandafter\def\csname gz:corr:heavy_tail_example:IREG_HC1:0.99\endcsname{93.7}
\expandafter\def\csname gz:corr:heavy_tail_example:IREG_HC1c:0.80\endcsname{75.3}
\expandafter\def\csname gz:corr:heavy_tail_example:IREG_HC1c:0.90\endcsname{88.1}
\expandafter\def\csname gz:corr:heavy_tail_example:IREG_HC1c:0.95\endcsname{94.4}
\expandafter\def\csname gz:corr:heavy_tail_example:IREG_HC1c:0.98\endcsname{97.9}
\expandafter\def\csname gz:corr:heavy_tail_example:IREG_HC1c:0.99\endcsname{99.1}
\expandafter\def\csname gz:corr:heavy_tail_example:IREG_HC1o:0.80\endcsname{68.2}
\expandafter\def\csname gz:corr:heavy_tail_example:IREG_HC1o:0.90\endcsname{79.5}
\expandafter\def\csname gz:corr:heavy_tail_example:IREG_HC1o:0.95\endcsname{86.5}
\expandafter\def\csname gz:corr:heavy_tail_example:IREG_HC1o:0.98\endcsname{91.8}
\expandafter\def\csname gz:corr:heavy_tail_example:IREG_HC1o:0.99\endcsname{94.2}
\expandafter\def\csname gz:corr:heavy_tail_example:IREG_HC3:0.80\endcsname{79.8}
\expandafter\def\csname gz:corr:heavy_tail_example:IREG_HC3:0.90\endcsname{89.3}
\expandafter\def\csname gz:corr:heavy_tail_example:IREG_HC3:0.95\endcsname{94.0}
\expandafter\def\csname gz:corr:heavy_tail_example:IREG_HC3:0.98\endcsname{97.0}
\expandafter\def\csname gz:corr:heavy_tail_example:IREG_HC3:0.99\endcsname{98.2}
\expandafter\def\csname gz:corr:heavy_tail_example:IREG_HC3c:0.80\endcsname{84.9}
\expandafter\def\csname gz:corr:heavy_tail_example:IREG_HC3c:0.90\endcsname{93.6}
\expandafter\def\csname gz:corr:heavy_tail_example:IREG_HC3c:0.95\endcsname{97.3}
\expandafter\def\csname gz:corr:heavy_tail_example:IREG_HC3c:0.98\endcsname{99.1}
\expandafter\def\csname gz:corr:heavy_tail_example:IREG_HC3c:0.99\endcsname{99.6}
\expandafter\def\csname gz:corr:heavy_tail_example:R2\endcsname{0.009}
\expandafter\def\csname gz:corr:heavy_tail_example:REG_HC1:0.80\endcsname{76.2}
\expandafter\def\csname gz:corr:heavy_tail_example:REG_HC1:0.90\endcsname{88.3}
\expandafter\def\csname gz:corr:heavy_tail_example:REG_HC1:0.95\endcsname{94.4}
\expandafter\def\csname gz:corr:heavy_tail_example:REG_HC1:0.98\endcsname{97.8}
\expandafter\def\csname gz:corr:heavy_tail_example:REG_HC1:0.99\endcsname{99.0}
\expandafter\def\csname gz:corr:heavy_tail_example:ratio\endcsname{31}
\expandafter\def\csname gz:corr:heavy_tail_example:reps\endcsname{100{,}000}
\expandafter\def\csname gz:corr:heavy_tail_example:wr\endcsname{1.06}
\expandafter\def\csname gz:corr:heavy_tail_example:wr_cover\endcsname{1.04}
\expandafter\def\csname gz:corr:heavy_tail_example:wr_miss\endcsname{1.35}
\expandafter\def\csname gz:corr:large_hte_n500_p5:Delta\endcsname{4.2}
\expandafter\def\csname gz:corr:large_hte_n500_p5:IREG_HC1:0.80\endcsname{52.1}
\expandafter\def\csname gz:corr:large_hte_n500_p5:IREG_HC1:0.90\endcsname{63.9}
\expandafter\def\csname gz:corr:large_hte_n500_p5:IREG_HC1:0.95\endcsname{72.7}
\expandafter\def\csname gz:corr:large_hte_n500_p5:IREG_HC1:0.98\endcsname{80.7}
\expandafter\def\csname gz:corr:large_hte_n500_p5:IREG_HC1:0.99\endcsname{85.3}
\expandafter\def\csname gz:corr:large_hte_n500_p5:IREG_HC1c:0.80\endcsname{80.2}
\expandafter\def\csname gz:corr:large_hte_n500_p5:IREG_HC1c:0.90\endcsname{90.0}
\expandafter\def\csname gz:corr:large_hte_n500_p5:IREG_HC1c:0.95\endcsname{95.2}
\expandafter\def\csname gz:corr:large_hte_n500_p5:IREG_HC1c:0.98\endcsname{98.1}
\expandafter\def\csname gz:corr:large_hte_n500_p5:IREG_HC1c:0.99\endcsname{99.1}
\expandafter\def\csname gz:corr:large_hte_n500_p5:IREG_HC1o:0.80\endcsname{80.2}
\expandafter\def\csname gz:corr:large_hte_n500_p5:IREG_HC1o:0.90\endcsname{90.0}
\expandafter\def\csname gz:corr:large_hte_n500_p5:IREG_HC1o:0.95\endcsname{95.1}
\expandafter\def\csname gz:corr:large_hte_n500_p5:IREG_HC1o:0.98\endcsname{98.2}
\expandafter\def\csname gz:corr:large_hte_n500_p5:IREG_HC1o:0.99\endcsname{99.2}
\expandafter\def\csname gz:corr:large_hte_n500_p5:IREG_HC3:0.80\endcsname{52.7}
\expandafter\def\csname gz:corr:large_hte_n500_p5:IREG_HC3:0.90\endcsname{64.6}
\expandafter\def\csname gz:corr:large_hte_n500_p5:IREG_HC3:0.95\endcsname{73.3}
\expandafter\def\csname gz:corr:large_hte_n500_p5:IREG_HC3:0.98\endcsname{81.3}
\expandafter\def\csname gz:corr:large_hte_n500_p5:IREG_HC3:0.99\endcsname{85.7}
\expandafter\def\csname gz:corr:large_hte_n500_p5:IREG_HC3c:0.80\endcsname{80.4}
\expandafter\def\csname gz:corr:large_hte_n500_p5:IREG_HC3c:0.90\endcsname{90.1}
\expandafter\def\csname gz:corr:large_hte_n500_p5:IREG_HC3c:0.95\endcsname{95.3}
\expandafter\def\csname gz:corr:large_hte_n500_p5:IREG_HC3c:0.98\endcsname{98.2}
\expandafter\def\csname gz:corr:large_hte_n500_p5:IREG_HC3c:0.99\endcsname{99.1}
\expandafter\def\csname gz:corr:large_hte_n500_p5:R2\endcsname{0.692}
\expandafter\def\csname gz:corr:large_hte_n500_p5:REG_HC1:0.80\endcsname{80.2}
\expandafter\def\csname gz:corr:large_hte_n500_p5:REG_HC1:0.90\endcsname{90.2}
\expandafter\def\csname gz:corr:large_hte_n500_p5:REG_HC1:0.95\endcsname{95.2}
\expandafter\def\csname gz:corr:large_hte_n500_p5:REG_HC1:0.98\endcsname{98.2}
\expandafter\def\csname gz:corr:large_hte_n500_p5:REG_HC1:0.99\endcsname{99.1}
\expandafter\def\csname gz:corr:large_hte_n500_p5:ratio\endcsname{1.00}
\expandafter\def\csname gz:corr:large_hte_n500_p5:reps\endcsname{20{,}000}
\expandafter\def\csname gz:corr:large_hte_n500_p5:wr\endcsname{1.80}
\expandafter\def\csname gz:corr:large_hte_n500_p5:wr_cover\endcsname{1.80}
\expandafter\def\csname gz:corr:large_hte_n500_p5:wr_miss\endcsname{1.80}
\expandafter\def\csname gz:corr:lognormal_s2.0_prop:IREG_HC1:0.80\endcsname{64.4}
\expandafter\def\csname gz:corr:lognormal_s2.0_prop:IREG_HC1:0.90\endcsname{75.5}
\expandafter\def\csname gz:corr:lognormal_s2.0_prop:IREG_HC1:0.95\endcsname{83.0}
\expandafter\def\csname gz:corr:lognormal_s2.0_prop:IREG_HC1:0.98\endcsname{89.1}
\expandafter\def\csname gz:corr:lognormal_s2.0_prop:IREG_HC1:0.99\endcsname{91.9}
\expandafter\def\csname gz:corr:lognormal_s2.0_prop:IREG_HC1c:0.80\endcsname{74.6}
\expandafter\def\csname gz:corr:lognormal_s2.0_prop:IREG_HC1c:0.90\endcsname{88.3}
\expandafter\def\csname gz:corr:lognormal_s2.0_prop:IREG_HC1c:0.95\endcsname{95.1}
\expandafter\def\csname gz:corr:lognormal_s2.0_prop:IREG_HC1c:0.98\endcsname{98.6}
\expandafter\def\csname gz:corr:lognormal_s2.0_prop:IREG_HC1c:0.99\endcsname{99.5}
\expandafter\def\csname gz:corr:lognormal_s2.0_prop:IREG_HC1o:0.80\endcsname{64.4}
\expandafter\def\csname gz:corr:lognormal_s2.0_prop:IREG_HC1o:0.90\endcsname{75.5}
\expandafter\def\csname gz:corr:lognormal_s2.0_prop:IREG_HC1o:0.95\endcsname{83.0}
\expandafter\def\csname gz:corr:lognormal_s2.0_prop:IREG_HC1o:0.98\endcsname{89.1}
\expandafter\def\csname gz:corr:lognormal_s2.0_prop:IREG_HC1o:0.99\endcsname{91.9}
\expandafter\def\csname gz:corr:lognormal_s2.0_prop:IREG_HC3:0.80\endcsname{79.5}
\expandafter\def\csname gz:corr:lognormal_s2.0_prop:IREG_HC3:0.90\endcsname{89.2}
\expandafter\def\csname gz:corr:lognormal_s2.0_prop:IREG_HC3:0.95\endcsname{93.9}
\expandafter\def\csname gz:corr:lognormal_s2.0_prop:IREG_HC3:0.98\endcsname{96.9}
\expandafter\def\csname gz:corr:lognormal_s2.0_prop:IREG_HC3:0.99\endcsname{98.0}
\expandafter\def\csname gz:corr:lognormal_s2.0_prop:IREG_HC3c:0.80\endcsname{85.7}
\expandafter\def\csname gz:corr:lognormal_s2.0_prop:IREG_HC3c:0.90\endcsname{94.6}
\expandafter\def\csname gz:corr:lognormal_s2.0_prop:IREG_HC3c:0.95\endcsname{98.1}
\expandafter\def\csname gz:corr:lognormal_s2.0_prop:IREG_HC3c:0.98\endcsname{99.5}
\expandafter\def\csname gz:corr:lognormal_s2.0_prop:IREG_HC3c:0.99\endcsname{99.8}
\expandafter\def\csname gz:corr:lognormal_s2.0_prop:R2\endcsname{0.000}
\expandafter\def\csname gz:corr:lognormal_s2.0_prop:REG_HC1:0.80\endcsname{76.3}
\expandafter\def\csname gz:corr:lognormal_s2.0_prop:REG_HC1:0.90\endcsname{88.8}
\expandafter\def\csname gz:corr:lognormal_s2.0_prop:REG_HC1:0.95\endcsname{95.3}
\expandafter\def\csname gz:corr:lognormal_s2.0_prop:REG_HC1:0.98\endcsname{98.6}
\expandafter\def\csname gz:corr:lognormal_s2.0_prop:REG_HC1:0.99\endcsname{99.5}
\expandafter\def\csname gz:corr:lognormal_s2.0_prop:ratio\endcsname{--}
\expandafter\def\csname gz:corr:lognormal_s2.0_prop:reps\endcsname{20{,}000}
\expandafter\def\csname gz:corr:lognormal_s2.0_prop:wr\endcsname{1.07}
\expandafter\def\csname gz:corr:lognormal_s2.0_prop:wr_cover\endcsname{1.05}
\expandafter\def\csname gz:corr:lognormal_s2.0_prop:wr_miss\endcsname{1.48}
\expandafter\def\csname gz:corr:mcse:max\endcsname{0.4}
\expandafter\def\csname gz:def:s0.8\endcsname{12}
\expandafter\def\csname gz:def:s1.0\endcsname{41}
\expandafter\def\csname gz:def:s1.2\endcsname{166}
\expandafter\def\csname gz:defpred:s0.8:n2000\endcsname{0.6}
\expandafter\def\csname gz:defpred:s0.8:n500\endcsname{2.4}
\expandafter\def\csname gz:defpred:s0.8:n8000\endcsname{0.2}
\expandafter\def\csname gz:defpred:s1.0:n2000\endcsname{2.1}
\expandafter\def\csname gz:defpred:s1.0:n500\endcsname{8.3}
\expandafter\def\csname gz:defpred:s1.0:n8000\endcsname{0.5}
\expandafter\def\csname gz:defpred:s1.2:n2000\endcsname{8.3}
\expandafter\def\csname gz:defpred:s1.2:n500\endcsname{33.2}
\expandafter\def\csname gz:defpred:s1.2:n8000\endcsname{2.1}
\expandafter\def\csname gz:defsim:s0.8:n2000\endcsname{0.5}
\expandafter\def\csname gz:defsim:s0.8:n500\endcsname{1.8}
\expandafter\def\csname gz:defsim:s0.8:n8000\endcsname{0.1}
\expandafter\def\csname gz:defsim:s1.0:n2000\endcsname{1.4}
\expandafter\def\csname gz:defsim:s1.0:n500\endcsname{4.3}
\expandafter\def\csname gz:defsim:s1.0:n8000\endcsname{0.4}
\expandafter\def\csname gz:defsim:s1.2:n2000\endcsname{3.5}
\expandafter\def\csname gz:defsim:s1.2:n500\endcsname{9.2}
\expandafter\def\csname gz:defsim:s1.2:n8000\endcsname{1.3}
\expandafter\def\csname gz:diag:s0.6:seratio:IREG:median\endcsname{1.01}
\expandafter\def\csname gz:diag:s0.6:seratio:IREG:p90\endcsname{1.02}
\expandafter\def\csname gz:diag:s0.6:seratio:IREG:p95\endcsname{1.03}
\expandafter\def\csname gz:diag:s0.6:seratio:REG:median\endcsname{1.01}
\expandafter\def\csname gz:diag:s0.6:seratio:REG:p90\endcsname{1.01}
\expandafter\def\csname gz:diag:s0.6:seratio:REG:p95\endcsname{1.01}
\expandafter\def\csname gz:diag:s0.6:width:IREG_HC1\endcsname{0.82}
\expandafter\def\csname gz:diag:s0.6:width:IREG_HC3\endcsname{0.82}
\expandafter\def\csname gz:diag:s0.6:width:REG_HC1\endcsname{0.81}
\expandafter\def\csname gz:diag:s0.6:width:REG_HC3\endcsname{0.82}
\expandafter\def\csname gz:diag:s0.8:seratio:IREG:median\endcsname{1.02}
\expandafter\def\csname gz:diag:s0.8:seratio:IREG:p90\endcsname{1.05}
\expandafter\def\csname gz:diag:s0.8:seratio:IREG:p95\endcsname{1.07}
\expandafter\def\csname gz:diag:s0.8:seratio:REG:median\endcsname{1.01}
\expandafter\def\csname gz:diag:s0.8:seratio:REG:p90\endcsname{1.03}
\expandafter\def\csname gz:diag:s0.8:seratio:REG:p95\endcsname{1.04}
\expandafter\def\csname gz:diag:s0.8:width:IREG_HC1\endcsname{0.93}
\expandafter\def\csname gz:diag:s0.8:width:IREG_HC3\endcsname{0.95}
\expandafter\def\csname gz:diag:s0.8:width:REG_HC1\endcsname{0.93}
\expandafter\def\csname gz:diag:s0.8:width:REG_HC3\endcsname{0.94}
\expandafter\def\csname gz:diag:s1.0:seratio:IREG:median\endcsname{1.04}
\expandafter\def\csname gz:diag:s1.0:seratio:IREG:p90\endcsname{1.11}
\expandafter\def\csname gz:diag:s1.0:seratio:IREG:p95\endcsname{1.16}
\expandafter\def\csname gz:diag:s1.0:seratio:REG:median\endcsname{1.02}
\expandafter\def\csname gz:diag:s1.0:seratio:REG:p90\endcsname{1.05}
\expandafter\def\csname gz:diag:s1.0:seratio:REG:p95\endcsname{1.08}
\expandafter\def\csname gz:diag:s1.0:width:IREG_HC1\endcsname{1.13}
\expandafter\def\csname gz:diag:s1.0:width:IREG_HC3\endcsname{1.18}
\expandafter\def\csname gz:diag:s1.0:width:REG_HC1\endcsname{1.13}
\expandafter\def\csname gz:diag:s1.0:width:REG_HC3\endcsname{1.15}
\expandafter\def\csname gz:diag:s1.2:seratio:IREG:median\endcsname{1.07}
\expandafter\def\csname gz:diag:s1.2:seratio:IREG:p90\endcsname{1.22}
\expandafter\def\csname gz:diag:s1.2:seratio:IREG:p95\endcsname{1.31}
\expandafter\def\csname gz:diag:s1.2:seratio:REG:median\endcsname{1.03}
\expandafter\def\csname gz:diag:s1.2:seratio:REG:p90\endcsname{1.10}
\expandafter\def\csname gz:diag:s1.2:seratio:REG:p95\endcsname{1.15}
\expandafter\def\csname gz:diag:s1.2:width:IREG_HC1\endcsname{1.47}
\expandafter\def\csname gz:diag:s1.2:width:IREG_HC3\endcsname{1.58}
\expandafter\def\csname gz:diag:s1.2:width:REG_HC1\endcsname{1.46}
\expandafter\def\csname gz:diag:s1.2:width:REG_HC3\endcsname{1.51}
\expandafter\def\csname gz:diag:s1.5:seratio:IREG:median\endcsname{1.13}
\expandafter\def\csname gz:diag:s1.5:seratio:IREG:p90\endcsname{1.45}
\expandafter\def\csname gz:diag:s1.5:seratio:IREG:p95\endcsname{1.63}
\expandafter\def\csname gz:diag:s1.5:seratio:REG:median\endcsname{1.05}
\expandafter\def\csname gz:diag:s1.5:seratio:REG:p90\endcsname{1.19}
\expandafter\def\csname gz:diag:s1.5:seratio:REG:p95\endcsname{1.28}
\expandafter\def\csname gz:diag:s1.5:width:IREG_HC1\endcsname{2.43}
\expandafter\def\csname gz:diag:s1.5:width:IREG_HC3\endcsname{2.81}
\expandafter\def\csname gz:diag:s1.5:width:REG_HC1\endcsname{2.41}
\expandafter\def\csname gz:diag:s1.5:width:REG_HC3\endcsname{2.57}
\expandafter\def\csname gz:diag:s1.8:seratio:IREG:median\endcsname{1.22}
\expandafter\def\csname gz:diag:s1.8:seratio:IREG:p90\endcsname{1.71}
\expandafter\def\csname gz:diag:s1.8:seratio:IREG:p95\endcsname{2.04}
\expandafter\def\csname gz:diag:s1.8:seratio:REG:median\endcsname{1.07}
\expandafter\def\csname gz:diag:s1.8:seratio:REG:p90\endcsname{1.28}
\expandafter\def\csname gz:diag:s1.8:seratio:REG:p95\endcsname{1.41}
\expandafter\def\csname gz:diag:s1.8:width:IREG_HC1\endcsname{4.43}
\expandafter\def\csname gz:diag:s1.8:width:IREG_HC3\endcsname{5.56}
\expandafter\def\csname gz:diag:s1.8:width:REG_HC1\endcsname{4.39}
\expandafter\def\csname gz:diag:s1.8:width:REG_HC3\endcsname{4.81}
\expandafter\def\csname gz:diag:s2.0:seratio:IREG:median\endcsname{1.28}
\expandafter\def\csname gz:diag:s2.0:seratio:IREG:p90\endcsname{1.92}
\expandafter\def\csname gz:diag:s2.0:seratio:IREG:p95\endcsname{2.37}
\expandafter\def\csname gz:diag:s2.0:seratio:REG:median\endcsname{1.09}
\expandafter\def\csname gz:diag:s2.0:seratio:REG:p90\endcsname{1.34}
\expandafter\def\csname gz:diag:s2.0:seratio:REG:p95\endcsname{1.48}
\expandafter\def\csname gz:diag:s2.0:width:IREG_HC1\endcsname{6.87}
\expandafter\def\csname gz:diag:s2.0:width:IREG_HC3\endcsname{9.14}
\expandafter\def\csname gz:diag:s2.0:width:REG_HC1\endcsname{6.80}
\expandafter\def\csname gz:diag:s2.0:width:REG_HC3\endcsname{7.57}
\expandafter\def\csname gz:diag:scale:IREG_HC1\endcsname{1.43}
\expandafter\def\csname gz:diag:scale:IREG_HC1:0.80:obs\endcsname{64.4}
\expandafter\def\csname gz:diag:scale:IREG_HC1:0.80:pred\endcsname{63.0}
\expandafter\def\csname gz:diag:scale:IREG_HC1:0.80:qratio\endcsname{1.43}
\expandafter\def\csname gz:diag:scale:IREG_HC1:0.90:obs\endcsname{75.5}
\expandafter\def\csname gz:diag:scale:IREG_HC1:0.90:pred\endcsname{75.0}
\expandafter\def\csname gz:diag:scale:IREG_HC1:0.90:qratio\endcsname{1.46}
\expandafter\def\csname gz:diag:scale:IREG_HC1:0.95:obs\endcsname{83.0}
\expandafter\def\csname gz:diag:scale:IREG_HC1:0.95:pred\endcsname{83.0}
\expandafter\def\csname gz:diag:scale:IREG_HC1:0.95:qratio\endcsname{1.52}
\expandafter\def\csname gz:diag:scale:IREG_HC1:0.98:obs\endcsname{89.1}
\expandafter\def\csname gz:diag:scale:IREG_HC1:0.98:pred\endcsname{89.7}
\expandafter\def\csname gz:diag:scale:IREG_HC1:0.98:qratio\endcsname{1.64}
\expandafter\def\csname gz:diag:scale:IREG_HC1:0.99:obs\endcsname{91.9}
\expandafter\def\csname gz:diag:scale:IREG_HC1:0.99:pred\endcsname{92.9}
\expandafter\def\csname gz:diag:scale:IREG_HC1:0.99:qratio\endcsname{1.74}
\expandafter\def\csname gz:diag:scale:IREG_HC1:maxdev\endcsname{1.4}
\expandafter\def\csname gz:diag:scale:IREG_HC3\endcsname{1.05}
\expandafter\def\csname gz:diag:scale:IREG_HC3:0.80:obs\endcsname{79.5}
\expandafter\def\csname gz:diag:scale:IREG_HC3:0.80:pred\endcsname{77.9}
\expandafter\def\csname gz:diag:scale:IREG_HC3:0.80:qratio\endcsname{1.01}
\expandafter\def\csname gz:diag:scale:IREG_HC3:0.90:obs\endcsname{89.2}
\expandafter\def\csname gz:diag:scale:IREG_HC3:0.90:pred\endcsname{88.4}
\expandafter\def\csname gz:diag:scale:IREG_HC3:0.90:qratio\endcsname{1.02}
\expandafter\def\csname gz:diag:scale:IREG_HC3:0.95:obs\endcsname{93.9}
\expandafter\def\csname gz:diag:scale:IREG_HC3:0.95:pred\endcsname{93.9}
\expandafter\def\csname gz:diag:scale:IREG_HC3:0.95:qratio\endcsname{1.05}
\expandafter\def\csname gz:diag:scale:IREG_HC3:0.98:obs\endcsname{96.9}
\expandafter\def\csname gz:diag:scale:IREG_HC3:0.98:pred\endcsname{97.4}
\expandafter\def\csname gz:diag:scale:IREG_HC3:0.98:qratio\endcsname{1.10}
\expandafter\def\csname gz:diag:scale:IREG_HC3:0.99:obs\endcsname{98.0}
\expandafter\def\csname gz:diag:scale:IREG_HC3:0.99:pred\endcsname{98.6}
\expandafter\def\csname gz:diag:scale:IREG_HC3:0.99:qratio\endcsname{1.14}
\expandafter\def\csname gz:diag:scale:IREG_HC3:maxdev\endcsname{1.6}
\expandafter\def\csname gz:diag:scale:REG_HC1\endcsname{0.99}
\expandafter\def\csname gz:diag:scale:REG_HC1:0.80:obs\endcsname{76.3}
\expandafter\def\csname gz:diag:scale:REG_HC1:0.80:pred\endcsname{80.6}
\expandafter\def\csname gz:diag:scale:REG_HC1:0.80:qratio\endcsname{1.08}
\expandafter\def\csname gz:diag:scale:REG_HC1:0.90:obs\endcsname{88.8}
\expandafter\def\csname gz:diag:scale:REG_HC1:0.90:pred\endcsname{90.5}
\expandafter\def\csname gz:diag:scale:REG_HC1:0.90:qratio\endcsname{1.03}
\expandafter\def\csname gz:diag:scale:REG_HC1:0.95:obs\endcsname{95.3}
\expandafter\def\csname gz:diag:scale:REG_HC1:0.95:pred\endcsname{95.3}
\expandafter\def\csname gz:diag:scale:REG_HC1:0.95:qratio\endcsname{0.99}
\expandafter\def\csname gz:diag:scale:REG_HC1:0.98:obs\endcsname{98.6}
\expandafter\def\csname gz:diag:scale:REG_HC1:0.98:pred\endcsname{98.2}
\expandafter\def\csname gz:diag:scale:REG_HC1:0.98:qratio\endcsname{0.95}
\expandafter\def\csname gz:diag:scale:REG_HC1:0.99:obs\endcsname{99.5}
\expandafter\def\csname gz:diag:scale:REG_HC1:0.99:pred\endcsname{99.1}
\expandafter\def\csname gz:diag:scale:REG_HC1:0.99:qratio\endcsname{0.94}
\expandafter\def\csname gz:diag:scale:REG_HC1:maxdev\endcsname{4.3}
\expandafter\def\csname gz:estimatr:v1\endcsname{1.0.2}
\expandafter\def\csname gz:estimatr:v2\endcsname{2.0.0}
\expandafter\def\csname gz:fam:B\endcsname{20{,}000}
\expandafter\def\csname gz:fam:matrixcheck:maxrel\endcsname{5}
\expandafter\def\csname gz:fam:maxabsz\endcsname{1.55}
\expandafter\def\csname gz:fam:n\endcsname{4{,}000}
\expandafter\def\csname gz:fam:rate:heavytail1000:z\endcsname{-2.21}
\expandafter\def\csname gz:fam:rate:heavytail:nbr\endcsname{0.6}
\expandafter\def\csname gz:fam:rate:maxabsz\endcsname{2.21}
\expandafter\def\csname gz:fam:seratio:max\endcsname{11.4}
\expandafter\def\csname gz:fam:seratio:min\endcsname{3.9}
\expandafter\def\csname gz:fam:target:max\endcsname{36.3}
\expandafter\def\csname gz:fam:target:min\endcsname{-21.6}
\expandafter\def\csname gz:gauss_quad:IREG_HC1:all\endcsname{94.9}
\expandafter\def\csname gz:gauss_quad:IREG_HC1:all:nointerval\endcsname{0.0}
\expandafter\def\csname gz:gauss_quad:IREG_HC1:full\endcsname{94.9}
\expandafter\def\csname gz:gauss_quad:IREG_HC1:full:nointerval\endcsname{0.0}
\expandafter\def\csname gz:gauss_quad:IREG_HC1:lev1\endcsname{--}
\expandafter\def\csname gz:gauss_quad:IREG_HC2:all\endcsname{94.9}
\expandafter\def\csname gz:gauss_quad:IREG_HC2:all:nointerval\endcsname{0.0}
\expandafter\def\csname gz:gauss_quad:IREG_HC2:full\endcsname{94.9}
\expandafter\def\csname gz:gauss_quad:IREG_HC2:full:nointerval\endcsname{0.0}
\expandafter\def\csname gz:gauss_quad:IREG_HC2:lev1\endcsname{--}
\expandafter\def\csname gz:gauss_quad:IREG_HC3:all\endcsname{95.0}
\expandafter\def\csname gz:gauss_quad:IREG_HC3:all:nointerval\endcsname{0.0}
\expandafter\def\csname gz:gauss_quad:IREG_HC3:full\endcsname{95.0}
\expandafter\def\csname gz:gauss_quad:IREG_HC3:full:nointerval\endcsname{0.0}
\expandafter\def\csname gz:gauss_quad:IREG_HC3:lev1\endcsname{--}
\expandafter\def\csname gz:gauss_quad:REG_HC1:all\endcsname{94.9}
\expandafter\def\csname gz:gauss_quad:REG_HC1:all:nointerval\endcsname{0.0}
\expandafter\def\csname gz:gauss_quad:REG_HC1:full\endcsname{94.9}
\expandafter\def\csname gz:gauss_quad:REG_HC1:full:nointerval\endcsname{0.0}
\expandafter\def\csname gz:gauss_quad:REG_HC1:lev1\endcsname{--}
\expandafter\def\csname gz:gauss_quad:REG_HC3:all\endcsname{95.0}
\expandafter\def\csname gz:gauss_quad:REG_HC3:all:nointerval\endcsname{0.0}
\expandafter\def\csname gz:gauss_quad:REG_HC3:full\endcsname{95.0}
\expandafter\def\csname gz:gauss_quad:REG_HC3:full:nointerval\endcsname{0.0}
\expandafter\def\csname gz:gauss_quad:REG_HC3:lev1\endcsname{--}
\expandafter\def\csname gz:gauss_quad:lev1\endcsname{0.0}
\expandafter\def\csname gz:gauss_quad:medlev\endcsname{0.04}
\expandafter\def\csname gz:gauss_quad:rankdef\endcsname{0.0}
\expandafter\def\csname gz:gauss_quad:reps\endcsname{20{,}000}
\expandafter\def\csname gz:hte:ba:gauss_high:IREG_HC1:cal:0.80\endcsname{52.4}
\expandafter\def\csname gz:hte:ba:gauss_high:IREG_HC1:cal:0.90\endcsname{63.8}
\expandafter\def\csname gz:hte:ba:gauss_high:IREG_HC1:cal:0.95\endcsname{72.2}
\expandafter\def\csname gz:hte:ba:gauss_high:IREG_HC1:cal:0.98\endcsname{80.5}
\expandafter\def\csname gz:hte:ba:gauss_high:IREG_HC1:cal:0.99\endcsname{85.2}
\expandafter\def\csname gz:hte:ba:gauss_high:IREG_HC1:cov\endcsname{72.2}
\expandafter\def\csname gz:hte:ba:gauss_high:IREG_HC1:width\endcsname{0.35}
\expandafter\def\csname gz:hte:ba:gauss_high:IREG_HC1c:cal:0.80\endcsname{80.1}
\expandafter\def\csname gz:hte:ba:gauss_high:IREG_HC1c:cal:0.90\endcsname{90.0}
\expandafter\def\csname gz:hte:ba:gauss_high:IREG_HC1c:cal:0.95\endcsname{95.1}
\expandafter\def\csname gz:hte:ba:gauss_high:IREG_HC1c:cal:0.98\endcsname{98.1}
\expandafter\def\csname gz:hte:ba:gauss_high:IREG_HC1c:cal:0.99\endcsname{99.1}
\expandafter\def\csname gz:hte:ba:gauss_high:IREG_HC1c:cov\endcsname{95.1}
\expandafter\def\csname gz:hte:ba:gauss_high:IREG_HC1c:width\endcsname{0.64}
\expandafter\def\csname gz:hte:ba:gauss_high:IREG_HC1o:cal:0.80\endcsname{80.1}
\expandafter\def\csname gz:hte:ba:gauss_high:IREG_HC1o:cal:0.90\endcsname{90.0}
\expandafter\def\csname gz:hte:ba:gauss_high:IREG_HC1o:cal:0.95\endcsname{95.1}
\expandafter\def\csname gz:hte:ba:gauss_high:IREG_HC1o:cal:0.98\endcsname{98.1}
\expandafter\def\csname gz:hte:ba:gauss_high:IREG_HC1o:cal:0.99\endcsname{99.1}
\expandafter\def\csname gz:hte:ba:gauss_high:IREG_HC1o:cov\endcsname{95.1}
\expandafter\def\csname gz:hte:ba:gauss_high:IREG_HC1o:width\endcsname{0.63}
\expandafter\def\csname gz:hte:ba:gauss_high:IREG_HC3:cal:0.80\endcsname{53.0}
\expandafter\def\csname gz:hte:ba:gauss_high:IREG_HC3:cal:0.90\endcsname{64.5}
\expandafter\def\csname gz:hte:ba:gauss_high:IREG_HC3:cal:0.95\endcsname{72.9}
\expandafter\def\csname gz:hte:ba:gauss_high:IREG_HC3:cal:0.98\endcsname{81.0}
\expandafter\def\csname gz:hte:ba:gauss_high:IREG_HC3:cal:0.99\endcsname{85.7}
\expandafter\def\csname gz:hte:ba:gauss_high:IREG_HC3:cov\endcsname{72.9}
\expandafter\def\csname gz:hte:ba:gauss_high:IREG_HC3:width\endcsname{0.36}
\expandafter\def\csname gz:hte:ba:gauss_high:IREG_HC3c:cal:0.80\endcsname{80.3}
\expandafter\def\csname gz:hte:ba:gauss_high:IREG_HC3c:cal:0.90\endcsname{90.2}
\expandafter\def\csname gz:hte:ba:gauss_high:IREG_HC3c:cal:0.95\endcsname{95.2}
\expandafter\def\csname gz:hte:ba:gauss_high:IREG_HC3c:cal:0.98\endcsname{98.2}
\expandafter\def\csname gz:hte:ba:gauss_high:IREG_HC3c:cal:0.99\endcsname{99.1}
\expandafter\def\csname gz:hte:ba:gauss_high:IREG_HC3c:cov\endcsname{95.2}
\expandafter\def\csname gz:hte:ba:gauss_high:IREG_HC3c:width\endcsname{0.64}
\expandafter\def\csname gz:hte:ba:gauss_high:R2\endcsname{0.69}
\expandafter\def\csname gz:hte:ba:gauss_high:REG_HC1:cal:0.80\endcsname{80.1}
\expandafter\def\csname gz:hte:ba:gauss_high:REG_HC1:cal:0.90\endcsname{89.9}
\expandafter\def\csname gz:hte:ba:gauss_high:REG_HC1:cal:0.95\endcsname{95.0}
\expandafter\def\csname gz:hte:ba:gauss_high:REG_HC1:cal:0.98\endcsname{98.1}
\expandafter\def\csname gz:hte:ba:gauss_high:REG_HC1:cal:0.99\endcsname{99.0}
\expandafter\def\csname gz:hte:ba:gauss_high:REG_HC1:cov\endcsname{95.0}
\expandafter\def\csname gz:hte:ba:gauss_high:REG_HC1:width\endcsname{0.64}
\expandafter\def\csname gz:hte:ba:gauss_high:REG_HC3:cal:0.80\endcsname{80.6}
\expandafter\def\csname gz:hte:ba:gauss_high:REG_HC3:cal:0.90\endcsname{90.3}
\expandafter\def\csname gz:hte:ba:gauss_high:REG_HC3:cal:0.95\endcsname{95.3}
\expandafter\def\csname gz:hte:ba:gauss_high:REG_HC3:cal:0.98\endcsname{98.2}
\expandafter\def\csname gz:hte:ba:gauss_high:REG_HC3:cal:0.99\endcsname{99.1}
\expandafter\def\csname gz:hte:ba:gauss_high:REG_HC3:cov\endcsname{95.3}
\expandafter\def\csname gz:hte:ba:gauss_high:REG_HC3:width\endcsname{0.64}
\expandafter\def\csname gz:hte:ba:gauss_high:gstar\endcsname{3.00}
\expandafter\def\csname gz:hte:ba:gauss_high:lev1\endcsname{0.0}
\expandafter\def\csname gz:hte:ba:gauss_high:mse\endcsname{0.99}
\expandafter\def\csname gz:hte:ba:gauss_high:plugratio\endcsname{1.01}
\expandafter\def\csname gz:hte:ba:gauss_high:plugratio:fmt\endcsname{1.0}
\expandafter\def\csname gz:hte:ba:gauss_high:rankdef\endcsname{0.0}
\expandafter\def\csname gz:hte:ba:gauss_high:wr1\endcsname{1.00}
\expandafter\def\csname gz:hte:ba:gauss_high:wr3\endcsname{0.99}
\expandafter\def\csname gz:hte:ba:gauss_low:IREG_HC1:cal:0.80\endcsname{79.5}
\expandafter\def\csname gz:hte:ba:gauss_low:IREG_HC1:cal:0.90\endcsname{89.5}
\expandafter\def\csname gz:hte:ba:gauss_low:IREG_HC1:cal:0.95\endcsname{94.6}
\expandafter\def\csname gz:hte:ba:gauss_low:IREG_HC1:cal:0.98\endcsname{97.8}
\expandafter\def\csname gz:hte:ba:gauss_low:IREG_HC1:cal:0.99\endcsname{98.9}
\expandafter\def\csname gz:hte:ba:gauss_low:IREG_HC1:cov\endcsname{94.6}
\expandafter\def\csname gz:hte:ba:gauss_low:IREG_HC1:width\endcsname{0.35}
\expandafter\def\csname gz:hte:ba:gauss_low:IREG_HC1c:cal:0.80\endcsname{80.1}
\expandafter\def\csname gz:hte:ba:gauss_low:IREG_HC1c:cal:0.90\endcsname{90.0}
\expandafter\def\csname gz:hte:ba:gauss_low:IREG_HC1c:cal:0.95\endcsname{95.0}
\expandafter\def\csname gz:hte:ba:gauss_low:IREG_HC1c:cal:0.98\endcsname{98.1}
\expandafter\def\csname gz:hte:ba:gauss_low:IREG_HC1c:cal:0.99\endcsname{99.1}
\expandafter\def\csname gz:hte:ba:gauss_low:IREG_HC1c:cov\endcsname{95.0}
\expandafter\def\csname gz:hte:ba:gauss_low:IREG_HC1c:width\endcsname{0.36}
\expandafter\def\csname gz:hte:ba:gauss_low:IREG_HC1o:cal:0.80\endcsname{79.9}
\expandafter\def\csname gz:hte:ba:gauss_low:IREG_HC1o:cal:0.90\endcsname{89.8}
\expandafter\def\csname gz:hte:ba:gauss_low:IREG_HC1o:cal:0.95\endcsname{94.9}
\expandafter\def\csname gz:hte:ba:gauss_low:IREG_HC1o:cal:0.98\endcsname{98.0}
\expandafter\def\csname gz:hte:ba:gauss_low:IREG_HC1o:cal:0.99\endcsname{99.0}
\expandafter\def\csname gz:hte:ba:gauss_low:IREG_HC1o:cov\endcsname{94.9}
\expandafter\def\csname gz:hte:ba:gauss_low:IREG_HC1o:width\endcsname{0.36}
\expandafter\def\csname gz:hte:ba:gauss_low:IREG_HC3:cal:0.80\endcsname{79.9}
\expandafter\def\csname gz:hte:ba:gauss_low:IREG_HC3:cal:0.90\endcsname{89.8}
\expandafter\def\csname gz:hte:ba:gauss_low:IREG_HC3:cal:0.95\endcsname{94.9}
\expandafter\def\csname gz:hte:ba:gauss_low:IREG_HC3:cal:0.98\endcsname{98.0}
\expandafter\def\csname gz:hte:ba:gauss_low:IREG_HC3:cal:0.99\endcsname{99.0}
\expandafter\def\csname gz:hte:ba:gauss_low:IREG_HC3:cov\endcsname{94.9}
\expandafter\def\csname gz:hte:ba:gauss_low:IREG_HC3:width\endcsname{0.36}
\expandafter\def\csname gz:hte:ba:gauss_low:IREG_HC3c:cal:0.80\endcsname{80.7}
\expandafter\def\csname gz:hte:ba:gauss_low:IREG_HC3c:cal:0.90\endcsname{90.4}
\expandafter\def\csname gz:hte:ba:gauss_low:IREG_HC3c:cal:0.95\endcsname{95.2}
\expandafter\def\csname gz:hte:ba:gauss_low:IREG_HC3c:cal:0.98\endcsname{98.2}
\expandafter\def\csname gz:hte:ba:gauss_low:IREG_HC3c:cal:0.99\endcsname{99.1}
\expandafter\def\csname gz:hte:ba:gauss_low:IREG_HC3c:cov\endcsname{95.2}
\expandafter\def\csname gz:hte:ba:gauss_low:IREG_HC3c:width\endcsname{0.36}
\expandafter\def\csname gz:hte:ba:gauss_low:R2\endcsname{0.02}
\expandafter\def\csname gz:hte:ba:gauss_low:REG_HC1:cal:0.80\endcsname{79.9}
\expandafter\def\csname gz:hte:ba:gauss_low:REG_HC1:cal:0.90\endcsname{89.9}
\expandafter\def\csname gz:hte:ba:gauss_low:REG_HC1:cal:0.95\endcsname{94.8}
\expandafter\def\csname gz:hte:ba:gauss_low:REG_HC1:cal:0.98\endcsname{98.0}
\expandafter\def\csname gz:hte:ba:gauss_low:REG_HC1:cal:0.99\endcsname{99.0}
\expandafter\def\csname gz:hte:ba:gauss_low:REG_HC1:cov\endcsname{94.8}
\expandafter\def\csname gz:hte:ba:gauss_low:REG_HC1:width\endcsname{0.36}
\expandafter\def\csname gz:hte:ba:gauss_low:REG_HC3:cal:0.80\endcsname{80.2}
\expandafter\def\csname gz:hte:ba:gauss_low:REG_HC3:cal:0.90\endcsname{90.1}
\expandafter\def\csname gz:hte:ba:gauss_low:REG_HC3:cal:0.95\endcsname{95.0}
\expandafter\def\csname gz:hte:ba:gauss_low:REG_HC3:cal:0.98\endcsname{98.1}
\expandafter\def\csname gz:hte:ba:gauss_low:REG_HC3:cal:0.99\endcsname{99.0}
\expandafter\def\csname gz:hte:ba:gauss_low:REG_HC3:cov\endcsname{95.0}
\expandafter\def\csname gz:hte:ba:gauss_low:REG_HC3:width\endcsname{0.36}
\expandafter\def\csname gz:hte:ba:gauss_low:gstar\endcsname{0.30}
\expandafter\def\csname gz:hte:ba:gauss_low:lev1\endcsname{0.0}
\expandafter\def\csname gz:hte:ba:gauss_low:mse\endcsname{1.00}
\expandafter\def\csname gz:hte:ba:gauss_low:plugratio\endcsname{1.45}
\expandafter\def\csname gz:hte:ba:gauss_low:plugratio:fmt\endcsname{1.4}
\expandafter\def\csname gz:hte:ba:gauss_low:rankdef\endcsname{0.0}
\expandafter\def\csname gz:hte:ba:gauss_low:wr1\endcsname{1.00}
\expandafter\def\csname gz:hte:ba:gauss_low:wr3\endcsname{1.01}
\expandafter\def\csname gz:hte:ba:heavy_ex:IREG_HC1:cal:0.80\endcsname{67.4}
\expandafter\def\csname gz:hte:ba:heavy_ex:IREG_HC1:cal:0.90\endcsname{78.5}
\expandafter\def\csname gz:hte:ba:heavy_ex:IREG_HC1:cal:0.95\endcsname{85.5}
\expandafter\def\csname gz:hte:ba:heavy_ex:IREG_HC1:cal:0.98\endcsname{91.2}
\expandafter\def\csname gz:hte:ba:heavy_ex:IREG_HC1:cal:0.99\endcsname{93.7}
\expandafter\def\csname gz:hte:ba:heavy_ex:IREG_HC1:cov\endcsname{85.5}
\expandafter\def\csname gz:hte:ba:heavy_ex:IREG_HC1:width\endcsname{6.56}
\expandafter\def\csname gz:hte:ba:heavy_ex:IREG_HC1c:cal:0.80\endcsname{75.5}
\expandafter\def\csname gz:hte:ba:heavy_ex:IREG_HC1c:cal:0.90\endcsname{87.9}
\expandafter\def\csname gz:hte:ba:heavy_ex:IREG_HC1c:cal:0.95\endcsname{94.3}
\expandafter\def\csname gz:hte:ba:heavy_ex:IREG_HC1c:cal:0.98\endcsname{98.0}
\expandafter\def\csname gz:hte:ba:heavy_ex:IREG_HC1c:cal:0.99\endcsname{99.1}
\expandafter\def\csname gz:hte:ba:heavy_ex:IREG_HC1c:cov\endcsname{94.3}
\expandafter\def\csname gz:hte:ba:heavy_ex:IREG_HC1c:width\endcsname{7.26}
\expandafter\def\csname gz:hte:ba:heavy_ex:IREG_HC1o:cal:0.80\endcsname{68.3}
\expandafter\def\csname gz:hte:ba:heavy_ex:IREG_HC1o:cal:0.90\endcsname{79.3}
\expandafter\def\csname gz:hte:ba:heavy_ex:IREG_HC1o:cal:0.95\endcsname{86.4}
\expandafter\def\csname gz:hte:ba:heavy_ex:IREG_HC1o:cal:0.98\endcsname{91.7}
\expandafter\def\csname gz:hte:ba:heavy_ex:IREG_HC1o:cal:0.99\endcsname{94.2}
\expandafter\def\csname gz:hte:ba:heavy_ex:IREG_HC1o:cov\endcsname{86.4}
\expandafter\def\csname gz:hte:ba:heavy_ex:IREG_HC1o:width\endcsname{6.69}
\expandafter\def\csname gz:hte:ba:heavy_ex:IREG_HC3:cal:0.80\endcsname{80.0}
\expandafter\def\csname gz:hte:ba:heavy_ex:IREG_HC3:cal:0.90\endcsname{89.3}
\expandafter\def\csname gz:hte:ba:heavy_ex:IREG_HC3:cal:0.95\endcsname{94.0}
\expandafter\def\csname gz:hte:ba:heavy_ex:IREG_HC3:cal:0.98\endcsname{97.1}
\expandafter\def\csname gz:hte:ba:heavy_ex:IREG_HC3:cal:0.99\endcsname{98.1}
\expandafter\def\csname gz:hte:ba:heavy_ex:IREG_HC3:cov\endcsname{94.0}
\expandafter\def\csname gz:hte:ba:heavy_ex:IREG_HC3:width\endcsname{8.26}
\expandafter\def\csname gz:hte:ba:heavy_ex:IREG_HC3c:cal:0.80\endcsname{84.9}
\expandafter\def\csname gz:hte:ba:heavy_ex:IREG_HC3c:cal:0.90\endcsname{93.7}
\expandafter\def\csname gz:hte:ba:heavy_ex:IREG_HC3c:cal:0.95\endcsname{97.3}
\expandafter\def\csname gz:hte:ba:heavy_ex:IREG_HC3c:cal:0.98\endcsname{99.1}
\expandafter\def\csname gz:hte:ba:heavy_ex:IREG_HC3c:cal:0.99\endcsname{99.7}
\expandafter\def\csname gz:hte:ba:heavy_ex:IREG_HC3c:cov\endcsname{97.3}
\expandafter\def\csname gz:hte:ba:heavy_ex:IREG_HC3c:width\endcsname{8.85}
\expandafter\def\csname gz:hte:ba:heavy_ex:R2\endcsname{0.009}
\expandafter\def\csname gz:hte:ba:heavy_ex:REG_HC1:cal:0.80\endcsname{76.4}
\expandafter\def\csname gz:hte:ba:heavy_ex:REG_HC1:cal:0.90\endcsname{88.3}
\expandafter\def\csname gz:hte:ba:heavy_ex:REG_HC1:cal:0.95\endcsname{94.3}
\expandafter\def\csname gz:hte:ba:heavy_ex:REG_HC1:cal:0.98\endcsname{97.9}
\expandafter\def\csname gz:hte:ba:heavy_ex:REG_HC1:cal:0.99\endcsname{99.0}
\expandafter\def\csname gz:hte:ba:heavy_ex:REG_HC1:cov\endcsname{94.3}
\expandafter\def\csname gz:hte:ba:heavy_ex:REG_HC1:width\endcsname{6.57}
\expandafter\def\csname gz:hte:ba:heavy_ex:REG_HC3:cal:0.80\endcsname{83.1}
\expandafter\def\csname gz:hte:ba:heavy_ex:REG_HC3:cal:0.90\endcsname{92.5}
\expandafter\def\csname gz:hte:ba:heavy_ex:REG_HC3:cal:0.95\endcsname{96.6}
\expandafter\def\csname gz:hte:ba:heavy_ex:REG_HC3:cal:0.98\endcsname{98.9}
\expandafter\def\csname gz:hte:ba:heavy_ex:REG_HC3:cal:0.99\endcsname{99.5}
\expandafter\def\csname gz:hte:ba:heavy_ex:REG_HC3:cov\endcsname{96.6}
\expandafter\def\csname gz:hte:ba:heavy_ex:REG_HC3:width\endcsname{7.19}
\expandafter\def\csname gz:hte:ba:heavy_ex:gstar\endcsname{0.30}
\expandafter\def\csname gz:hte:ba:heavy_ex:lev1\endcsname{0.0}
\expandafter\def\csname gz:hte:ba:heavy_ex:mse\endcsname{1.66}
\expandafter\def\csname gz:hte:ba:heavy_ex:plugratio\endcsname{25.64}
\expandafter\def\csname gz:hte:ba:heavy_ex:plugratio:fmt\endcsname{25.6}
\expandafter\def\csname gz:hte:ba:heavy_ex:rankdef\endcsname{0.0}
\expandafter\def\csname gz:hte:ba:heavy_ex:wr1\endcsname{1.11}
\expandafter\def\csname gz:hte:ba:heavy_ex:wr3\endcsname{1.23}
\expandafter\def\csname gz:hte:ba:heavy_high:IREG_HC1:cal:0.80\endcsname{67.4}
\expandafter\def\csname gz:hte:ba:heavy_high:IREG_HC1:cal:0.90\endcsname{79.3}
\expandafter\def\csname gz:hte:ba:heavy_high:IREG_HC1:cal:0.95\endcsname{86.8}
\expandafter\def\csname gz:hte:ba:heavy_high:IREG_HC1:cal:0.98\endcsname{92.3}
\expandafter\def\csname gz:hte:ba:heavy_high:IREG_HC1:cal:0.99\endcsname{94.8}
\expandafter\def\csname gz:hte:ba:heavy_high:IREG_HC1:cov\endcsname{86.8}
\expandafter\def\csname gz:hte:ba:heavy_high:IREG_HC1:width\endcsname{0.62}
\expandafter\def\csname gz:hte:ba:heavy_high:IREG_HC1c:cal:0.80\endcsname{78.6}
\expandafter\def\csname gz:hte:ba:heavy_high:IREG_HC1c:cal:0.90\endcsname{89.0}
\expandafter\def\csname gz:hte:ba:heavy_high:IREG_HC1c:cal:0.95\endcsname{94.2}
\expandafter\def\csname gz:hte:ba:heavy_high:IREG_HC1c:cal:0.98\endcsname{97.3}
\expandafter\def\csname gz:hte:ba:heavy_high:IREG_HC1c:cal:0.99\endcsname{98.5}
\expandafter\def\csname gz:hte:ba:heavy_high:IREG_HC1c:cov\endcsname{94.2}
\expandafter\def\csname gz:hte:ba:heavy_high:IREG_HC1c:width\endcsname{0.75}
\expandafter\def\csname gz:hte:ba:heavy_high:IREG_HC1o:cal:0.80\endcsname{80.0}
\expandafter\def\csname gz:hte:ba:heavy_high:IREG_HC1o:cal:0.90\endcsname{89.7}
\expandafter\def\csname gz:hte:ba:heavy_high:IREG_HC1o:cal:0.95\endcsname{94.6}
\expandafter\def\csname gz:hte:ba:heavy_high:IREG_HC1o:cal:0.98\endcsname{97.6}
\expandafter\def\csname gz:hte:ba:heavy_high:IREG_HC1o:cal:0.99\endcsname{98.8}
\expandafter\def\csname gz:hte:ba:heavy_high:IREG_HC1o:cov\endcsname{94.6}
\expandafter\def\csname gz:hte:ba:heavy_high:IREG_HC1o:width\endcsname{0.82}
\expandafter\def\csname gz:hte:ba:heavy_high:IREG_HC3:cal:0.80\endcsname{71.6}
\expandafter\def\csname gz:hte:ba:heavy_high:IREG_HC3:cal:0.90\endcsname{82.9}
\expandafter\def\csname gz:hte:ba:heavy_high:IREG_HC3:cal:0.95\endcsname{89.6}
\expandafter\def\csname gz:hte:ba:heavy_high:IREG_HC3:cal:0.98\endcsname{94.5}
\expandafter\def\csname gz:hte:ba:heavy_high:IREG_HC3:cal:0.99\endcsname{96.5}
\expandafter\def\csname gz:hte:ba:heavy_high:IREG_HC3:cov\endcsname{89.6}
\expandafter\def\csname gz:hte:ba:heavy_high:IREG_HC3:width\endcsname{0.66}
\expandafter\def\csname gz:hte:ba:heavy_high:IREG_HC3c:cal:0.80\endcsname{81.0}
\expandafter\def\csname gz:hte:ba:heavy_high:IREG_HC3c:cal:0.90\endcsname{90.7}
\expandafter\def\csname gz:hte:ba:heavy_high:IREG_HC3c:cal:0.95\endcsname{95.2}
\expandafter\def\csname gz:hte:ba:heavy_high:IREG_HC3c:cal:0.98\endcsname{97.9}
\expandafter\def\csname gz:hte:ba:heavy_high:IREG_HC3c:cal:0.99\endcsname{98.8}
\expandafter\def\csname gz:hte:ba:heavy_high:IREG_HC3c:cov\endcsname{95.2}
\expandafter\def\csname gz:hte:ba:heavy_high:IREG_HC3c:width\endcsname{0.79}
\expandafter\def\csname gz:hte:ba:heavy_high:R2\endcsname{0.37}
\expandafter\def\csname gz:hte:ba:heavy_high:REG_HC1:cal:0.80\endcsname{75.5}
\expandafter\def\csname gz:hte:ba:heavy_high:REG_HC1:cal:0.90\endcsname{86.0}
\expandafter\def\csname gz:hte:ba:heavy_high:REG_HC1:cal:0.95\endcsname{91.9}
\expandafter\def\csname gz:hte:ba:heavy_high:REG_HC1:cal:0.98\endcsname{96.0}
\expandafter\def\csname gz:hte:ba:heavy_high:REG_HC1:cal:0.99\endcsname{97.7}
\expandafter\def\csname gz:hte:ba:heavy_high:REG_HC1:cov\endcsname{91.9}
\expandafter\def\csname gz:hte:ba:heavy_high:REG_HC1:width\endcsname{0.73}
\expandafter\def\csname gz:hte:ba:heavy_high:REG_HC3:cal:0.80\endcsname{78.2}
\expandafter\def\csname gz:hte:ba:heavy_high:REG_HC3:cal:0.90\endcsname{87.7}
\expandafter\def\csname gz:hte:ba:heavy_high:REG_HC3:cal:0.95\endcsname{93.1}
\expandafter\def\csname gz:hte:ba:heavy_high:REG_HC3:cal:0.98\endcsname{96.7}
\expandafter\def\csname gz:hte:ba:heavy_high:REG_HC3:cal:0.99\endcsname{98.0}
\expandafter\def\csname gz:hte:ba:heavy_high:REG_HC3:cov\endcsname{93.1}
\expandafter\def\csname gz:hte:ba:heavy_high:REG_HC3:width\endcsname{0.76}
\expandafter\def\csname gz:hte:ba:heavy_high:gstar\endcsname{3.00}
\expandafter\def\csname gz:hte:ba:heavy_high:lev1\endcsname{0.0}
\expandafter\def\csname gz:hte:ba:heavy_high:mse\endcsname{1.08}
\expandafter\def\csname gz:hte:ba:heavy_high:plugratio\endcsname{1.14}
\expandafter\def\csname gz:hte:ba:heavy_high:plugratio:fmt\endcsname{1.1}
\expandafter\def\csname gz:hte:ba:heavy_high:rankdef\endcsname{0.0}
\expandafter\def\csname gz:hte:ba:heavy_high:wr1\endcsname{1.03}
\expandafter\def\csname gz:hte:ba:heavy_high:wr3\endcsname{1.04}
\expandafter\def\csname gz:hte:ba:heavy_low:IREG_HC1:cal:0.80\endcsname{77.0}
\expandafter\def\csname gz:hte:ba:heavy_low:IREG_HC1:cal:0.90\endcsname{87.9}
\expandafter\def\csname gz:hte:ba:heavy_low:IREG_HC1:cal:0.95\endcsname{93.2}
\expandafter\def\csname gz:hte:ba:heavy_low:IREG_HC1:cal:0.98\endcsname{97.1}
\expandafter\def\csname gz:hte:ba:heavy_low:IREG_HC1:cal:0.99\endcsname{98.3}
\expandafter\def\csname gz:hte:ba:heavy_low:IREG_HC1:cov\endcsname{93.2}
\expandafter\def\csname gz:hte:ba:heavy_low:IREG_HC1:width\endcsname{0.62}
\expandafter\def\csname gz:hte:ba:heavy_low:IREG_HC1c:cal:0.80\endcsname{79.1}
\expandafter\def\csname gz:hte:ba:heavy_low:IREG_HC1c:cal:0.90\endcsname{89.7}
\expandafter\def\csname gz:hte:ba:heavy_low:IREG_HC1c:cal:0.95\endcsname{94.7}
\expandafter\def\csname gz:hte:ba:heavy_low:IREG_HC1c:cal:0.98\endcsname{98.0}
\expandafter\def\csname gz:hte:ba:heavy_low:IREG_HC1c:cal:0.99\endcsname{98.9}
\expandafter\def\csname gz:hte:ba:heavy_low:IREG_HC1c:cov\endcsname{94.7}
\expandafter\def\csname gz:hte:ba:heavy_low:IREG_HC1c:width\endcsname{0.64}
\expandafter\def\csname gz:hte:ba:heavy_low:IREG_HC1o:cal:0.80\endcsname{77.5}
\expandafter\def\csname gz:hte:ba:heavy_low:IREG_HC1o:cal:0.90\endcsname{88.2}
\expandafter\def\csname gz:hte:ba:heavy_low:IREG_HC1o:cal:0.95\endcsname{93.5}
\expandafter\def\csname gz:hte:ba:heavy_low:IREG_HC1o:cal:0.98\endcsname{97.3}
\expandafter\def\csname gz:hte:ba:heavy_low:IREG_HC1o:cal:0.99\endcsname{98.4}
\expandafter\def\csname gz:hte:ba:heavy_low:IREG_HC1o:cov\endcsname{93.5}
\expandafter\def\csname gz:hte:ba:heavy_low:IREG_HC1o:width\endcsname{0.63}
\expandafter\def\csname gz:hte:ba:heavy_low:IREG_HC3:cal:0.80\endcsname{81.0}
\expandafter\def\csname gz:hte:ba:heavy_low:IREG_HC3:cal:0.90\endcsname{90.8}
\expandafter\def\csname gz:hte:ba:heavy_low:IREG_HC3:cal:0.95\endcsname{95.3}
\expandafter\def\csname gz:hte:ba:heavy_low:IREG_HC3:cal:0.98\endcsname{98.2}
\expandafter\def\csname gz:hte:ba:heavy_low:IREG_HC3:cal:0.99\endcsname{99.0}
\expandafter\def\csname gz:hte:ba:heavy_low:IREG_HC3:cov\endcsname{95.3}
\expandafter\def\csname gz:hte:ba:heavy_low:IREG_HC3:width\endcsname{0.66}
\expandafter\def\csname gz:hte:ba:heavy_low:IREG_HC3c:cal:0.80\endcsname{82.6}
\expandafter\def\csname gz:hte:ba:heavy_low:IREG_HC3c:cal:0.90\endcsname{92.0}
\expandafter\def\csname gz:hte:ba:heavy_low:IREG_HC3c:cal:0.95\endcsname{96.2}
\expandafter\def\csname gz:hte:ba:heavy_low:IREG_HC3c:cal:0.98\endcsname{98.6}
\expandafter\def\csname gz:hte:ba:heavy_low:IREG_HC3c:cal:0.99\endcsname{99.3}
\expandafter\def\csname gz:hte:ba:heavy_low:IREG_HC3c:cov\endcsname{96.2}
\expandafter\def\csname gz:hte:ba:heavy_low:IREG_HC3c:width\endcsname{0.68}
\expandafter\def\csname gz:hte:ba:heavy_low:R2\endcsname{0.02}
\expandafter\def\csname gz:hte:ba:heavy_low:REG_HC1:cal:0.80\endcsname{79.0}
\expandafter\def\csname gz:hte:ba:heavy_low:REG_HC1:cal:0.90\endcsname{89.6}
\expandafter\def\csname gz:hte:ba:heavy_low:REG_HC1:cal:0.95\endcsname{94.7}
\expandafter\def\csname gz:hte:ba:heavy_low:REG_HC1:cal:0.98\endcsname{97.9}
\expandafter\def\csname gz:hte:ba:heavy_low:REG_HC1:cal:0.99\endcsname{98.9}
\expandafter\def\csname gz:hte:ba:heavy_low:REG_HC1:cov\endcsname{94.7}
\expandafter\def\csname gz:hte:ba:heavy_low:REG_HC1:width\endcsname{0.63}
\expandafter\def\csname gz:hte:ba:heavy_low:REG_HC3:cal:0.80\endcsname{80.9}
\expandafter\def\csname gz:hte:ba:heavy_low:REG_HC3:cal:0.90\endcsname{91.0}
\expandafter\def\csname gz:hte:ba:heavy_low:REG_HC3:cal:0.95\endcsname{95.5}
\expandafter\def\csname gz:hte:ba:heavy_low:REG_HC3:cal:0.98\endcsname{98.3}
\expandafter\def\csname gz:hte:ba:heavy_low:REG_HC3:cal:0.99\endcsname{99.2}
\expandafter\def\csname gz:hte:ba:heavy_low:REG_HC3:cov\endcsname{95.5}
\expandafter\def\csname gz:hte:ba:heavy_low:REG_HC3:width\endcsname{0.64}
\expandafter\def\csname gz:hte:ba:heavy_low:gstar\endcsname{0.50}
\expandafter\def\csname gz:hte:ba:heavy_low:lev1\endcsname{0.0}
\expandafter\def\csname gz:hte:ba:heavy_low:mse\endcsname{1.08}
\expandafter\def\csname gz:hte:ba:heavy_low:plugratio\endcsname{5.70}
\expandafter\def\csname gz:hte:ba:heavy_low:plugratio:fmt\endcsname{5.7}
\expandafter\def\csname gz:hte:ba:heavy_low:rankdef\endcsname{0.0}
\expandafter\def\csname gz:hte:ba:heavy_low:wr1\endcsname{1.02}
\expandafter\def\csname gz:hte:ba:heavy_low:wr3\endcsname{1.05}
\expandafter\def\csname gz:hte:ba:ln12_small:IREG_HC1:cal:0.80\endcsname{75.8}
\expandafter\def\csname gz:hte:ba:ln12_small:IREG_HC1:cal:0.90\endcsname{86.7}
\expandafter\def\csname gz:hte:ba:ln12_small:IREG_HC1:cal:0.95\endcsname{92.6}
\expandafter\def\csname gz:hte:ba:ln12_small:IREG_HC1:cal:0.98\endcsname{96.5}
\expandafter\def\csname gz:hte:ba:ln12_small:IREG_HC1:cal:0.99\endcsname{97.9}
\expandafter\def\csname gz:hte:ba:ln12_small:IREG_HC1:cov\endcsname{92.6}
\expandafter\def\csname gz:hte:ba:ln12_small:IREG_HC1:width\endcsname{1.47}
\expandafter\def\csname gz:hte:ba:ln12_small:IREG_HC1c:cal:0.80\endcsname{78.4}
\expandafter\def\csname gz:hte:ba:ln12_small:IREG_HC1c:cal:0.90\endcsname{88.9}
\expandafter\def\csname gz:hte:ba:ln12_small:IREG_HC1c:cal:0.95\endcsname{94.6}
\expandafter\def\csname gz:hte:ba:ln12_small:IREG_HC1c:cal:0.98\endcsname{97.9}
\expandafter\def\csname gz:hte:ba:ln12_small:IREG_HC1c:cal:0.99\endcsname{98.9}
\expandafter\def\csname gz:hte:ba:ln12_small:IREG_HC1c:cov\endcsname{94.6}
\expandafter\def\csname gz:hte:ba:ln12_small:IREG_HC1c:width\endcsname{1.52}
\expandafter\def\csname gz:hte:ba:ln12_small:IREG_HC1o:cal:0.80\endcsname{76.3}
\expandafter\def\csname gz:hte:ba:ln12_small:IREG_HC1o:cal:0.90\endcsname{87.0}
\expandafter\def\csname gz:hte:ba:ln12_small:IREG_HC1o:cal:0.95\endcsname{92.9}
\expandafter\def\csname gz:hte:ba:ln12_small:IREG_HC1o:cal:0.98\endcsname{96.7}
\expandafter\def\csname gz:hte:ba:ln12_small:IREG_HC1o:cal:0.99\endcsname{98.1}
\expandafter\def\csname gz:hte:ba:ln12_small:IREG_HC1o:cov\endcsname{92.9}
\expandafter\def\csname gz:hte:ba:ln12_small:IREG_HC1o:width\endcsname{1.48}
\expandafter\def\csname gz:hte:ba:ln12_small:IREG_HC3:cal:0.80\endcsname{80.5}
\expandafter\def\csname gz:hte:ba:ln12_small:IREG_HC3:cal:0.90\endcsname{90.1}
\expandafter\def\csname gz:hte:ba:ln12_small:IREG_HC3:cal:0.95\endcsname{94.9}
\expandafter\def\csname gz:hte:ba:ln12_small:IREG_HC3:cal:0.98\endcsname{97.9}
\expandafter\def\csname gz:hte:ba:ln12_small:IREG_HC3:cal:0.99\endcsname{98.8}
\expandafter\def\csname gz:hte:ba:ln12_small:IREG_HC3:cov\endcsname{94.9}
\expandafter\def\csname gz:hte:ba:ln12_small:IREG_HC3:width\endcsname{1.59}
\expandafter\def\csname gz:hte:ba:ln12_small:IREG_HC3c:cal:0.80\endcsname{82.3}
\expandafter\def\csname gz:hte:ba:ln12_small:IREG_HC3c:cal:0.90\endcsname{91.7}
\expandafter\def\csname gz:hte:ba:ln12_small:IREG_HC3c:cal:0.95\endcsname{96.1}
\expandafter\def\csname gz:hte:ba:ln12_small:IREG_HC3c:cal:0.98\endcsname{98.5}
\expandafter\def\csname gz:hte:ba:ln12_small:IREG_HC3c:cal:0.99\endcsname{99.3}
\expandafter\def\csname gz:hte:ba:ln12_small:IREG_HC3c:cov\endcsname{96.1}
\expandafter\def\csname gz:hte:ba:ln12_small:IREG_HC3c:width\endcsname{1.63}
\expandafter\def\csname gz:hte:ba:ln12_small:R2\endcsname{0.01}
\expandafter\def\csname gz:hte:ba:ln12_small:REG_HC1:cal:0.80\endcsname{78.6}
\expandafter\def\csname gz:hte:ba:ln12_small:REG_HC1:cal:0.90\endcsname{89.0}
\expandafter\def\csname gz:hte:ba:ln12_small:REG_HC1:cal:0.95\endcsname{94.6}
\expandafter\def\csname gz:hte:ba:ln12_small:REG_HC1:cal:0.98\endcsname{97.9}
\expandafter\def\csname gz:hte:ba:ln12_small:REG_HC1:cal:0.99\endcsname{98.9}
\expandafter\def\csname gz:hte:ba:ln12_small:REG_HC1:cov\endcsname{94.6}
\expandafter\def\csname gz:hte:ba:ln12_small:REG_HC1:width\endcsname{1.47}
\expandafter\def\csname gz:hte:ba:ln12_small:REG_HC3:cal:0.80\endcsname{80.8}
\expandafter\def\csname gz:hte:ba:ln12_small:REG_HC3:cal:0.90\endcsname{90.8}
\expandafter\def\csname gz:hte:ba:ln12_small:REG_HC3:cal:0.95\endcsname{95.7}
\expandafter\def\csname gz:hte:ba:ln12_small:REG_HC3:cal:0.98\endcsname{98.4}
\expandafter\def\csname gz:hte:ba:ln12_small:REG_HC3:cal:0.99\endcsname{99.2}
\expandafter\def\csname gz:hte:ba:ln12_small:REG_HC3:cov\endcsname{95.7}
\expandafter\def\csname gz:hte:ba:ln12_small:REG_HC3:width\endcsname{1.53}
\expandafter\def\csname gz:hte:ba:ln12_small:gstar\endcsname{0.30}
\expandafter\def\csname gz:hte:ba:ln12_small:lev1\endcsname{0.0}
\expandafter\def\csname gz:hte:ba:ln12_small:mse\endcsname{1.09}
\expandafter\def\csname gz:hte:ba:ln12_small:plugratio\endcsname{6.66}
\expandafter\def\csname gz:hte:ba:ln12_small:plugratio:fmt\endcsname{6.7}
\expandafter\def\csname gz:hte:ba:ln12_small:rankdef\endcsname{0.0}
\expandafter\def\csname gz:hte:ba:ln12_small:wr1\endcsname{1.03}
\expandafter\def\csname gz:hte:ba:ln12_small:wr3\endcsname{1.07}
\expandafter\def\csname gz:hte:ba:ln15_g05:IREG_HC1:cal:0.80\endcsname{70.7}
\expandafter\def\csname gz:hte:ba:ln15_g05:IREG_HC1:cal:0.90\endcsname{81.9}
\expandafter\def\csname gz:hte:ba:ln15_g05:IREG_HC1:cal:0.95\endcsname{88.8}
\expandafter\def\csname gz:hte:ba:ln15_g05:IREG_HC1:cal:0.98\endcsname{93.7}
\expandafter\def\csname gz:hte:ba:ln15_g05:IREG_HC1:cal:0.99\endcsname{95.9}
\expandafter\def\csname gz:hte:ba:ln15_g05:IREG_HC1:cov\endcsname{88.8}
\expandafter\def\csname gz:hte:ba:ln15_g05:IREG_HC1:width\endcsname{2.43}
\expandafter\def\csname gz:hte:ba:ln15_g05:IREG_HC1c:cal:0.80\endcsname{76.6}
\expandafter\def\csname gz:hte:ba:ln15_g05:IREG_HC1c:cal:0.90\endcsname{87.7}
\expandafter\def\csname gz:hte:ba:ln15_g05:IREG_HC1c:cal:0.95\endcsname{93.7}
\expandafter\def\csname gz:hte:ba:ln15_g05:IREG_HC1c:cal:0.98\endcsname{97.2}
\expandafter\def\csname gz:hte:ba:ln15_g05:IREG_HC1c:cal:0.99\endcsname{98.5}
\expandafter\def\csname gz:hte:ba:ln15_g05:IREG_HC1c:cov\endcsname{93.7}
\expandafter\def\csname gz:hte:ba:ln15_g05:IREG_HC1c:width\endcsname{2.62}
\expandafter\def\csname gz:hte:ba:ln15_g05:IREG_HC1o:cal:0.80\endcsname{73.2}
\expandafter\def\csname gz:hte:ba:ln15_g05:IREG_HC1o:cal:0.90\endcsname{84.2}
\expandafter\def\csname gz:hte:ba:ln15_g05:IREG_HC1o:cal:0.95\endcsname{90.5}
\expandafter\def\csname gz:hte:ba:ln15_g05:IREG_HC1o:cal:0.98\endcsname{95.0}
\expandafter\def\csname gz:hte:ba:ln15_g05:IREG_HC1o:cal:0.99\endcsname{96.8}
\expandafter\def\csname gz:hte:ba:ln15_g05:IREG_HC1o:cov\endcsname{90.5}
\expandafter\def\csname gz:hte:ba:ln15_g05:IREG_HC1o:width\endcsname{2.55}
\expandafter\def\csname gz:hte:ba:ln15_g05:IREG_HC3:cal:0.80\endcsname{79.0}
\expandafter\def\csname gz:hte:ba:ln15_g05:IREG_HC3:cal:0.90\endcsname{88.8}
\expandafter\def\csname gz:hte:ba:ln15_g05:IREG_HC3:cal:0.95\endcsname{93.8}
\expandafter\def\csname gz:hte:ba:ln15_g05:IREG_HC3:cal:0.98\endcsname{97.1}
\expandafter\def\csname gz:hte:ba:ln15_g05:IREG_HC3:cal:0.99\endcsname{98.4}
\expandafter\def\csname gz:hte:ba:ln15_g05:IREG_HC3:cov\endcsname{93.8}
\expandafter\def\csname gz:hte:ba:ln15_g05:IREG_HC3:width\endcsname{2.80}
\expandafter\def\csname gz:hte:ba:ln15_g05:IREG_HC3c:cal:0.80\endcsname{83.0}
\expandafter\def\csname gz:hte:ba:ln15_g05:IREG_HC3c:cal:0.90\endcsname{92.0}
\expandafter\def\csname gz:hte:ba:ln15_g05:IREG_HC3c:cal:0.95\endcsname{96.2}
\expandafter\def\csname gz:hte:ba:ln15_g05:IREG_HC3c:cal:0.98\endcsname{98.5}
\expandafter\def\csname gz:hte:ba:ln15_g05:IREG_HC3c:cal:0.99\endcsname{99.2}
\expandafter\def\csname gz:hte:ba:ln15_g05:IREG_HC3c:cov\endcsname{96.2}
\expandafter\def\csname gz:hte:ba:ln15_g05:IREG_HC3c:width\endcsname{2.98}
\expandafter\def\csname gz:hte:ba:ln15_g05:R2\endcsname{0.05}
\expandafter\def\csname gz:hte:ba:ln15_g05:REG_HC1:cal:0.80\endcsname{76.4}
\expandafter\def\csname gz:hte:ba:ln15_g05:REG_HC1:cal:0.90\endcsname{87.3}
\expandafter\def\csname gz:hte:ba:ln15_g05:REG_HC1:cal:0.95\endcsname{93.2}
\expandafter\def\csname gz:hte:ba:ln15_g05:REG_HC1:cal:0.98\endcsname{96.9}
\expandafter\def\csname gz:hte:ba:ln15_g05:REG_HC1:cal:0.99\endcsname{98.3}
\expandafter\def\csname gz:hte:ba:ln15_g05:REG_HC1:cov\endcsname{93.2}
\expandafter\def\csname gz:hte:ba:ln15_g05:REG_HC1:width\endcsname{2.46}
\expandafter\def\csname gz:hte:ba:ln15_g05:REG_HC3:cal:0.80\endcsname{80.6}
\expandafter\def\csname gz:hte:ba:ln15_g05:REG_HC3:cal:0.90\endcsname{90.4}
\expandafter\def\csname gz:hte:ba:ln15_g05:REG_HC3:cal:0.95\endcsname{95.1}
\expandafter\def\csname gz:hte:ba:ln15_g05:REG_HC3:cal:0.98\endcsname{97.9}
\expandafter\def\csname gz:hte:ba:ln15_g05:REG_HC3:cal:0.99\endcsname{98.8}
\expandafter\def\csname gz:hte:ba:ln15_g05:REG_HC3:cov\endcsname{95.1}
\expandafter\def\csname gz:hte:ba:ln15_g05:REG_HC3:width\endcsname{2.61}
\expandafter\def\csname gz:hte:ba:ln15_g05:gstar\endcsname{0.50}
\expandafter\def\csname gz:hte:ba:ln15_g05:lev1\endcsname{0.0}
\expandafter\def\csname gz:hte:ba:ln15_g05:mse\endcsname{1.24}
\expandafter\def\csname gz:hte:ba:ln15_g05:plugratio\endcsname{4.08}
\expandafter\def\csname gz:hte:ba:ln15_g05:plugratio:fmt\endcsname{4.1}
\expandafter\def\csname gz:hte:ba:ln15_g05:rankdef\endcsname{0.0}
\expandafter\def\csname gz:hte:ba:ln15_g05:wr1\endcsname{1.06}
\expandafter\def\csname gz:hte:ba:ln15_g05:wr3\endcsname{1.14}
\expandafter\def\csname gz:hte:ba:ln15_large:IREG_HC1:cal:0.80\endcsname{45.8}
\expandafter\def\csname gz:hte:ba:ln15_large:IREG_HC1:cal:0.90\endcsname{56.3}
\expandafter\def\csname gz:hte:ba:ln15_large:IREG_HC1:cal:0.95\endcsname{64.4}
\expandafter\def\csname gz:hte:ba:ln15_large:IREG_HC1:cal:0.98\endcsname{72.4}
\expandafter\def\csname gz:hte:ba:ln15_large:IREG_HC1:cal:0.99\endcsname{76.9}
\expandafter\def\csname gz:hte:ba:ln15_large:IREG_HC1:cov\endcsname{64.4}
\expandafter\def\csname gz:hte:ba:ln15_large:IREG_HC1:width\endcsname{2.42}
\expandafter\def\csname gz:hte:ba:ln15_large:IREG_HC1c:cal:0.80\endcsname{75.3}
\expandafter\def\csname gz:hte:ba:ln15_large:IREG_HC1c:cal:0.90\endcsname{85.1}
\expandafter\def\csname gz:hte:ba:ln15_large:IREG_HC1c:cal:0.95\endcsname{89.8}
\expandafter\def\csname gz:hte:ba:ln15_large:IREG_HC1c:cal:0.98\endcsname{93.7}
\expandafter\def\csname gz:hte:ba:ln15_large:IREG_HC1c:cal:0.99\endcsname{95.4}
\expandafter\def\csname gz:hte:ba:ln15_large:IREG_HC1c:cov\endcsname{89.8}
\expandafter\def\csname gz:hte:ba:ln15_large:IREG_HC1c:width\endcsname{4.46}
\expandafter\def\csname gz:hte:ba:ln15_large:IREG_HC1o:cal:0.80\endcsname{81.1}
\expandafter\def\csname gz:hte:ba:ln15_large:IREG_HC1o:cal:0.90\endcsname{90.9}
\expandafter\def\csname gz:hte:ba:ln15_large:IREG_HC1o:cal:0.95\endcsname{95.2}
\expandafter\def\csname gz:hte:ba:ln15_large:IREG_HC1o:cal:0.98\endcsname{97.7}
\expandafter\def\csname gz:hte:ba:ln15_large:IREG_HC1o:cal:0.99\endcsname{98.5}
\expandafter\def\csname gz:hte:ba:ln15_large:IREG_HC1o:cov\endcsname{95.2}
\expandafter\def\csname gz:hte:ba:ln15_large:IREG_HC1o:width\endcsname{5.32}
\expandafter\def\csname gz:hte:ba:ln15_large:IREG_HC3:cal:0.80\endcsname{53.3}
\expandafter\def\csname gz:hte:ba:ln15_large:IREG_HC3:cal:0.90\endcsname{64.5}
\expandafter\def\csname gz:hte:ba:ln15_large:IREG_HC3:cal:0.95\endcsname{72.2}
\expandafter\def\csname gz:hte:ba:ln15_large:IREG_HC3:cal:0.98\endcsname{79.9}
\expandafter\def\csname gz:hte:ba:ln15_large:IREG_HC3:cal:0.99\endcsname{83.8}
\expandafter\def\csname gz:hte:ba:ln15_large:IREG_HC3:cov\endcsname{72.2}
\expandafter\def\csname gz:hte:ba:ln15_large:IREG_HC3:width\endcsname{2.81}
\expandafter\def\csname gz:hte:ba:ln15_large:IREG_HC3c:cal:0.80\endcsname{78.1}
\expandafter\def\csname gz:hte:ba:ln15_large:IREG_HC3c:cal:0.90\endcsname{87.1}
\expandafter\def\csname gz:hte:ba:ln15_large:IREG_HC3c:cal:0.95\endcsname{91.4}
\expandafter\def\csname gz:hte:ba:ln15_large:IREG_HC3c:cal:0.98\endcsname{94.8}
\expandafter\def\csname gz:hte:ba:ln15_large:IREG_HC3c:cal:0.99\endcsname{96.3}
\expandafter\def\csname gz:hte:ba:ln15_large:IREG_HC3c:cov\endcsname{91.4}
\expandafter\def\csname gz:hte:ba:ln15_large:IREG_HC3c:width\endcsname{4.75}
\expandafter\def\csname gz:hte:ba:ln15_large:R2\endcsname{0.65}
\expandafter\def\csname gz:hte:ba:ln15_large:REG_HC1:cal:0.80\endcsname{67.8}
\expandafter\def\csname gz:hte:ba:ln15_large:REG_HC1:cal:0.90\endcsname{78.1}
\expandafter\def\csname gz:hte:ba:ln15_large:REG_HC1:cal:0.95\endcsname{84.1}
\expandafter\def\csname gz:hte:ba:ln15_large:REG_HC1:cal:0.98\endcsname{89.2}
\expandafter\def\csname gz:hte:ba:ln15_large:REG_HC1:cal:0.99\endcsname{91.8}
\expandafter\def\csname gz:hte:ba:ln15_large:REG_HC1:cov\endcsname{84.1}
\expandafter\def\csname gz:hte:ba:ln15_large:REG_HC1:width\endcsname{4.17}
\expandafter\def\csname gz:hte:ba:ln15_large:REG_HC3:cal:0.80\endcsname{71.5}
\expandafter\def\csname gz:hte:ba:ln15_large:REG_HC3:cal:0.90\endcsname{80.5}
\expandafter\def\csname gz:hte:ba:ln15_large:REG_HC3:cal:0.95\endcsname{86.0}
\expandafter\def\csname gz:hte:ba:ln15_large:REG_HC3:cal:0.98\endcsname{90.6}
\expandafter\def\csname gz:hte:ba:ln15_large:REG_HC3:cal:0.99\endcsname{93.0}
\expandafter\def\csname gz:hte:ba:ln15_large:REG_HC3:cov\endcsname{86.0}
\expandafter\def\csname gz:hte:ba:ln15_large:REG_HC3:width\endcsname{4.52}
\expandafter\def\csname gz:hte:ba:ln15_large:gstar\endcsname{3.00}
\expandafter\def\csname gz:hte:ba:ln15_large:lev1\endcsname{0.0}
\expandafter\def\csname gz:hte:ba:ln15_large:mse\endcsname{1.11}
\expandafter\def\csname gz:hte:ba:ln15_large:plugratio\endcsname{1.06}
\expandafter\def\csname gz:hte:ba:ln15_large:plugratio:fmt\endcsname{1.1}
\expandafter\def\csname gz:hte:ba:ln15_large:rankdef\endcsname{0.0}
\expandafter\def\csname gz:hte:ba:ln15_large:wr1\endcsname{1.07}
\expandafter\def\csname gz:hte:ba:ln15_large:wr3\endcsname{1.05}
\expandafter\def\csname gz:hte:ba:ln15_mid:IREG_HC1:cal:0.80\endcsname{66.9}
\expandafter\def\csname gz:hte:ba:ln15_mid:IREG_HC1:cal:0.90\endcsname{78.1}
\expandafter\def\csname gz:hte:ba:ln15_mid:IREG_HC1:cal:0.95\endcsname{85.3}
\expandafter\def\csname gz:hte:ba:ln15_mid:IREG_HC1:cal:0.98\endcsname{91.3}
\expandafter\def\csname gz:hte:ba:ln15_mid:IREG_HC1:cal:0.99\endcsname{94.0}
\expandafter\def\csname gz:hte:ba:ln15_mid:IREG_HC1:cov\endcsname{85.3}
\expandafter\def\csname gz:hte:ba:ln15_mid:IREG_HC1:width\endcsname{2.42}
\expandafter\def\csname gz:hte:ba:ln15_mid:IREG_HC1c:cal:0.80\endcsname{75.8}
\expandafter\def\csname gz:hte:ba:ln15_mid:IREG_HC1c:cal:0.90\endcsname{86.6}
\expandafter\def\csname gz:hte:ba:ln15_mid:IREG_HC1c:cal:0.95\endcsname{92.0}
\expandafter\def\csname gz:hte:ba:ln15_mid:IREG_HC1c:cal:0.98\endcsname{95.8}
\expandafter\def\csname gz:hte:ba:ln15_mid:IREG_HC1c:cal:0.99\endcsname{97.3}
\expandafter\def\csname gz:hte:ba:ln15_mid:IREG_HC1c:cov\endcsname{92.0}
\expandafter\def\csname gz:hte:ba:ln15_mid:IREG_HC1c:width\endcsname{2.81}
\expandafter\def\csname gz:hte:ba:ln15_mid:IREG_HC1o:cal:0.80\endcsname{75.3}
\expandafter\def\csname gz:hte:ba:ln15_mid:IREG_HC1o:cal:0.90\endcsname{86.0}
\expandafter\def\csname gz:hte:ba:ln15_mid:IREG_HC1o:cal:0.95\endcsname{91.9}
\expandafter\def\csname gz:hte:ba:ln15_mid:IREG_HC1o:cal:0.98\endcsname{96.0}
\expandafter\def\csname gz:hte:ba:ln15_mid:IREG_HC1o:cal:0.99\endcsname{97.6}
\expandafter\def\csname gz:hte:ba:ln15_mid:IREG_HC1o:cov\endcsname{91.9}
\expandafter\def\csname gz:hte:ba:ln15_mid:IREG_HC1o:width\endcsname{2.89}
\expandafter\def\csname gz:hte:ba:ln15_mid:IREG_HC3:cal:0.80\endcsname{74.9}
\expandafter\def\csname gz:hte:ba:ln15_mid:IREG_HC3:cal:0.90\endcsname{85.3}
\expandafter\def\csname gz:hte:ba:ln15_mid:IREG_HC3:cal:0.95\endcsname{91.2}
\expandafter\def\csname gz:hte:ba:ln15_mid:IREG_HC3:cal:0.98\endcsname{95.5}
\expandafter\def\csname gz:hte:ba:ln15_mid:IREG_HC3:cal:0.99\endcsname{97.2}
\expandafter\def\csname gz:hte:ba:ln15_mid:IREG_HC3:cov\endcsname{91.2}
\expandafter\def\csname gz:hte:ba:ln15_mid:IREG_HC3:width\endcsname{2.80}
\expandafter\def\csname gz:hte:ba:ln15_mid:IREG_HC3c:cal:0.80\endcsname{81.4}
\expandafter\def\csname gz:hte:ba:ln15_mid:IREG_HC3c:cal:0.90\endcsname{90.3}
\expandafter\def\csname gz:hte:ba:ln15_mid:IREG_HC3c:cal:0.95\endcsname{94.9}
\expandafter\def\csname gz:hte:ba:ln15_mid:IREG_HC3c:cal:0.98\endcsname{97.5}
\expandafter\def\csname gz:hte:ba:ln15_mid:IREG_HC3c:cal:0.99\endcsname{98.5}
\expandafter\def\csname gz:hte:ba:ln15_mid:IREG_HC3c:cov\endcsname{94.9}
\expandafter\def\csname gz:hte:ba:ln15_mid:IREG_HC3c:width\endcsname{3.16}
\expandafter\def\csname gz:hte:ba:ln15_mid:R2\endcsname{0.17}
\expandafter\def\csname gz:hte:ba:ln15_mid:REG_HC1:cal:0.80\endcsname{73.3}
\expandafter\def\csname gz:hte:ba:ln15_mid:REG_HC1:cal:0.90\endcsname{84.3}
\expandafter\def\csname gz:hte:ba:ln15_mid:REG_HC1:cal:0.95\endcsname{90.2}
\expandafter\def\csname gz:hte:ba:ln15_mid:REG_HC1:cal:0.98\endcsname{94.5}
\expandafter\def\csname gz:hte:ba:ln15_mid:REG_HC1:cal:0.99\endcsname{96.4}
\expandafter\def\csname gz:hte:ba:ln15_mid:REG_HC1:cov\endcsname{90.2}
\expandafter\def\csname gz:hte:ba:ln15_mid:REG_HC1:width\endcsname{2.64}
\expandafter\def\csname gz:hte:ba:ln15_mid:REG_HC3:cal:0.80\endcsname{77.1}
\expandafter\def\csname gz:hte:ba:ln15_mid:REG_HC3:cal:0.90\endcsname{86.7}
\expandafter\def\csname gz:hte:ba:ln15_mid:REG_HC3:cal:0.95\endcsname{92.0}
\expandafter\def\csname gz:hte:ba:ln15_mid:REG_HC3:cal:0.98\endcsname{95.5}
\expandafter\def\csname gz:hte:ba:ln15_mid:REG_HC3:cal:0.99\endcsname{97.1}
\expandafter\def\csname gz:hte:ba:ln15_mid:REG_HC3:cov\endcsname{92.0}
\expandafter\def\csname gz:hte:ba:ln15_mid:REG_HC3:width\endcsname{2.81}
\expandafter\def\csname gz:hte:ba:ln15_mid:gstar\endcsname{1.00}
\expandafter\def\csname gz:hte:ba:ln15_mid:lev1\endcsname{0.0}
\expandafter\def\csname gz:hte:ba:ln15_mid:mse\endcsname{1.19}
\expandafter\def\csname gz:hte:ba:ln15_mid:plugratio\endcsname{1.70}
\expandafter\def\csname gz:hte:ba:ln15_mid:plugratio:fmt\endcsname{1.7}
\expandafter\def\csname gz:hte:ba:ln15_mid:rankdef\endcsname{0.0}
\expandafter\def\csname gz:hte:ba:ln15_mid:wr1\endcsname{1.06}
\expandafter\def\csname gz:hte:ba:ln15_mid:wr3\endcsname{1.12}
\expandafter\def\csname gz:hte:ba:ln15_small:IREG_HC1:cal:0.80\endcsname{72.5}
\expandafter\def\csname gz:hte:ba:ln15_small:IREG_HC1:cal:0.90\endcsname{83.7}
\expandafter\def\csname gz:hte:ba:ln15_small:IREG_HC1:cal:0.95\endcsname{90.3}
\expandafter\def\csname gz:hte:ba:ln15_small:IREG_HC1:cal:0.98\endcsname{94.9}
\expandafter\def\csname gz:hte:ba:ln15_small:IREG_HC1:cal:0.99\endcsname{96.8}
\expandafter\def\csname gz:hte:ba:ln15_small:IREG_HC1:cov\endcsname{90.3}
\expandafter\def\csname gz:hte:ba:ln15_small:IREG_HC1:width\endcsname{2.42}
\expandafter\def\csname gz:hte:ba:ln15_small:IREG_HC1c:cal:0.80\endcsname{77.0}
\expandafter\def\csname gz:hte:ba:ln15_small:IREG_HC1c:cal:0.90\endcsname{88.9}
\expandafter\def\csname gz:hte:ba:ln15_small:IREG_HC1c:cal:0.95\endcsname{94.7}
\expandafter\def\csname gz:hte:ba:ln15_small:IREG_HC1c:cal:0.98\endcsname{98.2}
\expandafter\def\csname gz:hte:ba:ln15_small:IREG_HC1c:cal:0.99\endcsname{99.3}
\expandafter\def\csname gz:hte:ba:ln15_small:IREG_HC1c:cov\endcsname{94.7}
\expandafter\def\csname gz:hte:ba:ln15_small:IREG_HC1c:width\endcsname{2.55}
\expandafter\def\csname gz:hte:ba:ln15_small:IREG_HC1o:cal:0.80\endcsname{72.6}
\expandafter\def\csname gz:hte:ba:ln15_small:IREG_HC1o:cal:0.90\endcsname{83.7}
\expandafter\def\csname gz:hte:ba:ln15_small:IREG_HC1o:cal:0.95\endcsname{90.4}
\expandafter\def\csname gz:hte:ba:ln15_small:IREG_HC1o:cal:0.98\endcsname{94.9}
\expandafter\def\csname gz:hte:ba:ln15_small:IREG_HC1o:cal:0.99\endcsname{96.8}
\expandafter\def\csname gz:hte:ba:ln15_small:IREG_HC1o:cov\endcsname{90.4}
\expandafter\def\csname gz:hte:ba:ln15_small:IREG_HC1o:width\endcsname{2.42}
\expandafter\def\csname gz:hte:ba:ln15_small:IREG_HC3:cal:0.80\endcsname{80.7}
\expandafter\def\csname gz:hte:ba:ln15_small:IREG_HC3:cal:0.90\endcsname{90.3}
\expandafter\def\csname gz:hte:ba:ln15_small:IREG_HC3:cal:0.95\endcsname{94.9}
\expandafter\def\csname gz:hte:ba:ln15_small:IREG_HC3:cal:0.98\endcsname{97.8}
\expandafter\def\csname gz:hte:ba:ln15_small:IREG_HC3:cal:0.99\endcsname{98.8}
\expandafter\def\csname gz:hte:ba:ln15_small:IREG_HC3:cov\endcsname{94.9}
\expandafter\def\csname gz:hte:ba:ln15_small:IREG_HC3:width\endcsname{2.80}
\expandafter\def\csname gz:hte:ba:ln15_small:IREG_HC3c:cal:0.80\endcsname{83.9}
\expandafter\def\csname gz:hte:ba:ln15_small:IREG_HC3c:cal:0.90\endcsname{93.1}
\expandafter\def\csname gz:hte:ba:ln15_small:IREG_HC3c:cal:0.95\endcsname{97.0}
\expandafter\def\csname gz:hte:ba:ln15_small:IREG_HC3c:cal:0.98\endcsname{99.2}
\expandafter\def\csname gz:hte:ba:ln15_small:IREG_HC3c:cal:0.99\endcsname{99.6}
\expandafter\def\csname gz:hte:ba:ln15_small:IREG_HC3c:cov\endcsname{97.0}
\expandafter\def\csname gz:hte:ba:ln15_small:IREG_HC3c:width\endcsname{2.91}
\expandafter\def\csname gz:hte:ba:ln15_small:R2\endcsname{0.002}
\expandafter\def\csname gz:hte:ba:ln15_small:REG_HC1:cal:0.80\endcsname{77.8}
\expandafter\def\csname gz:hte:ba:ln15_small:REG_HC1:cal:0.90\endcsname{89.3}
\expandafter\def\csname gz:hte:ba:ln15_small:REG_HC1:cal:0.95\endcsname{94.9}
\expandafter\def\csname gz:hte:ba:ln15_small:REG_HC1:cal:0.98\endcsname{98.2}
\expandafter\def\csname gz:hte:ba:ln15_small:REG_HC1:cal:0.99\endcsname{99.3}
\expandafter\def\csname gz:hte:ba:ln15_small:REG_HC1:cov\endcsname{94.9}
\expandafter\def\csname gz:hte:ba:ln15_small:REG_HC1:width\endcsname{2.40}
\expandafter\def\csname gz:hte:ba:ln15_small:REG_HC3:cal:0.80\endcsname{82.5}
\expandafter\def\csname gz:hte:ba:ln15_small:REG_HC3:cal:0.90\endcsname{92.2}
\expandafter\def\csname gz:hte:ba:ln15_small:REG_HC3:cal:0.95\endcsname{96.7}
\expandafter\def\csname gz:hte:ba:ln15_small:REG_HC3:cal:0.98\endcsname{99.0}
\expandafter\def\csname gz:hte:ba:ln15_small:REG_HC3:cal:0.99\endcsname{99.6}
\expandafter\def\csname gz:hte:ba:ln15_small:REG_HC3:cov\endcsname{96.7}
\expandafter\def\csname gz:hte:ba:ln15_small:REG_HC3:width\endcsname{2.55}
\expandafter\def\csname gz:hte:ba:ln15_small:gstar\endcsname{0.10}
\expandafter\def\csname gz:hte:ba:ln15_small:lev1\endcsname{0.0}
\expandafter\def\csname gz:hte:ba:ln15_small:mse\endcsname{1.24}
\expandafter\def\csname gz:hte:ba:ln15_small:plugratio\endcsname{71.01}
\expandafter\def\csname gz:hte:ba:ln15_small:plugratio:fmt\endcsname{71.0}
\expandafter\def\csname gz:hte:ba:ln15_small:rankdef\endcsname{0.0}
\expandafter\def\csname gz:hte:ba:ln15_small:wr1\endcsname{1.06}
\expandafter\def\csname gz:hte:ba:ln15_small:wr3\endcsname{1.14}
\expandafter\def\csname gz:hte:ba:ln20_small:IREG_HC1:cal:0.80\endcsname{64.9}
\expandafter\def\csname gz:hte:ba:ln20_small:IREG_HC1:cal:0.90\endcsname{75.7}
\expandafter\def\csname gz:hte:ba:ln20_small:IREG_HC1:cal:0.95\endcsname{83.1}
\expandafter\def\csname gz:hte:ba:ln20_small:IREG_HC1:cal:0.98\endcsname{89.0}
\expandafter\def\csname gz:hte:ba:ln20_small:IREG_HC1:cal:0.99\endcsname{91.8}
\expandafter\def\csname gz:hte:ba:ln20_small:IREG_HC1:cov\endcsname{83.1}
\expandafter\def\csname gz:hte:ba:ln20_small:IREG_HC1:width\endcsname{6.84}
\expandafter\def\csname gz:hte:ba:ln20_small:IREG_HC1c:cal:0.80\endcsname{75.1}
\expandafter\def\csname gz:hte:ba:ln20_small:IREG_HC1c:cal:0.90\endcsname{88.2}
\expandafter\def\csname gz:hte:ba:ln20_small:IREG_HC1c:cal:0.95\endcsname{94.8}
\expandafter\def\csname gz:hte:ba:ln20_small:IREG_HC1c:cal:0.98\endcsname{98.3}
\expandafter\def\csname gz:hte:ba:ln20_small:IREG_HC1c:cal:0.99\endcsname{99.3}
\expandafter\def\csname gz:hte:ba:ln20_small:IREG_HC1c:cov\endcsname{94.8}
\expandafter\def\csname gz:hte:ba:ln20_small:IREG_HC1c:width\endcsname{7.76}
\expandafter\def\csname gz:hte:ba:ln20_small:IREG_HC1o:cal:0.80\endcsname{65.3}
\expandafter\def\csname gz:hte:ba:ln20_small:IREG_HC1o:cal:0.90\endcsname{76.2}
\expandafter\def\csname gz:hte:ba:ln20_small:IREG_HC1o:cal:0.95\endcsname{83.5}
\expandafter\def\csname gz:hte:ba:ln20_small:IREG_HC1o:cal:0.98\endcsname{89.3}
\expandafter\def\csname gz:hte:ba:ln20_small:IREG_HC1o:cal:0.99\endcsname{92.1}
\expandafter\def\csname gz:hte:ba:ln20_small:IREG_HC1o:cov\endcsname{83.5}
\expandafter\def\csname gz:hte:ba:ln20_small:IREG_HC1o:width\endcsname{6.91}
\expandafter\def\csname gz:hte:ba:ln20_small:IREG_HC3:cal:0.80\endcsname{79.8}
\expandafter\def\csname gz:hte:ba:ln20_small:IREG_HC3:cal:0.90\endcsname{89.1}
\expandafter\def\csname gz:hte:ba:ln20_small:IREG_HC3:cal:0.95\endcsname{93.8}
\expandafter\def\csname gz:hte:ba:ln20_small:IREG_HC3:cal:0.98\endcsname{96.8}
\expandafter\def\csname gz:hte:ba:ln20_small:IREG_HC3:cal:0.99\endcsname{97.9}
\expandafter\def\csname gz:hte:ba:ln20_small:IREG_HC3:cov\endcsname{93.8}
\expandafter\def\csname gz:hte:ba:ln20_small:IREG_HC3:width\endcsname{9.11}
\expandafter\def\csname gz:hte:ba:ln20_small:IREG_HC3c:cal:0.80\endcsname{85.8}
\expandafter\def\csname gz:hte:ba:ln20_small:IREG_HC3c:cal:0.90\endcsname{94.4}
\expandafter\def\csname gz:hte:ba:ln20_small:IREG_HC3c:cal:0.95\endcsname{97.8}
\expandafter\def\csname gz:hte:ba:ln20_small:IREG_HC3c:cal:0.98\endcsname{99.4}
\expandafter\def\csname gz:hte:ba:ln20_small:IREG_HC3c:cal:0.99\endcsname{99.8}
\expandafter\def\csname gz:hte:ba:ln20_small:IREG_HC3c:cov\endcsname{97.8}
\expandafter\def\csname gz:hte:ba:ln20_small:IREG_HC3c:width\endcsname{9.84}
\expandafter\def\csname gz:hte:ba:ln20_small:R2\endcsname{0.002}
\expandafter\def\csname gz:hte:ba:ln20_small:REG_HC1:cal:0.80\endcsname{76.1}
\expandafter\def\csname gz:hte:ba:ln20_small:REG_HC1:cal:0.90\endcsname{88.7}
\expandafter\def\csname gz:hte:ba:ln20_small:REG_HC1:cal:0.95\endcsname{95.0}
\expandafter\def\csname gz:hte:ba:ln20_small:REG_HC1:cal:0.98\endcsname{98.3}
\expandafter\def\csname gz:hte:ba:ln20_small:REG_HC1:cal:0.99\endcsname{99.3}
\expandafter\def\csname gz:hte:ba:ln20_small:REG_HC1:cov\endcsname{95.0}
\expandafter\def\csname gz:hte:ba:ln20_small:REG_HC1:width\endcsname{6.80}
\expandafter\def\csname gz:hte:ba:ln20_small:REG_HC3:cal:0.80\endcsname{84.2}
\expandafter\def\csname gz:hte:ba:ln20_small:REG_HC3:cal:0.90\endcsname{93.7}
\expandafter\def\csname gz:hte:ba:ln20_small:REG_HC3:cal:0.95\endcsname{97.6}
\expandafter\def\csname gz:hte:ba:ln20_small:REG_HC3:cal:0.98\endcsname{99.3}
\expandafter\def\csname gz:hte:ba:ln20_small:REG_HC3:cal:0.99\endcsname{99.7}
\expandafter\def\csname gz:hte:ba:ln20_small:REG_HC3:cov\endcsname{97.6}
\expandafter\def\csname gz:hte:ba:ln20_small:REG_HC3:width\endcsname{7.56}
\expandafter\def\csname gz:hte:ba:ln20_small:gstar\endcsname{0.10}
\expandafter\def\csname gz:hte:ba:ln20_small:lev1\endcsname{0.0}
\expandafter\def\csname gz:hte:ba:ln20_small:mse\endcsname{1.95}
\expandafter\def\csname gz:hte:ba:ln20_small:plugratio\endcsname{90.25}
\expandafter\def\csname gz:hte:ba:ln20_small:plugratio:fmt\endcsname{90.2}
\expandafter\def\csname gz:hte:ba:ln20_small:rankdef\endcsname{0.0}
\expandafter\def\csname gz:hte:ba:ln20_small:wr1\endcsname{1.14}
\expandafter\def\csname gz:hte:ba:ln20_small:wr3\endcsname{1.30}
\expandafter\def\csname gz:hte:ba:misspec_null:IREG_HC1:cal:0.80\endcsname{79.5}
\expandafter\def\csname gz:hte:ba:misspec_null:IREG_HC1:cal:0.90\endcsname{89.7}
\expandafter\def\csname gz:hte:ba:misspec_null:IREG_HC1:cal:0.95\endcsname{94.8}
\expandafter\def\csname gz:hte:ba:misspec_null:IREG_HC1:cal:0.98\endcsname{97.9}
\expandafter\def\csname gz:hte:ba:misspec_null:IREG_HC1:cal:0.99\endcsname{98.9}
\expandafter\def\csname gz:hte:ba:misspec_null:IREG_HC1:cov\endcsname{94.8}
\expandafter\def\csname gz:hte:ba:misspec_null:IREG_HC1:width\endcsname{0.38}
\expandafter\def\csname gz:hte:ba:misspec_null:IREG_HC1c:cal:0.80\endcsname{79.8}
\expandafter\def\csname gz:hte:ba:misspec_null:IREG_HC1c:cal:0.90\endcsname{89.9}
\expandafter\def\csname gz:hte:ba:misspec_null:IREG_HC1c:cal:0.95\endcsname{94.9}
\expandafter\def\csname gz:hte:ba:misspec_null:IREG_HC1c:cal:0.98\endcsname{98.0}
\expandafter\def\csname gz:hte:ba:misspec_null:IREG_HC1c:cal:0.99\endcsname{99.0}
\expandafter\def\csname gz:hte:ba:misspec_null:IREG_HC1c:cov\endcsname{94.9}
\expandafter\def\csname gz:hte:ba:misspec_null:IREG_HC1c:width\endcsname{0.38}
\expandafter\def\csname gz:hte:ba:misspec_null:IREG_HC1o:cal:0.80\endcsname{79.5}
\expandafter\def\csname gz:hte:ba:misspec_null:IREG_HC1o:cal:0.90\endcsname{89.7}
\expandafter\def\csname gz:hte:ba:misspec_null:IREG_HC1o:cal:0.95\endcsname{94.8}
\expandafter\def\csname gz:hte:ba:misspec_null:IREG_HC1o:cal:0.98\endcsname{97.9}
\expandafter\def\csname gz:hte:ba:misspec_null:IREG_HC1o:cal:0.99\endcsname{98.9}
\expandafter\def\csname gz:hte:ba:misspec_null:IREG_HC1o:cov\endcsname{94.8}
\expandafter\def\csname gz:hte:ba:misspec_null:IREG_HC1o:width\endcsname{0.38}
\expandafter\def\csname gz:hte:ba:misspec_null:IREG_HC3:cal:0.80\endcsname{80.1}
\expandafter\def\csname gz:hte:ba:misspec_null:IREG_HC3:cal:0.90\endcsname{90.2}
\expandafter\def\csname gz:hte:ba:misspec_null:IREG_HC3:cal:0.95\endcsname{95.1}
\expandafter\def\csname gz:hte:ba:misspec_null:IREG_HC3:cal:0.98\endcsname{98.1}
\expandafter\def\csname gz:hte:ba:misspec_null:IREG_HC3:cal:0.99\endcsname{99.1}
\expandafter\def\csname gz:hte:ba:misspec_null:IREG_HC3:cov\endcsname{95.1}
\expandafter\def\csname gz:hte:ba:misspec_null:IREG_HC3:width\endcsname{0.39}
\expandafter\def\csname gz:hte:ba:misspec_null:IREG_HC3c:cal:0.80\endcsname{80.4}
\expandafter\def\csname gz:hte:ba:misspec_null:IREG_HC3c:cal:0.90\endcsname{90.4}
\expandafter\def\csname gz:hte:ba:misspec_null:IREG_HC3c:cal:0.95\endcsname{95.3}
\expandafter\def\csname gz:hte:ba:misspec_null:IREG_HC3c:cal:0.98\endcsname{98.1}
\expandafter\def\csname gz:hte:ba:misspec_null:IREG_HC3c:cal:0.99\endcsname{99.1}
\expandafter\def\csname gz:hte:ba:misspec_null:IREG_HC3c:cov\endcsname{95.3}
\expandafter\def\csname gz:hte:ba:misspec_null:IREG_HC3c:width\endcsname{0.39}
\expandafter\def\csname gz:hte:ba:misspec_null:R2\endcsname{0.000}
\expandafter\def\csname gz:hte:ba:misspec_null:REG_HC1:cal:0.80\endcsname{79.6}
\expandafter\def\csname gz:hte:ba:misspec_null:REG_HC1:cal:0.90\endcsname{89.7}
\expandafter\def\csname gz:hte:ba:misspec_null:REG_HC1:cal:0.95\endcsname{94.8}
\expandafter\def\csname gz:hte:ba:misspec_null:REG_HC1:cal:0.98\endcsname{97.9}
\expandafter\def\csname gz:hte:ba:misspec_null:REG_HC1:cal:0.99\endcsname{98.9}
\expandafter\def\csname gz:hte:ba:misspec_null:REG_HC1:cov\endcsname{94.8}
\expandafter\def\csname gz:hte:ba:misspec_null:REG_HC1:width\endcsname{0.38}
\expandafter\def\csname gz:hte:ba:misspec_null:REG_HC3:cal:0.80\endcsname{79.9}
\expandafter\def\csname gz:hte:ba:misspec_null:REG_HC3:cal:0.90\endcsname{90.0}
\expandafter\def\csname gz:hte:ba:misspec_null:REG_HC3:cal:0.95\endcsname{95.0}
\expandafter\def\csname gz:hte:ba:misspec_null:REG_HC3:cal:0.98\endcsname{98.0}
\expandafter\def\csname gz:hte:ba:misspec_null:REG_HC3:cal:0.99\endcsname{99.0}
\expandafter\def\csname gz:hte:ba:misspec_null:REG_HC3:cov\endcsname{95.0}
\expandafter\def\csname gz:hte:ba:misspec_null:REG_HC3:width\endcsname{0.38}
\expandafter\def\csname gz:hte:ba:misspec_null:gstar\endcsname{0.00}
\expandafter\def\csname gz:hte:ba:misspec_null:lev1\endcsname{0.0}
\expandafter\def\csname gz:hte:ba:misspec_null:mse\endcsname{1.00}
\expandafter\def\csname gz:hte:ba:misspec_null:plugratio:fmt\endcsname{--}
\expandafter\def\csname gz:hte:ba:misspec_null:rankdef\endcsname{0.0}
\expandafter\def\csname gz:hte:ba:misspec_null:wr1\endcsname{1.01}
\expandafter\def\csname gz:hte:ba:misspec_null:wr3\endcsname{1.01}
\expandafter\def\csname gz:hte:ba:mse:max\endcsname{1.95}
\expandafter\def\csname gz:hte:ba:mse:min\endcsname{0.97}
\expandafter\def\csname gz:hte:ba:ndesigns\endcsname{14}
\expandafter\def\csname gz:hte:ba:rare:IREG_HC1:cal:0.80\endcsname{73.6}
\expandafter\def\csname gz:hte:ba:rare:IREG_HC1:cal:0.90\endcsname{84.8}
\expandafter\def\csname gz:hte:ba:rare:IREG_HC1:cal:0.95\endcsname{90.8}
\expandafter\def\csname gz:hte:ba:rare:IREG_HC1:cal:0.98\endcsname{95.1}
\expandafter\def\csname gz:hte:ba:rare:IREG_HC1:cal:0.99\endcsname{96.7}
\expandafter\def\csname gz:hte:ba:rare:IREG_HC1:cov\endcsname{90.8}
\expandafter\def\csname gz:hte:ba:rare:IREG_HC1:width\endcsname{0.35}
\expandafter\def\csname gz:hte:ba:rare:IREG_HC1c:cal:0.80\endcsname{78.8}
\expandafter\def\csname gz:hte:ba:rare:IREG_HC1c:cal:0.90\endcsname{89.1}
\expandafter\def\csname gz:hte:ba:rare:IREG_HC1c:cal:0.95\endcsname{94.0}
\expandafter\def\csname gz:hte:ba:rare:IREG_HC1c:cal:0.98\endcsname{96.8}
\expandafter\def\csname gz:hte:ba:rare:IREG_HC1c:cal:0.99\endcsname{97.7}
\expandafter\def\csname gz:hte:ba:rare:IREG_HC1c:cov\endcsname{94.0}
\expandafter\def\csname gz:hte:ba:rare:IREG_HC1c:width\endcsname{0.40}
\expandafter\def\csname gz:hte:ba:rare:IREG_HC1o:cal:0.80\endcsname{78.6}
\expandafter\def\csname gz:hte:ba:rare:IREG_HC1o:cal:0.90\endcsname{89.0}
\expandafter\def\csname gz:hte:ba:rare:IREG_HC1o:cal:0.95\endcsname{93.8}
\expandafter\def\csname gz:hte:ba:rare:IREG_HC1o:cal:0.98\endcsname{96.8}
\expandafter\def\csname gz:hte:ba:rare:IREG_HC1o:cal:0.99\endcsname{97.8}
\expandafter\def\csname gz:hte:ba:rare:IREG_HC1o:cov\endcsname{93.8}
\expandafter\def\csname gz:hte:ba:rare:IREG_HC1o:width\endcsname{0.39}
\expandafter\def\csname gz:hte:ba:rare:IREG_HC3:cal:0.80\endcsname{74.3}
\expandafter\def\csname gz:hte:ba:rare:IREG_HC3:cal:0.90\endcsname{85.3}
\expandafter\def\csname gz:hte:ba:rare:IREG_HC3:cal:0.95\endcsname{91.2}
\expandafter\def\csname gz:hte:ba:rare:IREG_HC3:cal:0.98\endcsname{95.3}
\expandafter\def\csname gz:hte:ba:rare:IREG_HC3:cal:0.99\endcsname{96.9}
\expandafter\def\csname gz:hte:ba:rare:IREG_HC3:cov\endcsname{91.2}
\expandafter\def\csname gz:hte:ba:rare:IREG_HC3:width\endcsname{0.36}
\expandafter\def\csname gz:hte:ba:rare:IREG_HC3c:cal:0.80\endcsname{79.3}
\expandafter\def\csname gz:hte:ba:rare:IREG_HC3c:cal:0.90\endcsname{89.5}
\expandafter\def\csname gz:hte:ba:rare:IREG_HC3c:cal:0.95\endcsname{94.2}
\expandafter\def\csname gz:hte:ba:rare:IREG_HC3c:cal:0.98\endcsname{96.9}
\expandafter\def\csname gz:hte:ba:rare:IREG_HC3c:cal:0.99\endcsname{97.8}
\expandafter\def\csname gz:hte:ba:rare:IREG_HC3c:cov\endcsname{94.2}
\expandafter\def\csname gz:hte:ba:rare:IREG_HC3c:width\endcsname{0.40}
\expandafter\def\csname gz:hte:ba:rare:R2\endcsname{0.20}
\expandafter\def\csname gz:hte:ba:rare:REG_HC1:cal:0.80\endcsname{78.9}
\expandafter\def\csname gz:hte:ba:rare:REG_HC1:cal:0.90\endcsname{89.3}
\expandafter\def\csname gz:hte:ba:rare:REG_HC1:cal:0.95\endcsname{94.7}
\expandafter\def\csname gz:hte:ba:rare:REG_HC1:cal:0.98\endcsname{97.7}
\expandafter\def\csname gz:hte:ba:rare:REG_HC1:cal:0.99\endcsname{98.8}
\expandafter\def\csname gz:hte:ba:rare:REG_HC1:cov\endcsname{94.7}
\expandafter\def\csname gz:hte:ba:rare:REG_HC1:width\endcsname{0.40}
\expandafter\def\csname gz:hte:ba:rare:REG_HC3:cal:0.80\endcsname{80.4}
\expandafter\def\csname gz:hte:ba:rare:REG_HC3:cal:0.90\endcsname{90.3}
\expandafter\def\csname gz:hte:ba:rare:REG_HC3:cal:0.95\endcsname{95.5}
\expandafter\def\csname gz:hte:ba:rare:REG_HC3:cal:0.98\endcsname{98.1}
\expandafter\def\csname gz:hte:ba:rare:REG_HC3:cal:0.99\endcsname{99.0}
\expandafter\def\csname gz:hte:ba:rare:REG_HC3:cov\endcsname{95.5}
\expandafter\def\csname gz:hte:ba:rare:REG_HC3:width\endcsname{0.41}
\expandafter\def\csname gz:hte:ba:rare:gstar\endcsname{1.00}
\expandafter\def\csname gz:hte:ba:rare:lev1\endcsname{6.2}
\expandafter\def\csname gz:hte:ba:rare:mse\endcsname{0.97}
\expandafter\def\csname gz:hte:ba:rare:plugratio\endcsname{1.04}
\expandafter\def\csname gz:hte:ba:rare:plugratio:fmt\endcsname{1.0}
\expandafter\def\csname gz:hte:ba:rare:rankdef\endcsname{1.3}
\expandafter\def\csname gz:hte:ba:rare:wr1\endcsname{0.99}
\expandafter\def\csname gz:hte:ba:rare:wr3\endcsname{0.97}
\expandafter\def\csname gz:hte:ba:skew_null:IREG_HC1:cal:0.80\endcsname{79.5}
\expandafter\def\csname gz:hte:ba:skew_null:IREG_HC1:cal:0.90\endcsname{89.6}
\expandafter\def\csname gz:hte:ba:skew_null:IREG_HC1:cal:0.95\endcsname{94.6}
\expandafter\def\csname gz:hte:ba:skew_null:IREG_HC1:cal:0.98\endcsname{97.8}
\expandafter\def\csname gz:hte:ba:skew_null:IREG_HC1:cal:0.99\endcsname{98.8}
\expandafter\def\csname gz:hte:ba:skew_null:IREG_HC1:cov\endcsname{94.6}
\expandafter\def\csname gz:hte:ba:skew_null:IREG_HC1:width\endcsname{0.64}
\expandafter\def\csname gz:hte:ba:skew_null:IREG_HC1c:cal:0.80\endcsname{80.1}
\expandafter\def\csname gz:hte:ba:skew_null:IREG_HC1c:cal:0.90\endcsname{90.2}
\expandafter\def\csname gz:hte:ba:skew_null:IREG_HC1c:cal:0.95\endcsname{95.0}
\expandafter\def\csname gz:hte:ba:skew_null:IREG_HC1c:cal:0.98\endcsname{98.1}
\expandafter\def\csname gz:hte:ba:skew_null:IREG_HC1c:cal:0.99\endcsname{99.0}
\expandafter\def\csname gz:hte:ba:skew_null:IREG_HC1c:cov\endcsname{95.0}
\expandafter\def\csname gz:hte:ba:skew_null:IREG_HC1c:width\endcsname{0.65}
\expandafter\def\csname gz:hte:ba:skew_null:IREG_HC1o:cal:0.80\endcsname{79.5}
\expandafter\def\csname gz:hte:ba:skew_null:IREG_HC1o:cal:0.90\endcsname{89.6}
\expandafter\def\csname gz:hte:ba:skew_null:IREG_HC1o:cal:0.95\endcsname{94.6}
\expandafter\def\csname gz:hte:ba:skew_null:IREG_HC1o:cal:0.98\endcsname{97.8}
\expandafter\def\csname gz:hte:ba:skew_null:IREG_HC1o:cal:0.99\endcsname{98.8}
\expandafter\def\csname gz:hte:ba:skew_null:IREG_HC1o:cov\endcsname{94.6}
\expandafter\def\csname gz:hte:ba:skew_null:IREG_HC1o:width\endcsname{0.64}
\expandafter\def\csname gz:hte:ba:skew_null:IREG_HC3:cal:0.80\endcsname{81.2}
\expandafter\def\csname gz:hte:ba:skew_null:IREG_HC3:cal:0.90\endcsname{90.9}
\expandafter\def\csname gz:hte:ba:skew_null:IREG_HC3:cal:0.95\endcsname{95.6}
\expandafter\def\csname gz:hte:ba:skew_null:IREG_HC3:cal:0.98\endcsname{98.3}
\expandafter\def\csname gz:hte:ba:skew_null:IREG_HC3:cal:0.99\endcsname{99.1}
\expandafter\def\csname gz:hte:ba:skew_null:IREG_HC3:cov\endcsname{95.6}
\expandafter\def\csname gz:hte:ba:skew_null:IREG_HC3:width\endcsname{0.66}
\expandafter\def\csname gz:hte:ba:skew_null:IREG_HC3c:cal:0.80\endcsname{81.8}
\expandafter\def\csname gz:hte:ba:skew_null:IREG_HC3c:cal:0.90\endcsname{91.4}
\expandafter\def\csname gz:hte:ba:skew_null:IREG_HC3c:cal:0.95\endcsname{95.9}
\expandafter\def\csname gz:hte:ba:skew_null:IREG_HC3c:cal:0.98\endcsname{98.5}
\expandafter\def\csname gz:hte:ba:skew_null:IREG_HC3c:cal:0.99\endcsname{99.2}
\expandafter\def\csname gz:hte:ba:skew_null:IREG_HC3c:cov\endcsname{95.9}
\expandafter\def\csname gz:hte:ba:skew_null:IREG_HC3c:width\endcsname{0.67}
\expandafter\def\csname gz:hte:ba:skew_null:R2\endcsname{0.000}
\expandafter\def\csname gz:hte:ba:skew_null:REG_HC1:cal:0.80\endcsname{80.0}
\expandafter\def\csname gz:hte:ba:skew_null:REG_HC1:cal:0.90\endcsname{90.1}
\expandafter\def\csname gz:hte:ba:skew_null:REG_HC1:cal:0.95\endcsname{95.0}
\expandafter\def\csname gz:hte:ba:skew_null:REG_HC1:cal:0.98\endcsname{98.0}
\expandafter\def\csname gz:hte:ba:skew_null:REG_HC1:cal:0.99\endcsname{98.9}
\expandafter\def\csname gz:hte:ba:skew_null:REG_HC1:cov\endcsname{95.0}
\expandafter\def\csname gz:hte:ba:skew_null:REG_HC1:width\endcsname{0.64}
\expandafter\def\csname gz:hte:ba:skew_null:REG_HC3:cal:0.80\endcsname{80.9}
\expandafter\def\csname gz:hte:ba:skew_null:REG_HC3:cal:0.90\endcsname{90.8}
\expandafter\def\csname gz:hte:ba:skew_null:REG_HC3:cal:0.95\endcsname{95.5}
\expandafter\def\csname gz:hte:ba:skew_null:REG_HC3:cal:0.98\endcsname{98.3}
\expandafter\def\csname gz:hte:ba:skew_null:REG_HC3:cal:0.99\endcsname{99.1}
\expandafter\def\csname gz:hte:ba:skew_null:REG_HC3:cov\endcsname{95.5}
\expandafter\def\csname gz:hte:ba:skew_null:REG_HC3:width\endcsname{0.65}
\expandafter\def\csname gz:hte:ba:skew_null:gstar\endcsname{0.00}
\expandafter\def\csname gz:hte:ba:skew_null:lev1\endcsname{0.0}
\expandafter\def\csname gz:hte:ba:skew_null:mse\endcsname{1.02}
\expandafter\def\csname gz:hte:ba:skew_null:plugratio:fmt\endcsname{--}
\expandafter\def\csname gz:hte:ba:skew_null:rankdef\endcsname{0.0}
\expandafter\def\csname gz:hte:ba:skew_null:wr1\endcsname{1.01}
\expandafter\def\csname gz:hte:ba:skew_null:wr3\endcsname{1.03}
\expandafter\def\csname gz:hte:ba:small:I1:max\endcsname{94.8}
\expandafter\def\csname gz:hte:ba:small:I1:min\endcsname{83.1}
\expandafter\def\csname gz:hte:ba:small:I1c:max\endcsname{95.0}
\expandafter\def\csname gz:hte:ba:small:I1c:min\endcsname{93.7}
\expandafter\def\csname gz:hte:ba:small:I3:max\endcsname{95.6}
\expandafter\def\csname gz:hte:ba:small:I3:min\endcsname{93.8}
\expandafter\def\csname gz:hte:ba:small:I3c:max\endcsname{97.8}
\expandafter\def\csname gz:hte:ba:small:I3c:min\endcsname{95.2}
\expandafter\def\csname gz:hte:ba:small:I3cal80:max\endcsname{80.7}
\expandafter\def\csname gz:hte:ba:small:I3cal80:min\endcsname{79.0}
\expandafter\def\csname gz:hte:ba:small:R1:max\endcsname{95.0}
\expandafter\def\csname gz:hte:ba:small:R1:min\endcsname{93.2}
\expandafter\def\csname gz:hte:ba:small:plug:max\endcsname{90}
\expandafter\def\csname gz:hte:ba:small:plug:min\endcsname{1}
\expandafter\def\csname gz:hte:ba:small:wr1:max\endcsname{1.14}
\expandafter\def\csname gz:hte:ba:small:wr3:max\endcsname{1.30}
\expandafter\def\csname gz:hte:ba:smallng:plug:max\endcsname{90}
\expandafter\def\csname gz:hte:ba:smallng:plug:min\endcsname{4}
\expandafter\def\csname gz:hte:mcse:max\endcsname{0.35}
\expandafter\def\csname gz:hte:reps\endcsname{20{,}000}
\expandafter\def\csname gz:hte:sw:R2_0.00:IREG_HC1:cov\endcsname{94.9}
\expandafter\def\csname gz:hte:sw:R2_0.00:IREG_HC1:width\endcsname{0.35}
\expandafter\def\csname gz:hte:sw:R2_0.00:IREG_HC1c:cov\endcsname{95.0}
\expandafter\def\csname gz:hte:sw:R2_0.00:IREG_HC1c:width\endcsname{0.35}
\expandafter\def\csname gz:hte:sw:R2_0.00:IREG_HC1o:cov\endcsname{94.9}
\expandafter\def\csname gz:hte:sw:R2_0.00:IREG_HC1o:width\endcsname{0.35}
\expandafter\def\csname gz:hte:sw:R2_0.00:IREG_HC3:cov\endcsname{95.1}
\expandafter\def\csname gz:hte:sw:R2_0.00:IREG_HC3:width\endcsname{0.36}
\expandafter\def\csname gz:hte:sw:R2_0.00:IREG_HC3c:cov\endcsname{95.3}
\expandafter\def\csname gz:hte:sw:R2_0.00:IREG_HC3c:width\endcsname{0.36}
\expandafter\def\csname gz:hte:sw:R2_0.00:R2\endcsname{0.000}
\expandafter\def\csname gz:hte:sw:R2_0.00:REG_HC1:cov\endcsname{94.9}
\expandafter\def\csname gz:hte:sw:R2_0.00:REG_HC1:width\endcsname{0.35}
\expandafter\def\csname gz:hte:sw:R2_0.00:REG_HC3:cov\endcsname{95.0}
\expandafter\def\csname gz:hte:sw:R2_0.00:REG_HC3:width\endcsname{0.36}
\expandafter\def\csname gz:hte:sw:R2_0.00:asy\endcsname{95.0}
\expandafter\def\csname gz:hte:sw:R2_0.00:gstar\endcsname{0.00}
\expandafter\def\csname gz:hte:sw:R2_0.00:lev1\endcsname{0.0}
\expandafter\def\csname gz:hte:sw:R2_0.00:mse\endcsname{1.00}
\expandafter\def\csname gz:hte:sw:R2_0.00:rankdef\endcsname{0.0}
\expandafter\def\csname gz:hte:sw:R2_0.10:IREG_HC1:cov\endcsname{93.7}
\expandafter\def\csname gz:hte:sw:R2_0.10:IREG_HC1:width\endcsname{0.35}
\expandafter\def\csname gz:hte:sw:R2_0.10:IREG_HC1c:cov\endcsname{95.1}
\expandafter\def\csname gz:hte:sw:R2_0.10:IREG_HC1c:width\endcsname{0.37}
\expandafter\def\csname gz:hte:sw:R2_0.10:IREG_HC1o:cov\endcsname{95.0}
\expandafter\def\csname gz:hte:sw:R2_0.10:IREG_HC1o:width\endcsname{0.37}
\expandafter\def\csname gz:hte:sw:R2_0.10:IREG_HC3:cov\endcsname{94.0}
\expandafter\def\csname gz:hte:sw:R2_0.10:IREG_HC3:width\endcsname{0.36}
\expandafter\def\csname gz:hte:sw:R2_0.10:IREG_HC3c:cov\endcsname{95.3}
\expandafter\def\csname gz:hte:sw:R2_0.10:IREG_HC3c:width\endcsname{0.38}
\expandafter\def\csname gz:hte:sw:R2_0.10:R2\endcsname{0.10}
\expandafter\def\csname gz:hte:sw:R2_0.10:REG_HC1:cov\endcsname{94.9}
\expandafter\def\csname gz:hte:sw:R2_0.10:REG_HC1:width\endcsname{0.37}
\expandafter\def\csname gz:hte:sw:R2_0.10:REG_HC3:cov\endcsname{95.1}
\expandafter\def\csname gz:hte:sw:R2_0.10:REG_HC3:width\endcsname{0.37}
\expandafter\def\csname gz:hte:sw:R2_0.10:asy\endcsname{93.7}
\expandafter\def\csname gz:hte:sw:R2_0.10:gstar\endcsname{0.67}
\expandafter\def\csname gz:hte:sw:R2_0.10:lev1\endcsname{0.0}
\expandafter\def\csname gz:hte:sw:R2_0.10:mse\endcsname{1.00}
\expandafter\def\csname gz:hte:sw:R2_0.10:plugratio\endcsname{1.09}
\expandafter\def\csname gz:hte:sw:R2_0.10:rankdef\endcsname{0.0}
\expandafter\def\csname gz:hte:sw:R2_0.25:IREG_HC1:cov\endcsname{91.1}
\expandafter\def\csname gz:hte:sw:R2_0.25:IREG_HC1:width\endcsname{0.35}
\expandafter\def\csname gz:hte:sw:R2_0.25:IREG_HC1c:cov\endcsname{95.0}
\expandafter\def\csname gz:hte:sw:R2_0.25:IREG_HC1c:width\endcsname{0.41}
\expandafter\def\csname gz:hte:sw:R2_0.25:IREG_HC1o:cov\endcsname{94.9}
\expandafter\def\csname gz:hte:sw:R2_0.25:IREG_HC1o:width\endcsname{0.41}
\expandafter\def\csname gz:hte:sw:R2_0.25:IREG_HC3:cov\endcsname{91.5}
\expandafter\def\csname gz:hte:sw:R2_0.25:IREG_HC3:width\endcsname{0.36}
\expandafter\def\csname gz:hte:sw:R2_0.25:IREG_HC3c:cov\endcsname{95.1}
\expandafter\def\csname gz:hte:sw:R2_0.25:IREG_HC3c:width\endcsname{0.41}
\expandafter\def\csname gz:hte:sw:R2_0.25:R2\endcsname{0.25}
\expandafter\def\csname gz:hte:sw:R2_0.25:REG_HC1:cov\endcsname{94.9}
\expandafter\def\csname gz:hte:sw:R2_0.25:REG_HC1:width\endcsname{0.41}
\expandafter\def\csname gz:hte:sw:R2_0.25:REG_HC3:cov\endcsname{95.1}
\expandafter\def\csname gz:hte:sw:R2_0.25:REG_HC3:width\endcsname{0.41}
\expandafter\def\csname gz:hte:sw:R2_0.25:asy\endcsname{91.0}
\expandafter\def\csname gz:hte:sw:R2_0.25:gstar\endcsname{1.15}
\expandafter\def\csname gz:hte:sw:R2_0.25:lev1\endcsname{0.0}
\expandafter\def\csname gz:hte:sw:R2_0.25:mse\endcsname{1.00}
\expandafter\def\csname gz:hte:sw:R2_0.25:plugratio\endcsname{1.03}
\expandafter\def\csname gz:hte:sw:R2_0.25:rankdef\endcsname{0.0}
\expandafter\def\csname gz:hte:sw:R2_0.50:IREG_HC1:cov\endcsname{83.3}
\expandafter\def\csname gz:hte:sw:R2_0.50:IREG_HC1:width\endcsname{0.35}
\expandafter\def\csname gz:hte:sw:R2_0.50:IREG_HC1c:cov\endcsname{95.1}
\expandafter\def\csname gz:hte:sw:R2_0.50:IREG_HC1c:width\endcsname{0.50}
\expandafter\def\csname gz:hte:sw:R2_0.50:IREG_HC1o:cov\endcsname{95.1}
\expandafter\def\csname gz:hte:sw:R2_0.50:IREG_HC1o:width\endcsname{0.50}
\expandafter\def\csname gz:hte:sw:R2_0.50:IREG_HC3:cov\endcsname{83.8}
\expandafter\def\csname gz:hte:sw:R2_0.50:IREG_HC3:width\endcsname{0.36}
\expandafter\def\csname gz:hte:sw:R2_0.50:IREG_HC3c:cov\endcsname{95.2}
\expandafter\def\csname gz:hte:sw:R2_0.50:IREG_HC3c:width\endcsname{0.50}
\expandafter\def\csname gz:hte:sw:R2_0.50:R2\endcsname{0.50}
\expandafter\def\csname gz:hte:sw:R2_0.50:REG_HC1:cov\endcsname{95.1}
\expandafter\def\csname gz:hte:sw:R2_0.50:REG_HC1:width\endcsname{0.50}
\expandafter\def\csname gz:hte:sw:R2_0.50:REG_HC3:cov\endcsname{95.2}
\expandafter\def\csname gz:hte:sw:R2_0.50:REG_HC3:width\endcsname{0.50}
\expandafter\def\csname gz:hte:sw:R2_0.50:asy\endcsname{83.4}
\expandafter\def\csname gz:hte:sw:R2_0.50:gstar\endcsname{2.00}
\expandafter\def\csname gz:hte:sw:R2_0.50:lev1\endcsname{0.0}
\expandafter\def\csname gz:hte:sw:R2_0.50:mse\endcsname{0.99}
\expandafter\def\csname gz:hte:sw:R2_0.50:plugratio\endcsname{1.01}
\expandafter\def\csname gz:hte:sw:R2_0.50:rankdef\endcsname{0.0}
\expandafter\def\csname gz:hte:sw:R2_0.70:IREG_HC1:cov\endcsname{71.8}
\expandafter\def\csname gz:hte:sw:R2_0.70:IREG_HC1:width\endcsname{0.35}
\expandafter\def\csname gz:hte:sw:R2_0.70:IREG_HC1c:cov\endcsname{95.2}
\expandafter\def\csname gz:hte:sw:R2_0.70:IREG_HC1c:width\endcsname{0.64}
\expandafter\def\csname gz:hte:sw:R2_0.70:IREG_HC1o:cov\endcsname{95.2}
\expandafter\def\csname gz:hte:sw:R2_0.70:IREG_HC1o:width\endcsname{0.64}
\expandafter\def\csname gz:hte:sw:R2_0.70:IREG_HC3:cov\endcsname{72.4}
\expandafter\def\csname gz:hte:sw:R2_0.70:IREG_HC3:width\endcsname{0.36}
\expandafter\def\csname gz:hte:sw:R2_0.70:IREG_HC3c:cov\endcsname{95.2}
\expandafter\def\csname gz:hte:sw:R2_0.70:IREG_HC3c:width\endcsname{0.65}
\expandafter\def\csname gz:hte:sw:R2_0.70:R2\endcsname{0.70}
\expandafter\def\csname gz:hte:sw:R2_0.70:REG_HC1:cov\endcsname{95.2}
\expandafter\def\csname gz:hte:sw:R2_0.70:REG_HC1:width\endcsname{0.64}
\expandafter\def\csname gz:hte:sw:R2_0.70:REG_HC3:cov\endcsname{95.4}
\expandafter\def\csname gz:hte:sw:R2_0.70:REG_HC3:width\endcsname{0.65}
\expandafter\def\csname gz:hte:sw:R2_0.70:asy\endcsname{71.7}
\expandafter\def\csname gz:hte:sw:R2_0.70:gstar\endcsname{3.06}
\expandafter\def\csname gz:hte:sw:R2_0.70:lev1\endcsname{0.0}
\expandafter\def\csname gz:hte:sw:R2_0.70:mse\endcsname{0.99}
\expandafter\def\csname gz:hte:sw:R2_0.70:plugratio\endcsname{1.00}
\expandafter\def\csname gz:hte:sw:R2_0.70:rankdef\endcsname{0.0}
\expandafter\def\csname gz:hte:sw:R2_0.85:IREG_HC1:cov\endcsname{55.4}
\expandafter\def\csname gz:hte:sw:R2_0.85:IREG_HC1:width\endcsname{0.35}
\expandafter\def\csname gz:hte:sw:R2_0.85:IREG_HC1c:cov\endcsname{94.9}
\expandafter\def\csname gz:hte:sw:R2_0.85:IREG_HC1c:width\endcsname{0.91}
\expandafter\def\csname gz:hte:sw:R2_0.85:IREG_HC1o:cov\endcsname{95.0}
\expandafter\def\csname gz:hte:sw:R2_0.85:IREG_HC1o:width\endcsname{0.91}
\expandafter\def\csname gz:hte:sw:R2_0.85:IREG_HC3:cov\endcsname{55.9}
\expandafter\def\csname gz:hte:sw:R2_0.85:IREG_HC3:width\endcsname{0.36}
\expandafter\def\csname gz:hte:sw:R2_0.85:IREG_HC3c:cov\endcsname{95.0}
\expandafter\def\csname gz:hte:sw:R2_0.85:IREG_HC3c:width\endcsname{0.91}
\expandafter\def\csname gz:hte:sw:R2_0.85:R2\endcsname{0.85}
\expandafter\def\csname gz:hte:sw:R2_0.85:REG_HC1:cov\endcsname{95.0}
\expandafter\def\csname gz:hte:sw:R2_0.85:REG_HC1:width\endcsname{0.91}
\expandafter\def\csname gz:hte:sw:R2_0.85:REG_HC3:cov\endcsname{95.1}
\expandafter\def\csname gz:hte:sw:R2_0.85:REG_HC3:width\endcsname{0.92}
\expandafter\def\csname gz:hte:sw:R2_0.85:asy\endcsname{55.2}
\expandafter\def\csname gz:hte:sw:R2_0.85:gstar\endcsname{4.76}
\expandafter\def\csname gz:hte:sw:R2_0.85:lev1\endcsname{0.0}
\expandafter\def\csname gz:hte:sw:R2_0.85:mse\endcsname{0.99}
\expandafter\def\csname gz:hte:sw:R2_0.85:plugratio\endcsname{1.00}
\expandafter\def\csname gz:hte:sw:R2_0.85:rankdef\endcsname{0.0}
\expandafter\def\csname gz:hte:sw:R2_0.90:IREG_HC1:cov\endcsname{47.1}
\expandafter\def\csname gz:hte:sw:R2_0.90:IREG_HC1:width\endcsname{0.35}
\expandafter\def\csname gz:hte:sw:R2_0.90:IREG_HC1c:cov\endcsname{95.1}
\expandafter\def\csname gz:hte:sw:R2_0.90:IREG_HC1c:width\endcsname{1.11}
\expandafter\def\csname gz:hte:sw:R2_0.90:IREG_HC1o:cov\endcsname{95.2}
\expandafter\def\csname gz:hte:sw:R2_0.90:IREG_HC1o:width\endcsname{1.11}
\expandafter\def\csname gz:hte:sw:R2_0.90:IREG_HC3:cov\endcsname{47.6}
\expandafter\def\csname gz:hte:sw:R2_0.90:IREG_HC3:width\endcsname{0.36}
\expandafter\def\csname gz:hte:sw:R2_0.90:IREG_HC3c:cov\endcsname{95.1}
\expandafter\def\csname gz:hte:sw:R2_0.90:IREG_HC3c:width\endcsname{1.11}
\expandafter\def\csname gz:hte:sw:R2_0.90:R2\endcsname{0.90}
\expandafter\def\csname gz:hte:sw:R2_0.90:REG_HC1:cov\endcsname{95.1}
\expandafter\def\csname gz:hte:sw:R2_0.90:REG_HC1:width\endcsname{1.12}
\expandafter\def\csname gz:hte:sw:R2_0.90:REG_HC3:cov\endcsname{95.3}
\expandafter\def\csname gz:hte:sw:R2_0.90:REG_HC3:width\endcsname{1.13}
\expandafter\def\csname gz:hte:sw:R2_0.90:asy\endcsname{46.5}
\expandafter\def\csname gz:hte:sw:R2_0.90:gstar\endcsname{6.00}
\expandafter\def\csname gz:hte:sw:R2_0.90:lev1\endcsname{0.0}
\expandafter\def\csname gz:hte:sw:R2_0.90:mse\endcsname{0.99}
\expandafter\def\csname gz:hte:sw:R2_0.90:plugratio\endcsname{1.00}
\expandafter\def\csname gz:hte:sw:R2_0.90:rankdef\endcsname{0.0}
\expandafter\def\csname gz:hte:sw:maxdev:c:o\endcsname{0.1}
\expandafter\def\csname gz:hte:sw:oracle:max\endcsname{95.2}
\expandafter\def\csname gz:hte:sw:oracle:min\endcsname{94.9}
\expandafter\def\csname gz:hte:sw:reg:max\endcsname{95.2}
\expandafter\def\csname gz:hte:sw:reg:min\endcsname{94.9}
\expandafter\def\csname gz:intro:Delta\endcsname{83.4}
\expandafter\def\csname gz:intro:Deltaround\endcsname{83}
\expandafter\def\csname gz:intro:cov:IREG_HC1\endcsname{93.5}
\expandafter\def\csname gz:intro:cov:IREG_HC3\endcsname{95.2}
\expandafter\def\csname gz:intro:cov:REG_HC1\endcsname{94.9}
\expandafter\def\csname gz:intro:cov:mcse\endcsname{0.08}
\expandafter\def\csname gz:intro:pred:2000\endcsname{4.2}
\expandafter\def\csname gz:intro:pred:500\endcsname{16.7}
\expandafter\def\csname gz:intro:pred:8000\endcsname{1.04}
\expandafter\def\csname gz:intro:se:500\endcsname{0.07}
\expandafter\def\csname gz:intro:sim:2000\endcsname{2.3}
\expandafter\def\csname gz:intro:sim:500\endcsname{6.7}
\expandafter\def\csname gz:intro:sim:8000\endcsname{0.73}
\expandafter\def\csname gz:intro:skew\endcsname{8.2}
\expandafter\def\csname gz:jk:gauss_high:IREG_HC1c:cov\endcsname{95.1}
\expandafter\def\csname gz:jk:gauss_high:IREG_HC1c:width\endcsname{0.64}
\expandafter\def\csname gz:jk:gauss_high:IREG_JK:cov\endcsname{95.1}
\expandafter\def\csname gz:jk:gauss_high:IREG_JK:width\endcsname{0.64}
\expandafter\def\csname gz:jk:gauss_high:R2\endcsname{0.692}
\expandafter\def\csname gz:jk:gauss_high:REG_HC1:cov\endcsname{95.0}
\expandafter\def\csname gz:jk:gauss_high:REG_HC1:width\endcsname{0.64}
\expandafter\def\csname gz:jk:gauss_high:REG_JK:cov\endcsname{95.2}
\expandafter\def\csname gz:jk:gauss_high:REG_JK:width\endcsname{0.64}
\expandafter\def\csname gz:jk:gauss_low:IREG_HC1c:cov\endcsname{95.0}
\expandafter\def\csname gz:jk:gauss_low:IREG_HC1c:width\endcsname{0.36}
\expandafter\def\csname gz:jk:gauss_low:IREG_JK:cov\endcsname{95.0}
\expandafter\def\csname gz:jk:gauss_low:IREG_JK:width\endcsname{0.36}
\expandafter\def\csname gz:jk:gauss_low:R2\endcsname{0.022}
\expandafter\def\csname gz:jk:gauss_low:REG_HC1:cov\endcsname{94.8}
\expandafter\def\csname gz:jk:gauss_low:REG_HC1:width\endcsname{0.36}
\expandafter\def\csname gz:jk:gauss_low:REG_JK:cov\endcsname{95.0}
\expandafter\def\csname gz:jk:gauss_low:REG_JK:width\endcsname{0.36}
\expandafter\def\csname gz:jk:heavy_ex:IREG_HC1c:cov\endcsname{94.3}
\expandafter\def\csname gz:jk:heavy_ex:IREG_HC1c:width\endcsname{7.26}
\expandafter\def\csname gz:jk:heavy_ex:IREG_JK:cov\endcsname{96.4}
\expandafter\def\csname gz:jk:heavy_ex:IREG_JK:width\endcsname{8.04}
\expandafter\def\csname gz:jk:heavy_ex:R2\endcsname{0.009}
\expandafter\def\csname gz:jk:heavy_ex:REG_HC1:cov\endcsname{94.3}
\expandafter\def\csname gz:jk:heavy_ex:REG_HC1:width\endcsname{6.57}
\expandafter\def\csname gz:jk:heavy_ex:REG_JK:cov\endcsname{96.6}
\expandafter\def\csname gz:jk:heavy_ex:REG_JK:width\endcsname{7.18}
\expandafter\def\csname gz:jk:heavy_high:IREG_HC1c:cov\endcsname{94.2}
\expandafter\def\csname gz:jk:heavy_high:IREG_HC1c:width\endcsname{0.75}
\expandafter\def\csname gz:jk:heavy_high:IREG_JK:cov\endcsname{94.6}
\expandafter\def\csname gz:jk:heavy_high:IREG_JK:width\endcsname{0.77}
\expandafter\def\csname gz:jk:heavy_high:R2\endcsname{0.365}
\expandafter\def\csname gz:jk:heavy_high:REG_HC1:cov\endcsname{91.9}
\expandafter\def\csname gz:jk:heavy_high:REG_HC1:width\endcsname{0.73}
\expandafter\def\csname gz:jk:heavy_high:REG_JK:cov\endcsname{93.0}
\expandafter\def\csname gz:jk:heavy_high:REG_JK:width\endcsname{0.76}
\expandafter\def\csname gz:jk:heavy_low:IREG_HC1c:cov\endcsname{94.7}
\expandafter\def\csname gz:jk:heavy_low:IREG_HC1c:width\endcsname{0.64}
\expandafter\def\csname gz:jk:heavy_low:IREG_JK:cov\endcsname{95.5}
\expandafter\def\csname gz:jk:heavy_low:IREG_JK:width\endcsname{0.66}
\expandafter\def\csname gz:jk:heavy_low:R2\endcsname{0.016}
\expandafter\def\csname gz:jk:heavy_low:REG_HC1:cov\endcsname{94.7}
\expandafter\def\csname gz:jk:heavy_low:REG_HC1:width\endcsname{0.63}
\expandafter\def\csname gz:jk:heavy_low:REG_JK:cov\endcsname{95.5}
\expandafter\def\csname gz:jk:heavy_low:REG_JK:width\endcsname{0.64}
\expandafter\def\csname gz:jk:ln12_small:IREG_HC1c:cov\endcsname{94.6}
\expandafter\def\csname gz:jk:ln12_small:IREG_HC1c:width\endcsname{1.52}
\expandafter\def\csname gz:jk:ln12_small:IREG_JK:cov\endcsname{95.6}
\expandafter\def\csname gz:jk:ln12_small:IREG_JK:width\endcsname{1.57}
\expandafter\def\csname gz:jk:ln12_small:R2\endcsname{0.013}
\expandafter\def\csname gz:jk:ln12_small:REG_HC1:cov\endcsname{94.6}
\expandafter\def\csname gz:jk:ln12_small:REG_HC1:width\endcsname{1.47}
\expandafter\def\csname gz:jk:ln12_small:REG_JK:cov\endcsname{95.7}
\expandafter\def\csname gz:jk:ln12_small:REG_JK:width\endcsname{1.52}
\expandafter\def\csname gz:jk:ln15_g05:IREG_HC1c:cov\endcsname{93.7}
\expandafter\def\csname gz:jk:ln15_g05:IREG_HC1c:width\endcsname{2.62}
\expandafter\def\csname gz:jk:ln15_g05:IREG_JK:cov\endcsname{95.3}
\expandafter\def\csname gz:jk:ln15_g05:IREG_JK:width\endcsname{2.79}
\expandafter\def\csname gz:jk:ln15_g05:R2\endcsname{0.049}
\expandafter\def\csname gz:jk:ln15_g05:REG_HC1:cov\endcsname{93.2}
\expandafter\def\csname gz:jk:ln15_g05:REG_HC1:width\endcsname{2.46}
\expandafter\def\csname gz:jk:ln15_g05:REG_JK:cov\endcsname{95.1}
\expandafter\def\csname gz:jk:ln15_g05:REG_JK:width\endcsname{2.61}
\expandafter\def\csname gz:jk:ln15_large:IREG_HC1c:cov\endcsname{89.8}
\expandafter\def\csname gz:jk:ln15_large:IREG_HC1c:width\endcsname{4.46}
\expandafter\def\csname gz:jk:ln15_large:IREG_JK:cov\endcsname{90.4}
\expandafter\def\csname gz:jk:ln15_large:IREG_JK:width\endcsname{4.57}
\expandafter\def\csname gz:jk:ln15_large:R2\endcsname{0.651}
\expandafter\def\csname gz:jk:ln15_large:REG_HC1:cov\endcsname{84.1}
\expandafter\def\csname gz:jk:ln15_large:REG_HC1:width\endcsname{4.17}
\expandafter\def\csname gz:jk:ln15_large:REG_JK:cov\endcsname{86.0}
\expandafter\def\csname gz:jk:ln15_large:REG_JK:width\endcsname{4.51}
\expandafter\def\csname gz:jk:ln15_mid:IREG_HC1c:cov\endcsname{92.0}
\expandafter\def\csname gz:jk:ln15_mid:IREG_HC1c:width\endcsname{2.81}
\expandafter\def\csname gz:jk:ln15_mid:IREG_JK:cov\endcsname{93.7}
\expandafter\def\csname gz:jk:ln15_mid:IREG_JK:width\endcsname{2.97}
\expandafter\def\csname gz:jk:ln15_mid:R2\endcsname{0.172}
\expandafter\def\csname gz:jk:ln15_mid:REG_HC1:cov\endcsname{90.2}
\expandafter\def\csname gz:jk:ln15_mid:REG_HC1:width\endcsname{2.64}
\expandafter\def\csname gz:jk:ln15_mid:REG_JK:cov\endcsname{91.9}
\expandafter\def\csname gz:jk:ln15_mid:REG_JK:width\endcsname{2.81}
\expandafter\def\csname gz:jk:ln15_small:IREG_HC1c:cov\endcsname{94.7}
\expandafter\def\csname gz:jk:ln15_small:IREG_HC1c:width\endcsname{2.55}
\expandafter\def\csname gz:jk:ln15_small:IREG_JK:cov\endcsname{96.2}
\expandafter\def\csname gz:jk:ln15_small:IREG_JK:width\endcsname{2.73}
\expandafter\def\csname gz:jk:ln15_small:R2\endcsname{0.002}
\expandafter\def\csname gz:jk:ln15_small:REG_HC1:cov\endcsname{94.9}
\expandafter\def\csname gz:jk:ln15_small:REG_HC1:width\endcsname{2.40}
\expandafter\def\csname gz:jk:ln15_small:REG_JK:cov\endcsname{96.7}
\expandafter\def\csname gz:jk:ln15_small:REG_JK:width\endcsname{2.55}
\expandafter\def\csname gz:jk:ln20_small:IREG_HC1c:cov\endcsname{94.8}
\expandafter\def\csname gz:jk:ln20_small:IREG_HC1c:width\endcsname{7.76}
\expandafter\def\csname gz:jk:ln20_small:IREG_JK:cov\endcsname{97.0}
\expandafter\def\csname gz:jk:ln20_small:IREG_JK:width\endcsname{8.78}
\expandafter\def\csname gz:jk:ln20_small:R2\endcsname{0.002}
\expandafter\def\csname gz:jk:ln20_small:REG_HC1:cov\endcsname{95.0}
\expandafter\def\csname gz:jk:ln20_small:REG_HC1:width\endcsname{6.80}
\expandafter\def\csname gz:jk:ln20_small:REG_JK:cov\endcsname{97.6}
\expandafter\def\csname gz:jk:ln20_small:REG_JK:width\endcsname{7.55}
\expandafter\def\csname gz:jk:misspec_null:IREG_HC1c:cov\endcsname{94.9}
\expandafter\def\csname gz:jk:misspec_null:IREG_HC1c:width\endcsname{0.38}
\expandafter\def\csname gz:jk:misspec_null:IREG_JK:cov\endcsname{95.0}
\expandafter\def\csname gz:jk:misspec_null:IREG_JK:width\endcsname{0.38}
\expandafter\def\csname gz:jk:misspec_null:R2\endcsname{0}
\expandafter\def\csname gz:jk:misspec_null:REG_HC1:cov\endcsname{94.8}
\expandafter\def\csname gz:jk:misspec_null:REG_HC1:width\endcsname{0.38}
\expandafter\def\csname gz:jk:misspec_null:REG_JK:cov\endcsname{95.0}
\expandafter\def\csname gz:jk:misspec_null:REG_JK:width\endcsname{0.38}
\expandafter\def\csname gz:jk:rare:IREG:degcov\endcsname{89.9}
\expandafter\def\csname gz:jk:rare:IREG:degfrac\endcsname{6.3}
\expandafter\def\csname gz:jk:rare:IREG:othercov\endcsname{94.3}
\expandafter\def\csname gz:jk:rare:IREG:strict\endcsname{88.4}
\expandafter\def\csname gz:jk:rare:IREG_HC1c:cov\endcsname{94.0}
\expandafter\def\csname gz:jk:rare:IREG_HC1c:width\endcsname{0.40}
\expandafter\def\csname gz:jk:rare:IREG_JK:cov\endcsname{94.0}
\expandafter\def\csname gz:jk:rare:IREG_JK:width\endcsname{0.40}
\expandafter\def\csname gz:jk:rare:R2\endcsname{0.2}
\expandafter\def\csname gz:jk:rare:REG:degfrac\endcsname{0.03}
\expandafter\def\csname gz:jk:rare:REG_HC1:cov\endcsname{94.7}
\expandafter\def\csname gz:jk:rare:REG_HC1:width\endcsname{0.40}
\expandafter\def\csname gz:jk:rare:REG_JK:cov\endcsname{95.5}
\expandafter\def\csname gz:jk:rare:REG_JK:width\endcsname{0.41}
\expandafter\def\csname gz:jk:rare:rankdef\endcsname{1.3}
\expandafter\def\csname gz:jk:skew_null:IREG_HC1c:cov\endcsname{95.0}
\expandafter\def\csname gz:jk:skew_null:IREG_HC1c:width\endcsname{0.65}
\expandafter\def\csname gz:jk:skew_null:IREG_JK:cov\endcsname{95.4}
\expandafter\def\csname gz:jk:skew_null:IREG_JK:width\endcsname{0.66}
\expandafter\def\csname gz:jk:skew_null:R2\endcsname{0}
\expandafter\def\csname gz:jk:skew_null:REG_HC1:cov\endcsname{95.0}
\expandafter\def\csname gz:jk:skew_null:REG_HC1:width\endcsname{0.64}
\expandafter\def\csname gz:jk:skew_null:REG_JK:cov\endcsname{95.5}
\expandafter\def\csname gz:jk:skew_null:REG_JK:width\endcsname{0.65}
\expandafter\def\csname gz:jk:small:IREG_HC1c:0.80:max\endcsname{80.1}
\expandafter\def\csname gz:jk:small:IREG_HC1c:0.80:min\endcsname{75.1}
\expandafter\def\csname gz:jk:small:IREG_HC1c:0.95:max\endcsname{95.0}
\expandafter\def\csname gz:jk:small:IREG_HC1c:0.95:min\endcsname{93.7}
\expandafter\def\csname gz:jk:small:IREG_HC1c:0.99:max\endcsname{99.3}
\expandafter\def\csname gz:jk:small:IREG_HC1c:0.99:min\endcsname{98.5}
\expandafter\def\csname gz:jk:small:IREG_HC1c:max\endcsname{95.0}
\expandafter\def\csname gz:jk:small:IREG_HC1c:maxshort\endcsname{4.9}
\expandafter\def\csname gz:jk:small:IREG_HC1c:min\endcsname{93.7}
\expandafter\def\csname gz:jk:small:IREG_JK:0.80:max\endcsname{82.3}
\expandafter\def\csname gz:jk:small:IREG_JK:0.80:min\endcsname{79.8}
\expandafter\def\csname gz:jk:small:IREG_JK:0.95:max\endcsname{97.0}
\expandafter\def\csname gz:jk:small:IREG_JK:0.95:min\endcsname{95.0}
\expandafter\def\csname gz:jk:small:IREG_JK:0.99:max\endcsname{99.6}
\expandafter\def\csname gz:jk:small:IREG_JK:0.99:min\endcsname{98.9}
\expandafter\def\csname gz:jk:small:IREG_JK:max\endcsname{97.0}
\expandafter\def\csname gz:jk:small:IREG_JK:maxshort\endcsname{0.2}
\expandafter\def\csname gz:jk:small:IREG_JK:min\endcsname{95.0}
\expandafter\def\csname gz:jk:small:REG_HC1:0.80:max\endcsname{80.0}
\expandafter\def\csname gz:jk:small:REG_HC1:0.80:min\endcsname{76.1}
\expandafter\def\csname gz:jk:small:REG_HC1:0.95:max\endcsname{95.0}
\expandafter\def\csname gz:jk:small:REG_HC1:0.95:min\endcsname{93.2}
\expandafter\def\csname gz:jk:small:REG_HC1:0.99:max\endcsname{99.3}
\expandafter\def\csname gz:jk:small:REG_HC1:0.99:min\endcsname{98.3}
\expandafter\def\csname gz:jk:small:REG_HC1:max\endcsname{95.0}
\expandafter\def\csname gz:jk:small:REG_HC1:maxshort\endcsname{3.9}
\expandafter\def\csname gz:jk:small:REG_HC1:min\endcsname{93.2}
\expandafter\def\csname gz:jk:small:REG_JK:0.80:max\endcsname{84.2}
\expandafter\def\csname gz:jk:small:REG_JK:0.80:min\endcsname{79.9}
\expandafter\def\csname gz:jk:small:REG_JK:0.95:max\endcsname{97.6}
\expandafter\def\csname gz:jk:small:REG_JK:0.95:min\endcsname{95.0}
\expandafter\def\csname gz:jk:small:REG_JK:0.99:max\endcsname{99.7}
\expandafter\def\csname gz:jk:small:REG_JK:0.99:min\endcsname{98.8}
\expandafter\def\csname gz:jk:small:REG_JK:max\endcsname{97.6}
\expandafter\def\csname gz:jk:small:REG_JK:maxshort\endcsname{0.2}
\expandafter\def\csname gz:jk:small:REG_JK:min\endcsname{95.0}
\expandafter\def\csname gz:jk:small:mseratio:max\endcsname{1.95}
\expandafter\def\csname gz:jk:small:mseratio:min\endcsname{1.00}
\expandafter\def\csname gz:jk:small:wRJK_IHC1c:max\endcsname{1.00}
\expandafter\def\csname gz:jk:small:wRJK_IHC1c:maxpct\endcsname{0.5}
\expandafter\def\csname gz:jk:small:wRJK_IHC1c:min\endcsname{0.97}
\expandafter\def\csname gz:jk:small:wRJK_IJK:max\endcsname{1.00}
\expandafter\def\csname gz:jk:small:wRJK_IJK:min\endcsname{0.86}
\expandafter\def\csname gz:jk:small:wRJK_IJK:minpct\endcsname{14}
\expandafter\def\csname gz:jk:small:wRJK_RHC1:max\endcsname{1.11}
\expandafter\def\csname gz:lev1:mcse:max\endcsname{1.4}
\expandafter\def\csname gz:lev1:n:min\endcsname{1225}
\expandafter\def\csname gz:lognormal_s0.6_prop:IREG_HC1:all\endcsname{94.5}
\expandafter\def\csname gz:lognormal_s0.6_prop:IREG_HC1:all:nointerval\endcsname{0.0}
\expandafter\def\csname gz:lognormal_s0.6_prop:IREG_HC1:full\endcsname{94.5}
\expandafter\def\csname gz:lognormal_s0.6_prop:IREG_HC1:full:nointerval\endcsname{0.0}
\expandafter\def\csname gz:lognormal_s0.6_prop:IREG_HC1:lev1\endcsname{--}
\expandafter\def\csname gz:lognormal_s0.6_prop:IREG_HC2:all\endcsname{94.7}
\expandafter\def\csname gz:lognormal_s0.6_prop:IREG_HC2:all:nointerval\endcsname{0.0}
\expandafter\def\csname gz:lognormal_s0.6_prop:IREG_HC2:full\endcsname{94.7}
\expandafter\def\csname gz:lognormal_s0.6_prop:IREG_HC2:full:nointerval\endcsname{0.0}
\expandafter\def\csname gz:lognormal_s0.6_prop:IREG_HC2:lev1\endcsname{--}
\expandafter\def\csname gz:lognormal_s0.6_prop:IREG_HC3:all\endcsname{94.8}
\expandafter\def\csname gz:lognormal_s0.6_prop:IREG_HC3:all:nointerval\endcsname{0.0}
\expandafter\def\csname gz:lognormal_s0.6_prop:IREG_HC3:full\endcsname{94.8}
\expandafter\def\csname gz:lognormal_s0.6_prop:IREG_HC3:full:nointerval\endcsname{0.0}
\expandafter\def\csname gz:lognormal_s0.6_prop:IREG_HC3:lev1\endcsname{--}
\expandafter\def\csname gz:lognormal_s0.6_prop:REG_HC1:all\endcsname{94.7}
\expandafter\def\csname gz:lognormal_s0.6_prop:REG_HC1:all:nointerval\endcsname{0.0}
\expandafter\def\csname gz:lognormal_s0.6_prop:REG_HC1:full\endcsname{94.7}
\expandafter\def\csname gz:lognormal_s0.6_prop:REG_HC1:full:nointerval\endcsname{0.0}
\expandafter\def\csname gz:lognormal_s0.6_prop:REG_HC1:lev1\endcsname{--}
\expandafter\def\csname gz:lognormal_s0.6_prop:REG_HC3:all\endcsname{94.9}
\expandafter\def\csname gz:lognormal_s0.6_prop:REG_HC3:all:nointerval\endcsname{0.0}
\expandafter\def\csname gz:lognormal_s0.6_prop:REG_HC3:full\endcsname{94.9}
\expandafter\def\csname gz:lognormal_s0.6_prop:REG_HC3:full:nointerval\endcsname{0.0}
\expandafter\def\csname gz:lognormal_s0.6_prop:REG_HC3:lev1\endcsname{--}
\expandafter\def\csname gz:lognormal_s0.6_prop:lev1\endcsname{0.0}
\expandafter\def\csname gz:lognormal_s0.6_prop:medlev\endcsname{0.15}
\expandafter\def\csname gz:lognormal_s0.6_prop:rankdef\endcsname{0.0}
\expandafter\def\csname gz:lognormal_s0.6_prop:reps\endcsname{20{,}000}
\expandafter\def\csname gz:lognormal_s0.8_prop:IREG_HC1:all\endcsname{94.4}
\expandafter\def\csname gz:lognormal_s0.8_prop:IREG_HC1:all:nointerval\endcsname{0.0}
\expandafter\def\csname gz:lognormal_s0.8_prop:IREG_HC1:full\endcsname{94.4}
\expandafter\def\csname gz:lognormal_s0.8_prop:IREG_HC1:full:nointerval\endcsname{0.0}
\expandafter\def\csname gz:lognormal_s0.8_prop:IREG_HC1:lev1\endcsname{--}
\expandafter\def\csname gz:lognormal_s0.8_prop:IREG_HC2:all\endcsname{94.6}
\expandafter\def\csname gz:lognormal_s0.8_prop:IREG_HC2:all:nointerval\endcsname{0.0}
\expandafter\def\csname gz:lognormal_s0.8_prop:IREG_HC2:full\endcsname{94.6}
\expandafter\def\csname gz:lognormal_s0.8_prop:IREG_HC2:full:nointerval\endcsname{0.0}
\expandafter\def\csname gz:lognormal_s0.8_prop:IREG_HC2:lev1\endcsname{--}
\expandafter\def\csname gz:lognormal_s0.8_prop:IREG_HC3:all\endcsname{95.1}
\expandafter\def\csname gz:lognormal_s0.8_prop:IREG_HC3:all:nointerval\endcsname{0.0}
\expandafter\def\csname gz:lognormal_s0.8_prop:IREG_HC3:full\endcsname{95.1}
\expandafter\def\csname gz:lognormal_s0.8_prop:IREG_HC3:full:nointerval\endcsname{0.0}
\expandafter\def\csname gz:lognormal_s0.8_prop:IREG_HC3:lev1\endcsname{--}
\expandafter\def\csname gz:lognormal_s0.8_prop:REG_HC1:all\endcsname{94.7}
\expandafter\def\csname gz:lognormal_s0.8_prop:REG_HC1:all:nointerval\endcsname{0.0}
\expandafter\def\csname gz:lognormal_s0.8_prop:REG_HC1:full\endcsname{94.7}
\expandafter\def\csname gz:lognormal_s0.8_prop:REG_HC1:full:nointerval\endcsname{0.0}
\expandafter\def\csname gz:lognormal_s0.8_prop:REG_HC1:lev1\endcsname{--}
\expandafter\def\csname gz:lognormal_s0.8_prop:REG_HC3:all\endcsname{95.1}
\expandafter\def\csname gz:lognormal_s0.8_prop:REG_HC3:all:nointerval\endcsname{0.0}
\expandafter\def\csname gz:lognormal_s0.8_prop:REG_HC3:full\endcsname{95.1}
\expandafter\def\csname gz:lognormal_s0.8_prop:REG_HC3:full:nointerval\endcsname{0.0}
\expandafter\def\csname gz:lognormal_s0.8_prop:REG_HC3:lev1\endcsname{--}
\expandafter\def\csname gz:lognormal_s0.8_prop:lev1\endcsname{0.0}
\expandafter\def\csname gz:lognormal_s0.8_prop:medlev\endcsname{0.21}
\expandafter\def\csname gz:lognormal_s0.8_prop:rankdef\endcsname{0.0}
\expandafter\def\csname gz:lognormal_s0.8_prop:reps\endcsname{20{,}000}
\expandafter\def\csname gz:lognormal_s1.0_log:IREG_HC1:all\endcsname{94.5}
\expandafter\def\csname gz:lognormal_s1.0_log:IREG_HC1:all:nointerval\endcsname{0.0}
\expandafter\def\csname gz:lognormal_s1.0_log:IREG_HC1:full\endcsname{94.5}
\expandafter\def\csname gz:lognormal_s1.0_log:IREG_HC1:full:nointerval\endcsname{0.0}
\expandafter\def\csname gz:lognormal_s1.0_log:IREG_HC1:lev1\endcsname{--}
\expandafter\def\csname gz:lognormal_s1.0_log:IREG_HC2:all\endcsname{95.0}
\expandafter\def\csname gz:lognormal_s1.0_log:IREG_HC2:all:nointerval\endcsname{0.0}
\expandafter\def\csname gz:lognormal_s1.0_log:IREG_HC2:full\endcsname{95.0}
\expandafter\def\csname gz:lognormal_s1.0_log:IREG_HC2:full:nointerval\endcsname{0.0}
\expandafter\def\csname gz:lognormal_s1.0_log:IREG_HC2:lev1\endcsname{--}
\expandafter\def\csname gz:lognormal_s1.0_log:IREG_HC3:all\endcsname{95.7}
\expandafter\def\csname gz:lognormal_s1.0_log:IREG_HC3:all:nointerval\endcsname{0.0}
\expandafter\def\csname gz:lognormal_s1.0_log:IREG_HC3:full\endcsname{95.7}
\expandafter\def\csname gz:lognormal_s1.0_log:IREG_HC3:full:nointerval\endcsname{0.0}
\expandafter\def\csname gz:lognormal_s1.0_log:IREG_HC3:lev1\endcsname{--}
\expandafter\def\csname gz:lognormal_s1.0_log:REG_HC1:all\endcsname{95.0}
\expandafter\def\csname gz:lognormal_s1.0_log:REG_HC1:all:nointerval\endcsname{0.0}
\expandafter\def\csname gz:lognormal_s1.0_log:REG_HC1:full\endcsname{95.0}
\expandafter\def\csname gz:lognormal_s1.0_log:REG_HC1:full:nointerval\endcsname{0.0}
\expandafter\def\csname gz:lognormal_s1.0_log:REG_HC1:lev1\endcsname{--}
\expandafter\def\csname gz:lognormal_s1.0_log:REG_HC3:all\endcsname{95.6}
\expandafter\def\csname gz:lognormal_s1.0_log:REG_HC3:all:nointerval\endcsname{0.0}
\expandafter\def\csname gz:lognormal_s1.0_log:REG_HC3:full\endcsname{95.6}
\expandafter\def\csname gz:lognormal_s1.0_log:REG_HC3:full:nointerval\endcsname{0.0}
\expandafter\def\csname gz:lognormal_s1.0_log:REG_HC3:lev1\endcsname{--}
\expandafter\def\csname gz:lognormal_s1.0_log:lev1\endcsname{0.0}
\expandafter\def\csname gz:lognormal_s1.0_log:medlev\endcsname{0.29}
\expandafter\def\csname gz:lognormal_s1.0_log:rankdef\endcsname{0.0}
\expandafter\def\csname gz:lognormal_s1.0_log:reps\endcsname{20{,}000}
\expandafter\def\csname gz:lognormal_s1.0_log_prop:IREG_HC1:all\endcsname{94.2}
\expandafter\def\csname gz:lognormal_s1.0_log_prop:IREG_HC1:all:nointerval\endcsname{0.0}
\expandafter\def\csname gz:lognormal_s1.0_log_prop:IREG_HC1:full\endcsname{94.2}
\expandafter\def\csname gz:lognormal_s1.0_log_prop:IREG_HC1:full:nointerval\endcsname{0.0}
\expandafter\def\csname gz:lognormal_s1.0_log_prop:IREG_HC1:lev1\endcsname{--}
\expandafter\def\csname gz:lognormal_s1.0_log_prop:IREG_HC2:all\endcsname{94.7}
\expandafter\def\csname gz:lognormal_s1.0_log_prop:IREG_HC2:all:nointerval\endcsname{0.0}
\expandafter\def\csname gz:lognormal_s1.0_log_prop:IREG_HC2:full\endcsname{94.7}
\expandafter\def\csname gz:lognormal_s1.0_log_prop:IREG_HC2:full:nointerval\endcsname{0.0}
\expandafter\def\csname gz:lognormal_s1.0_log_prop:IREG_HC2:lev1\endcsname{--}
\expandafter\def\csname gz:lognormal_s1.0_log_prop:IREG_HC3:all\endcsname{95.5}
\expandafter\def\csname gz:lognormal_s1.0_log_prop:IREG_HC3:all:nointerval\endcsname{0.0}
\expandafter\def\csname gz:lognormal_s1.0_log_prop:IREG_HC3:full\endcsname{95.5}
\expandafter\def\csname gz:lognormal_s1.0_log_prop:IREG_HC3:full:nointerval\endcsname{0.0}
\expandafter\def\csname gz:lognormal_s1.0_log_prop:IREG_HC3:lev1\endcsname{--}
\expandafter\def\csname gz:lognormal_s1.0_log_prop:REG_HC1:all\endcsname{94.9}
\expandafter\def\csname gz:lognormal_s1.0_log_prop:REG_HC1:all:nointerval\endcsname{0.0}
\expandafter\def\csname gz:lognormal_s1.0_log_prop:REG_HC1:full\endcsname{94.9}
\expandafter\def\csname gz:lognormal_s1.0_log_prop:REG_HC1:full:nointerval\endcsname{0.0}
\expandafter\def\csname gz:lognormal_s1.0_log_prop:REG_HC1:lev1\endcsname{--}
\expandafter\def\csname gz:lognormal_s1.0_log_prop:REG_HC3:all\endcsname{95.6}
\expandafter\def\csname gz:lognormal_s1.0_log_prop:REG_HC3:all:nointerval\endcsname{0.0}
\expandafter\def\csname gz:lognormal_s1.0_log_prop:REG_HC3:full\endcsname{95.6}
\expandafter\def\csname gz:lognormal_s1.0_log_prop:REG_HC3:full:nointerval\endcsname{0.0}
\expandafter\def\csname gz:lognormal_s1.0_log_prop:REG_HC3:lev1\endcsname{--}
\expandafter\def\csname gz:lognormal_s1.0_log_prop:lev1\endcsname{0.0}
\expandafter\def\csname gz:lognormal_s1.0_log_prop:medlev\endcsname{0.29}
\expandafter\def\csname gz:lognormal_s1.0_log_prop:rankdef\endcsname{0.0}
\expandafter\def\csname gz:lognormal_s1.0_log_prop:reps\endcsname{20{,}000}
\expandafter\def\csname gz:lognormal_s1.0_prop:IREG_HC1:all\endcsname{94.2}
\expandafter\def\csname gz:lognormal_s1.0_prop:IREG_HC1:all:nointerval\endcsname{0.0}
\expandafter\def\csname gz:lognormal_s1.0_prop:IREG_HC1:full\endcsname{94.2}
\expandafter\def\csname gz:lognormal_s1.0_prop:IREG_HC1:full:nointerval\endcsname{0.0}
\expandafter\def\csname gz:lognormal_s1.0_prop:IREG_HC1:lev1\endcsname{--}
\expandafter\def\csname gz:lognormal_s1.0_prop:IREG_HC2:all\endcsname{94.7}
\expandafter\def\csname gz:lognormal_s1.0_prop:IREG_HC2:all:nointerval\endcsname{0.0}
\expandafter\def\csname gz:lognormal_s1.0_prop:IREG_HC2:full\endcsname{94.7}
\expandafter\def\csname gz:lognormal_s1.0_prop:IREG_HC2:full:nointerval\endcsname{0.0}
\expandafter\def\csname gz:lognormal_s1.0_prop:IREG_HC2:lev1\endcsname{--}
\expandafter\def\csname gz:lognormal_s1.0_prop:IREG_HC3:all\endcsname{95.4}
\expandafter\def\csname gz:lognormal_s1.0_prop:IREG_HC3:all:nointerval\endcsname{0.0}
\expandafter\def\csname gz:lognormal_s1.0_prop:IREG_HC3:full\endcsname{95.4}
\expandafter\def\csname gz:lognormal_s1.0_prop:IREG_HC3:full:nointerval\endcsname{0.0}
\expandafter\def\csname gz:lognormal_s1.0_prop:IREG_HC3:lev1\endcsname{--}
\expandafter\def\csname gz:lognormal_s1.0_prop:REG_HC1:all\endcsname{95.1}
\expandafter\def\csname gz:lognormal_s1.0_prop:REG_HC1:all:nointerval\endcsname{0.0}
\expandafter\def\csname gz:lognormal_s1.0_prop:REG_HC1:full\endcsname{95.1}
\expandafter\def\csname gz:lognormal_s1.0_prop:REG_HC1:full:nointerval\endcsname{0.0}
\expandafter\def\csname gz:lognormal_s1.0_prop:REG_HC1:lev1\endcsname{--}
\expandafter\def\csname gz:lognormal_s1.0_prop:REG_HC3:all\endcsname{95.8}
\expandafter\def\csname gz:lognormal_s1.0_prop:REG_HC3:all:nointerval\endcsname{0.0}
\expandafter\def\csname gz:lognormal_s1.0_prop:REG_HC3:full\endcsname{95.8}
\expandafter\def\csname gz:lognormal_s1.0_prop:REG_HC3:full:nointerval\endcsname{0.0}
\expandafter\def\csname gz:lognormal_s1.0_prop:REG_HC3:lev1\endcsname{--}
\expandafter\def\csname gz:lognormal_s1.0_prop:lev1\endcsname{0.0}
\expandafter\def\csname gz:lognormal_s1.0_prop:medlev\endcsname{0.29}
\expandafter\def\csname gz:lognormal_s1.0_prop:rankdef\endcsname{0.0}
\expandafter\def\csname gz:lognormal_s1.0_prop:reps\endcsname{20{,}000}
\expandafter\def\csname gz:lognormal_s1.0_prop_p5_n200:IREG_HC1:all\endcsname{93.1}
\expandafter\def\csname gz:lognormal_s1.0_prop_p5_n200:IREG_HC1:all:nointerval\endcsname{0.0}
\expandafter\def\csname gz:lognormal_s1.0_prop_p5_n200:IREG_HC1:full\endcsname{93.1}
\expandafter\def\csname gz:lognormal_s1.0_prop_p5_n200:IREG_HC1:full:nointerval\endcsname{0.0}
\expandafter\def\csname gz:lognormal_s1.0_prop_p5_n200:IREG_HC1:lev1\endcsname{--}
\expandafter\def\csname gz:lognormal_s1.0_prop_p5_n200:IREG_HC2:all\endcsname{93.8}
\expandafter\def\csname gz:lognormal_s1.0_prop_p5_n200:IREG_HC2:all:nointerval\endcsname{0.0}
\expandafter\def\csname gz:lognormal_s1.0_prop_p5_n200:IREG_HC2:full\endcsname{93.8}
\expandafter\def\csname gz:lognormal_s1.0_prop_p5_n200:IREG_HC2:full:nointerval\endcsname{0.0}
\expandafter\def\csname gz:lognormal_s1.0_prop_p5_n200:IREG_HC2:lev1\endcsname{--}
\expandafter\def\csname gz:lognormal_s1.0_prop_p5_n200:IREG_HC3:all\endcsname{95.5}
\expandafter\def\csname gz:lognormal_s1.0_prop_p5_n200:IREG_HC3:all:nointerval\endcsname{0.0}
\expandafter\def\csname gz:lognormal_s1.0_prop_p5_n200:IREG_HC3:full\endcsname{95.5}
\expandafter\def\csname gz:lognormal_s1.0_prop_p5_n200:IREG_HC3:full:nointerval\endcsname{0.0}
\expandafter\def\csname gz:lognormal_s1.0_prop_p5_n200:IREG_HC3:lev1\endcsname{--}
\expandafter\def\csname gz:lognormal_s1.0_prop_p5_n200:REG_HC1:all\endcsname{94.9}
\expandafter\def\csname gz:lognormal_s1.0_prop_p5_n200:REG_HC1:all:nointerval\endcsname{0.0}
\expandafter\def\csname gz:lognormal_s1.0_prop_p5_n200:REG_HC1:full\endcsname{94.9}
\expandafter\def\csname gz:lognormal_s1.0_prop_p5_n200:REG_HC1:full:nointerval\endcsname{0.0}
\expandafter\def\csname gz:lognormal_s1.0_prop_p5_n200:REG_HC1:lev1\endcsname{--}
\expandafter\def\csname gz:lognormal_s1.0_prop_p5_n200:REG_HC3:all\endcsname{96.1}
\expandafter\def\csname gz:lognormal_s1.0_prop_p5_n200:REG_HC3:all:nointerval\endcsname{0.0}
\expandafter\def\csname gz:lognormal_s1.0_prop_p5_n200:REG_HC3:full\endcsname{96.1}
\expandafter\def\csname gz:lognormal_s1.0_prop_p5_n200:REG_HC3:full:nointerval\endcsname{0.0}
\expandafter\def\csname gz:lognormal_s1.0_prop_p5_n200:REG_HC3:lev1\endcsname{--}
\expandafter\def\csname gz:lognormal_s1.0_prop_p5_n200:lev1\endcsname{0.0}
\expandafter\def\csname gz:lognormal_s1.0_prop_p5_n200:medlev\endcsname{0.42}
\expandafter\def\csname gz:lognormal_s1.0_prop_p5_n200:rankdef\endcsname{0.0}
\expandafter\def\csname gz:lognormal_s1.0_prop_p5_n200:reps\endcsname{20{,}000}
\expandafter\def\csname gz:lognormal_s1.0_prop_p5_n500:IREG_HC1:all\endcsname{94.2}
\expandafter\def\csname gz:lognormal_s1.0_prop_p5_n500:IREG_HC1:all:nointerval\endcsname{0.0}
\expandafter\def\csname gz:lognormal_s1.0_prop_p5_n500:IREG_HC1:full\endcsname{94.2}
\expandafter\def\csname gz:lognormal_s1.0_prop_p5_n500:IREG_HC1:full:nointerval\endcsname{0.0}
\expandafter\def\csname gz:lognormal_s1.0_prop_p5_n500:IREG_HC1:lev1\endcsname{--}
\expandafter\def\csname gz:lognormal_s1.0_prop_p5_n500:IREG_HC2:all\endcsname{94.7}
\expandafter\def\csname gz:lognormal_s1.0_prop_p5_n500:IREG_HC2:all:nointerval\endcsname{0.0}
\expandafter\def\csname gz:lognormal_s1.0_prop_p5_n500:IREG_HC2:full\endcsname{94.7}
\expandafter\def\csname gz:lognormal_s1.0_prop_p5_n500:IREG_HC2:full:nointerval\endcsname{0.0}
\expandafter\def\csname gz:lognormal_s1.0_prop_p5_n500:IREG_HC2:lev1\endcsname{--}
\expandafter\def\csname gz:lognormal_s1.0_prop_p5_n500:IREG_HC3:all\endcsname{95.5}
\expandafter\def\csname gz:lognormal_s1.0_prop_p5_n500:IREG_HC3:all:nointerval\endcsname{0.0}
\expandafter\def\csname gz:lognormal_s1.0_prop_p5_n500:IREG_HC3:full\endcsname{95.5}
\expandafter\def\csname gz:lognormal_s1.0_prop_p5_n500:IREG_HC3:full:nointerval\endcsname{0.0}
\expandafter\def\csname gz:lognormal_s1.0_prop_p5_n500:IREG_HC3:lev1\endcsname{--}
\expandafter\def\csname gz:lognormal_s1.0_prop_p5_n500:REG_HC1:all\endcsname{95.1}
\expandafter\def\csname gz:lognormal_s1.0_prop_p5_n500:REG_HC1:all:nointerval\endcsname{0.0}
\expandafter\def\csname gz:lognormal_s1.0_prop_p5_n500:REG_HC1:full\endcsname{95.1}
\expandafter\def\csname gz:lognormal_s1.0_prop_p5_n500:REG_HC1:full:nointerval\endcsname{0.0}
\expandafter\def\csname gz:lognormal_s1.0_prop_p5_n500:REG_HC1:lev1\endcsname{--}
\expandafter\def\csname gz:lognormal_s1.0_prop_p5_n500:REG_HC3:all\endcsname{95.7}
\expandafter\def\csname gz:lognormal_s1.0_prop_p5_n500:REG_HC3:all:nointerval\endcsname{0.0}
\expandafter\def\csname gz:lognormal_s1.0_prop_p5_n500:REG_HC3:full\endcsname{95.7}
\expandafter\def\csname gz:lognormal_s1.0_prop_p5_n500:REG_HC3:full:nointerval\endcsname{0.0}
\expandafter\def\csname gz:lognormal_s1.0_prop_p5_n500:REG_HC3:lev1\endcsname{--}
\expandafter\def\csname gz:lognormal_s1.0_prop_p5_n500:lev1\endcsname{0.0}
\expandafter\def\csname gz:lognormal_s1.0_prop_p5_n500:medlev\endcsname{0.30}
\expandafter\def\csname gz:lognormal_s1.0_prop_p5_n500:rankdef\endcsname{0.0}
\expandafter\def\csname gz:lognormal_s1.0_prop_p5_n500:reps\endcsname{20{,}000}
\expandafter\def\csname gz:lognormal_s1.0_quadstd:IREG_HC1:all\endcsname{93.0}
\expandafter\def\csname gz:lognormal_s1.0_quadstd:IREG_HC1:all:nointerval\endcsname{0.0}
\expandafter\def\csname gz:lognormal_s1.0_quadstd:IREG_HC1:full\endcsname{93.0}
\expandafter\def\csname gz:lognormal_s1.0_quadstd:IREG_HC1:full:nointerval\endcsname{0.0}
\expandafter\def\csname gz:lognormal_s1.0_quadstd:IREG_HC1:lev1\endcsname{--}
\expandafter\def\csname gz:lognormal_s1.0_quadstd:IREG_HC2:all\endcsname{94.5}
\expandafter\def\csname gz:lognormal_s1.0_quadstd:IREG_HC2:all:nointerval\endcsname{0.0}
\expandafter\def\csname gz:lognormal_s1.0_quadstd:IREG_HC2:full\endcsname{94.5}
\expandafter\def\csname gz:lognormal_s1.0_quadstd:IREG_HC2:full:nointerval\endcsname{0.0}
\expandafter\def\csname gz:lognormal_s1.0_quadstd:IREG_HC2:lev1\endcsname{--}
\expandafter\def\csname gz:lognormal_s1.0_quadstd:IREG_HC3:all\endcsname{95.7}
\expandafter\def\csname gz:lognormal_s1.0_quadstd:IREG_HC3:all:nointerval\endcsname{0.0}
\expandafter\def\csname gz:lognormal_s1.0_quadstd:IREG_HC3:full\endcsname{95.7}
\expandafter\def\csname gz:lognormal_s1.0_quadstd:IREG_HC3:full:nointerval\endcsname{0.0}
\expandafter\def\csname gz:lognormal_s1.0_quadstd:IREG_HC3:lev1\endcsname{--}
\expandafter\def\csname gz:lognormal_s1.0_quadstd:REG_HC1:all\endcsname{95.1}
\expandafter\def\csname gz:lognormal_s1.0_quadstd:REG_HC1:all:nointerval\endcsname{0.0}
\expandafter\def\csname gz:lognormal_s1.0_quadstd:REG_HC1:full\endcsname{95.1}
\expandafter\def\csname gz:lognormal_s1.0_quadstd:REG_HC1:full:nointerval\endcsname{0.0}
\expandafter\def\csname gz:lognormal_s1.0_quadstd:REG_HC1:lev1\endcsname{--}
\expandafter\def\csname gz:lognormal_s1.0_quadstd:REG_HC3:all\endcsname{96.2}
\expandafter\def\csname gz:lognormal_s1.0_quadstd:REG_HC3:all:nointerval\endcsname{0.0}
\expandafter\def\csname gz:lognormal_s1.0_quadstd:REG_HC3:full\endcsname{96.2}
\expandafter\def\csname gz:lognormal_s1.0_quadstd:REG_HC3:full:nointerval\endcsname{0.0}
\expandafter\def\csname gz:lognormal_s1.0_quadstd:REG_HC3:lev1\endcsname{--}
\expandafter\def\csname gz:lognormal_s1.0_quadstd:lev1\endcsname{0.0}
\expandafter\def\csname gz:lognormal_s1.0_quadstd:medlev\endcsname{0.29}
\expandafter\def\csname gz:lognormal_s1.0_quadstd:rankdef\endcsname{0.0}
\expandafter\def\csname gz:lognormal_s1.0_quadstd:reps\endcsname{20{,}000}
\expandafter\def\csname gz:lognormal_s1.2_prop:IREG_HC1:all\endcsname{92.6}
\expandafter\def\csname gz:lognormal_s1.2_prop:IREG_HC1:all:nointerval\endcsname{0.0}
\expandafter\def\csname gz:lognormal_s1.2_prop:IREG_HC1:full\endcsname{92.6}
\expandafter\def\csname gz:lognormal_s1.2_prop:IREG_HC1:full:nointerval\endcsname{0.0}
\expandafter\def\csname gz:lognormal_s1.2_prop:IREG_HC1:lev1\endcsname{--}
\expandafter\def\csname gz:lognormal_s1.2_prop:IREG_HC2:all\endcsname{93.7}
\expandafter\def\csname gz:lognormal_s1.2_prop:IREG_HC2:all:nointerval\endcsname{0.0}
\expandafter\def\csname gz:lognormal_s1.2_prop:IREG_HC2:full\endcsname{93.7}
\expandafter\def\csname gz:lognormal_s1.2_prop:IREG_HC2:full:nointerval\endcsname{0.0}
\expandafter\def\csname gz:lognormal_s1.2_prop:IREG_HC2:lev1\endcsname{--}
\expandafter\def\csname gz:lognormal_s1.2_prop:IREG_HC3:all\endcsname{95.1}
\expandafter\def\csname gz:lognormal_s1.2_prop:IREG_HC3:all:nointerval\endcsname{0.0}
\expandafter\def\csname gz:lognormal_s1.2_prop:IREG_HC3:full\endcsname{95.1}
\expandafter\def\csname gz:lognormal_s1.2_prop:IREG_HC3:full:nointerval\endcsname{0.0}
\expandafter\def\csname gz:lognormal_s1.2_prop:IREG_HC3:lev1\endcsname{--}
\expandafter\def\csname gz:lognormal_s1.2_prop:REG_HC1:all\endcsname{94.8}
\expandafter\def\csname gz:lognormal_s1.2_prop:REG_HC1:all:nointerval\endcsname{0.0}
\expandafter\def\csname gz:lognormal_s1.2_prop:REG_HC1:full\endcsname{94.8}
\expandafter\def\csname gz:lognormal_s1.2_prop:REG_HC1:full:nointerval\endcsname{0.0}
\expandafter\def\csname gz:lognormal_s1.2_prop:REG_HC1:lev1\endcsname{--}
\expandafter\def\csname gz:lognormal_s1.2_prop:REG_HC3:all\endcsname{95.8}
\expandafter\def\csname gz:lognormal_s1.2_prop:REG_HC3:all:nointerval\endcsname{0.0}
\expandafter\def\csname gz:lognormal_s1.2_prop:REG_HC3:full\endcsname{95.8}
\expandafter\def\csname gz:lognormal_s1.2_prop:REG_HC3:full:nointerval\endcsname{0.0}
\expandafter\def\csname gz:lognormal_s1.2_prop:REG_HC3:lev1\endcsname{--}
\expandafter\def\csname gz:lognormal_s1.2_prop:lev1\endcsname{0.0}
\expandafter\def\csname gz:lognormal_s1.2_prop:medlev\endcsname{0.37}
\expandafter\def\csname gz:lognormal_s1.2_prop:rankdef\endcsname{0.0}
\expandafter\def\csname gz:lognormal_s1.2_prop:reps\endcsname{20{,}000}
\expandafter\def\csname gz:lognormal_s1.2_prop_n2000:IREG_HC1:all\endcsname{94.6}
\expandafter\def\csname gz:lognormal_s1.2_prop_n2000:IREG_HC1:all:nointerval\endcsname{0.0}
\expandafter\def\csname gz:lognormal_s1.2_prop_n2000:IREG_HC1:full\endcsname{94.6}
\expandafter\def\csname gz:lognormal_s1.2_prop_n2000:IREG_HC1:full:nointerval\endcsname{0.0}
\expandafter\def\csname gz:lognormal_s1.2_prop_n2000:IREG_HC1:lev1\endcsname{--}
\expandafter\def\csname gz:lognormal_s1.2_prop_n2000:IREG_HC2:all\endcsname{94.9}
\expandafter\def\csname gz:lognormal_s1.2_prop_n2000:IREG_HC2:all:nointerval\endcsname{0.0}
\expandafter\def\csname gz:lognormal_s1.2_prop_n2000:IREG_HC2:full\endcsname{94.9}
\expandafter\def\csname gz:lognormal_s1.2_prop_n2000:IREG_HC2:full:nointerval\endcsname{0.0}
\expandafter\def\csname gz:lognormal_s1.2_prop_n2000:IREG_HC2:lev1\endcsname{--}
\expandafter\def\csname gz:lognormal_s1.2_prop_n2000:IREG_HC3:all\endcsname{95.6}
\expandafter\def\csname gz:lognormal_s1.2_prop_n2000:IREG_HC3:all:nointerval\endcsname{0.0}
\expandafter\def\csname gz:lognormal_s1.2_prop_n2000:IREG_HC3:full\endcsname{95.6}
\expandafter\def\csname gz:lognormal_s1.2_prop_n2000:IREG_HC3:full:nointerval\endcsname{0.0}
\expandafter\def\csname gz:lognormal_s1.2_prop_n2000:IREG_HC3:lev1\endcsname{--}
\expandafter\def\csname gz:lognormal_s1.2_prop_n2000:REG_HC1:all\endcsname{95.4}
\expandafter\def\csname gz:lognormal_s1.2_prop_n2000:REG_HC1:all:nointerval\endcsname{0.0}
\expandafter\def\csname gz:lognormal_s1.2_prop_n2000:REG_HC1:full\endcsname{95.4}
\expandafter\def\csname gz:lognormal_s1.2_prop_n2000:REG_HC1:full:nointerval\endcsname{0.0}
\expandafter\def\csname gz:lognormal_s1.2_prop_n2000:REG_HC1:lev1\endcsname{--}
\expandafter\def\csname gz:lognormal_s1.2_prop_n2000:REG_HC3:all\endcsname{95.8}
\expandafter\def\csname gz:lognormal_s1.2_prop_n2000:REG_HC3:all:nointerval\endcsname{0.0}
\expandafter\def\csname gz:lognormal_s1.2_prop_n2000:REG_HC3:full\endcsname{95.8}
\expandafter\def\csname gz:lognormal_s1.2_prop_n2000:REG_HC3:full:nointerval\endcsname{0.0}
\expandafter\def\csname gz:lognormal_s1.2_prop_n2000:REG_HC3:lev1\endcsname{--}
\expandafter\def\csname gz:lognormal_s1.2_prop_n2000:lev1\endcsname{0.0}
\expandafter\def\csname gz:lognormal_s1.2_prop_n2000:medlev\endcsname{0.24}
\expandafter\def\csname gz:lognormal_s1.2_prop_n2000:rankdef\endcsname{0.0}
\expandafter\def\csname gz:lognormal_s1.2_prop_n2000:reps\endcsname{10{,}000}
\expandafter\def\csname gz:lognormal_s1.2_prop_n200:IREG_HC1:all\endcsname{91.3}
\expandafter\def\csname gz:lognormal_s1.2_prop_n200:IREG_HC1:all:nointerval\endcsname{0.0}
\expandafter\def\csname gz:lognormal_s1.2_prop_n200:IREG_HC1:full\endcsname{91.3}
\expandafter\def\csname gz:lognormal_s1.2_prop_n200:IREG_HC1:full:nointerval\endcsname{0.0}
\expandafter\def\csname gz:lognormal_s1.2_prop_n200:IREG_HC1:lev1\endcsname{--}
\expandafter\def\csname gz:lognormal_s1.2_prop_n200:IREG_HC2:all\endcsname{92.8}
\expandafter\def\csname gz:lognormal_s1.2_prop_n200:IREG_HC2:all:nointerval\endcsname{0.0}
\expandafter\def\csname gz:lognormal_s1.2_prop_n200:IREG_HC2:full\endcsname{92.8}
\expandafter\def\csname gz:lognormal_s1.2_prop_n200:IREG_HC2:full:nointerval\endcsname{0.0}
\expandafter\def\csname gz:lognormal_s1.2_prop_n200:IREG_HC2:lev1\endcsname{--}
\expandafter\def\csname gz:lognormal_s1.2_prop_n200:IREG_HC3:all\endcsname{94.9}
\expandafter\def\csname gz:lognormal_s1.2_prop_n200:IREG_HC3:all:nointerval\endcsname{0.0}
\expandafter\def\csname gz:lognormal_s1.2_prop_n200:IREG_HC3:full\endcsname{94.9}
\expandafter\def\csname gz:lognormal_s1.2_prop_n200:IREG_HC3:full:nointerval\endcsname{0.0}
\expandafter\def\csname gz:lognormal_s1.2_prop_n200:IREG_HC3:lev1\endcsname{--}
\expandafter\def\csname gz:lognormal_s1.2_prop_n200:REG_HC1:all\endcsname{94.8}
\expandafter\def\csname gz:lognormal_s1.2_prop_n200:REG_HC1:all:nointerval\endcsname{0.0}
\expandafter\def\csname gz:lognormal_s1.2_prop_n200:REG_HC1:full\endcsname{94.8}
\expandafter\def\csname gz:lognormal_s1.2_prop_n200:REG_HC1:full:nointerval\endcsname{0.0}
\expandafter\def\csname gz:lognormal_s1.2_prop_n200:REG_HC1:lev1\endcsname{--}
\expandafter\def\csname gz:lognormal_s1.2_prop_n200:REG_HC3:all\endcsname{96.4}
\expandafter\def\csname gz:lognormal_s1.2_prop_n200:REG_HC3:all:nointerval\endcsname{0.0}
\expandafter\def\csname gz:lognormal_s1.2_prop_n200:REG_HC3:full\endcsname{96.4}
\expandafter\def\csname gz:lognormal_s1.2_prop_n200:REG_HC3:full:nointerval\endcsname{0.0}
\expandafter\def\csname gz:lognormal_s1.2_prop_n200:REG_HC3:lev1\endcsname{--}
\expandafter\def\csname gz:lognormal_s1.2_prop_n200:lev1\endcsname{0.0}
\expandafter\def\csname gz:lognormal_s1.2_prop_n200:medlev\endcsname{0.47}
\expandafter\def\csname gz:lognormal_s1.2_prop_n200:rankdef\endcsname{0.0}
\expandafter\def\csname gz:lognormal_s1.2_prop_n200:reps\endcsname{20{,}000}
\expandafter\def\csname gz:lognormal_s1.2_prop_n8000:IREG_HC1:all\endcsname{94.5}
\expandafter\def\csname gz:lognormal_s1.2_prop_n8000:IREG_HC1:all:nointerval\endcsname{0.0}
\expandafter\def\csname gz:lognormal_s1.2_prop_n8000:IREG_HC1:full\endcsname{94.5}
\expandafter\def\csname gz:lognormal_s1.2_prop_n8000:IREG_HC1:full:nointerval\endcsname{0.0}
\expandafter\def\csname gz:lognormal_s1.2_prop_n8000:IREG_HC1:lev1\endcsname{--}
\expandafter\def\csname gz:lognormal_s1.2_prop_n8000:IREG_HC2:all\endcsname{94.8}
\expandafter\def\csname gz:lognormal_s1.2_prop_n8000:IREG_HC2:all:nointerval\endcsname{0.0}
\expandafter\def\csname gz:lognormal_s1.2_prop_n8000:IREG_HC2:full\endcsname{94.8}
\expandafter\def\csname gz:lognormal_s1.2_prop_n8000:IREG_HC2:full:nointerval\endcsname{0.0}
\expandafter\def\csname gz:lognormal_s1.2_prop_n8000:IREG_HC2:lev1\endcsname{--}
\expandafter\def\csname gz:lognormal_s1.2_prop_n8000:IREG_HC3:all\endcsname{94.9}
\expandafter\def\csname gz:lognormal_s1.2_prop_n8000:IREG_HC3:all:nointerval\endcsname{0.0}
\expandafter\def\csname gz:lognormal_s1.2_prop_n8000:IREG_HC3:full\endcsname{94.9}
\expandafter\def\csname gz:lognormal_s1.2_prop_n8000:IREG_HC3:full:nointerval\endcsname{0.0}
\expandafter\def\csname gz:lognormal_s1.2_prop_n8000:IREG_HC3:lev1\endcsname{--}
\expandafter\def\csname gz:lognormal_s1.2_prop_n8000:REG_HC1:all\endcsname{94.8}
\expandafter\def\csname gz:lognormal_s1.2_prop_n8000:REG_HC1:all:nointerval\endcsname{0.0}
\expandafter\def\csname gz:lognormal_s1.2_prop_n8000:REG_HC1:full\endcsname{94.8}
\expandafter\def\csname gz:lognormal_s1.2_prop_n8000:REG_HC1:full:nointerval\endcsname{0.0}
\expandafter\def\csname gz:lognormal_s1.2_prop_n8000:REG_HC1:lev1\endcsname{--}
\expandafter\def\csname gz:lognormal_s1.2_prop_n8000:REG_HC3:all\endcsname{95.0}
\expandafter\def\csname gz:lognormal_s1.2_prop_n8000:REG_HC3:all:nointerval\endcsname{0.0}
\expandafter\def\csname gz:lognormal_s1.2_prop_n8000:REG_HC3:full\endcsname{95.0}
\expandafter\def\csname gz:lognormal_s1.2_prop_n8000:REG_HC3:full:nointerval\endcsname{0.0}
\expandafter\def\csname gz:lognormal_s1.2_prop_n8000:REG_HC3:lev1\endcsname{--}
\expandafter\def\csname gz:lognormal_s1.2_prop_n8000:lev1\endcsname{0.0}
\expandafter\def\csname gz:lognormal_s1.2_prop_n8000:medlev\endcsname{0.14}
\expandafter\def\csname gz:lognormal_s1.2_prop_n8000:rankdef\endcsname{0.0}
\expandafter\def\csname gz:lognormal_s1.2_prop_n8000:reps\endcsname{4{,}000}
\expandafter\def\csname gz:lognormal_s1.5_prop:IREG_HC1:all\endcsname{89.9}
\expandafter\def\csname gz:lognormal_s1.5_prop:IREG_HC1:all:nointerval\endcsname{0.0}
\expandafter\def\csname gz:lognormal_s1.5_prop:IREG_HC1:full\endcsname{89.9}
\expandafter\def\csname gz:lognormal_s1.5_prop:IREG_HC1:full:nointerval\endcsname{0.0}
\expandafter\def\csname gz:lognormal_s1.5_prop:IREG_HC1:lev1\endcsname{--}
\expandafter\def\csname gz:lognormal_s1.5_prop:IREG_HC2:all\endcsname{91.9}
\expandafter\def\csname gz:lognormal_s1.5_prop:IREG_HC2:all:nointerval\endcsname{0.0}
\expandafter\def\csname gz:lognormal_s1.5_prop:IREG_HC2:full\endcsname{91.9}
\expandafter\def\csname gz:lognormal_s1.5_prop:IREG_HC2:full:nointerval\endcsname{0.0}
\expandafter\def\csname gz:lognormal_s1.5_prop:IREG_HC2:lev1\endcsname{--}
\expandafter\def\csname gz:lognormal_s1.5_prop:IREG_HC3:all\endcsname{94.7}
\expandafter\def\csname gz:lognormal_s1.5_prop:IREG_HC3:all:nointerval\endcsname{0.0}
\expandafter\def\csname gz:lognormal_s1.5_prop:IREG_HC3:full\endcsname{94.7}
\expandafter\def\csname gz:lognormal_s1.5_prop:IREG_HC3:full:nointerval\endcsname{0.0}
\expandafter\def\csname gz:lognormal_s1.5_prop:IREG_HC3:lev1\endcsname{--}
\expandafter\def\csname gz:lognormal_s1.5_prop:REG_HC1:all\endcsname{95.0}
\expandafter\def\csname gz:lognormal_s1.5_prop:REG_HC1:all:nointerval\endcsname{0.0}
\expandafter\def\csname gz:lognormal_s1.5_prop:REG_HC1:full\endcsname{95.0}
\expandafter\def\csname gz:lognormal_s1.5_prop:REG_HC1:full:nointerval\endcsname{0.0}
\expandafter\def\csname gz:lognormal_s1.5_prop:REG_HC1:lev1\endcsname{--}
\expandafter\def\csname gz:lognormal_s1.5_prop:REG_HC3:all\endcsname{96.7}
\expandafter\def\csname gz:lognormal_s1.5_prop:REG_HC3:all:nointerval\endcsname{0.0}
\expandafter\def\csname gz:lognormal_s1.5_prop:REG_HC3:full\endcsname{96.7}
\expandafter\def\csname gz:lognormal_s1.5_prop:REG_HC3:full:nointerval\endcsname{0.0}
\expandafter\def\csname gz:lognormal_s1.5_prop:REG_HC3:lev1\endcsname{--}
\expandafter\def\csname gz:lognormal_s1.5_prop:lev1\endcsname{0.0}
\expandafter\def\csname gz:lognormal_s1.5_prop:medlev\endcsname{0.48}
\expandafter\def\csname gz:lognormal_s1.5_prop:rankdef\endcsname{0.0}
\expandafter\def\csname gz:lognormal_s1.5_prop:reps\endcsname{20{,}000}
\expandafter\def\csname gz:lognormal_s1.8_prop:IREG_HC1:all\endcsname{86.1}
\expandafter\def\csname gz:lognormal_s1.8_prop:IREG_HC1:all:nointerval\endcsname{0.0}
\expandafter\def\csname gz:lognormal_s1.8_prop:IREG_HC1:full\endcsname{86.1}
\expandafter\def\csname gz:lognormal_s1.8_prop:IREG_HC1:full:nointerval\endcsname{0.0}
\expandafter\def\csname gz:lognormal_s1.8_prop:IREG_HC1:lev1\endcsname{--}
\expandafter\def\csname gz:lognormal_s1.8_prop:IREG_HC2:all\endcsname{89.2}
\expandafter\def\csname gz:lognormal_s1.8_prop:IREG_HC2:all:nointerval\endcsname{0.0}
\expandafter\def\csname gz:lognormal_s1.8_prop:IREG_HC2:full\endcsname{89.2}
\expandafter\def\csname gz:lognormal_s1.8_prop:IREG_HC2:full:nointerval\endcsname{0.0}
\expandafter\def\csname gz:lognormal_s1.8_prop:IREG_HC2:lev1\endcsname{--}
\expandafter\def\csname gz:lognormal_s1.8_prop:IREG_HC3:all\endcsname{94.1}
\expandafter\def\csname gz:lognormal_s1.8_prop:IREG_HC3:all:nointerval\endcsname{0.0}
\expandafter\def\csname gz:lognormal_s1.8_prop:IREG_HC3:full\endcsname{94.1}
\expandafter\def\csname gz:lognormal_s1.8_prop:IREG_HC3:full:nointerval\endcsname{0.0}
\expandafter\def\csname gz:lognormal_s1.8_prop:IREG_HC3:lev1\endcsname{--}
\expandafter\def\csname gz:lognormal_s1.8_prop:REG_HC1:all\endcsname{95.1}
\expandafter\def\csname gz:lognormal_s1.8_prop:REG_HC1:all:nointerval\endcsname{0.0}
\expandafter\def\csname gz:lognormal_s1.8_prop:REG_HC1:full\endcsname{95.1}
\expandafter\def\csname gz:lognormal_s1.8_prop:REG_HC1:full:nointerval\endcsname{0.0}
\expandafter\def\csname gz:lognormal_s1.8_prop:REG_HC1:lev1\endcsname{--}
\expandafter\def\csname gz:lognormal_s1.8_prop:REG_HC3:all\endcsname{97.4}
\expandafter\def\csname gz:lognormal_s1.8_prop:REG_HC3:all:nointerval\endcsname{0.0}
\expandafter\def\csname gz:lognormal_s1.8_prop:REG_HC3:full\endcsname{97.4}
\expandafter\def\csname gz:lognormal_s1.8_prop:REG_HC3:full:nointerval\endcsname{0.0}
\expandafter\def\csname gz:lognormal_s1.8_prop:REG_HC3:lev1\endcsname{--}
\expandafter\def\csname gz:lognormal_s1.8_prop:lev1\endcsname{0.0}
\expandafter\def\csname gz:lognormal_s1.8_prop:medlev\endcsname{0.59}
\expandafter\def\csname gz:lognormal_s1.8_prop:rankdef\endcsname{0.0}
\expandafter\def\csname gz:lognormal_s1.8_prop:reps\endcsname{20{,}000}
\expandafter\def\csname gz:lognormal_s2.0_prop:IREG_HC1:all\endcsname{83.0}
\expandafter\def\csname gz:lognormal_s2.0_prop:IREG_HC1:all:nointerval\endcsname{0.0}
\expandafter\def\csname gz:lognormal_s2.0_prop:IREG_HC1:full\endcsname{83.0}
\expandafter\def\csname gz:lognormal_s2.0_prop:IREG_HC1:full:nointerval\endcsname{0.0}
\expandafter\def\csname gz:lognormal_s2.0_prop:IREG_HC1:lev1\endcsname{--}
\expandafter\def\csname gz:lognormal_s2.0_prop:IREG_HC2:all\endcsname{87.1}
\expandafter\def\csname gz:lognormal_s2.0_prop:IREG_HC2:all:nointerval\endcsname{0.0}
\expandafter\def\csname gz:lognormal_s2.0_prop:IREG_HC2:full\endcsname{87.1}
\expandafter\def\csname gz:lognormal_s2.0_prop:IREG_HC2:full:nointerval\endcsname{0.0}
\expandafter\def\csname gz:lognormal_s2.0_prop:IREG_HC2:lev1\endcsname{--}
\expandafter\def\csname gz:lognormal_s2.0_prop:IREG_HC3:all\endcsname{93.9}
\expandafter\def\csname gz:lognormal_s2.0_prop:IREG_HC3:all:nointerval\endcsname{0.0}
\expandafter\def\csname gz:lognormal_s2.0_prop:IREG_HC3:full\endcsname{93.9}
\expandafter\def\csname gz:lognormal_s2.0_prop:IREG_HC3:full:nointerval\endcsname{0.0}
\expandafter\def\csname gz:lognormal_s2.0_prop:IREG_HC3:lev1\endcsname{--}
\expandafter\def\csname gz:lognormal_s2.0_prop:REG_HC1:all\endcsname{95.3}
\expandafter\def\csname gz:lognormal_s2.0_prop:REG_HC1:all:nointerval\endcsname{0.0}
\expandafter\def\csname gz:lognormal_s2.0_prop:REG_HC1:full\endcsname{95.3}
\expandafter\def\csname gz:lognormal_s2.0_prop:REG_HC1:full:nointerval\endcsname{0.0}
\expandafter\def\csname gz:lognormal_s2.0_prop:REG_HC1:lev1\endcsname{--}
\expandafter\def\csname gz:lognormal_s2.0_prop:REG_HC3:all\endcsname{97.9}
\expandafter\def\csname gz:lognormal_s2.0_prop:REG_HC3:all:nointerval\endcsname{0.0}
\expandafter\def\csname gz:lognormal_s2.0_prop:REG_HC3:full\endcsname{97.9}
\expandafter\def\csname gz:lognormal_s2.0_prop:REG_HC3:full:nointerval\endcsname{0.0}
\expandafter\def\csname gz:lognormal_s2.0_prop:REG_HC3:lev1\endcsname{--}
\expandafter\def\csname gz:lognormal_s2.0_prop:lev1\endcsname{0.0}
\expandafter\def\csname gz:lognormal_s2.0_prop:medlev\endcsname{0.65}
\expandafter\def\csname gz:lognormal_s2.0_prop:rankdef\endcsname{0.0}
\expandafter\def\csname gz:lognormal_s2.0_prop:reps\endcsname{20{,}000}
\expandafter\def\csname gz:mc:heavy:se\endcsname{0.033}
\expandafter\def\csname gz:mcse:max20k\endcsname{0.27}
\expandafter\def\csname gz:mcse:maxall\endcsname{0.36}
\expandafter\def\csname gz:rare:theory:lev1\endcsname{6.39}
\expandafter\def\csname gz:rare:theory:rankdef\endcsname{1.28}
\expandafter\def\csname gz:rare_q0.02_homo:IREG_HC1:all\endcsname{94.7}
\expandafter\def\csname gz:rare_q0.02_homo:IREG_HC1:all:nointerval\endcsname{0.2}
\expandafter\def\csname gz:rare_q0.02_homo:IREG_HC1:full\endcsname{94.9}
\expandafter\def\csname gz:rare_q0.02_homo:IREG_HC1:full:nointerval\endcsname{0.0}
\expandafter\def\csname gz:rare_q0.02_homo:IREG_HC1:lev1\endcsname{95.7}
\expandafter\def\csname gz:rare_q0.02_homo:IREG_HC1:lev1:nointerval\endcsname{0.0}
\expandafter\def\csname gz:rare_q0.02_homo:IREG_HC2:all\endcsname{90.9}
\expandafter\def\csname gz:rare_q0.02_homo:IREG_HC2:all:nointerval\endcsname{4.2}
\expandafter\def\csname gz:rare_q0.02_homo:IREG_HC2:full\endcsname{91.2}
\expandafter\def\csname gz:rare_q0.02_homo:IREG_HC2:full:nointerval\endcsname{3.8}
\expandafter\def\csname gz:rare_q0.02_homo:IREG_HC2:lev1\endcsname{41.3}
\expandafter\def\csname gz:rare_q0.02_homo:IREG_HC2:lev1:nointerval\endcsname{56.4}
\expandafter\def\csname gz:rare_q0.02_homo:IREG_HC3:all\endcsname{94.7}
\expandafter\def\csname gz:rare_q0.02_homo:IREG_HC3:all:nointerval\endcsname{0.2}
\expandafter\def\csname gz:rare_q0.02_homo:IREG_HC3:full\endcsname{95.2}
\expandafter\def\csname gz:rare_q0.02_homo:IREG_HC3:full:nointerval\endcsname{0.0}
\expandafter\def\csname gz:rare_q0.02_homo:IREG_HC3:lev1\endcsname{97.6}
\expandafter\def\csname gz:rare_q0.02_homo:IREG_HC3:lev1:nointerval\endcsname{0.0}
\expandafter\def\csname gz:rare_q0.02_homo:REG_HC1:all\endcsname{95.0}
\expandafter\def\csname gz:rare_q0.02_homo:REG_HC1:all:nointerval\endcsname{0.0}
\expandafter\def\csname gz:rare_q0.02_homo:REG_HC1:full\endcsname{94.9}
\expandafter\def\csname gz:rare_q0.02_homo:REG_HC1:full:nointerval\endcsname{0.0}
\expandafter\def\csname gz:rare_q0.02_homo:REG_HC1:lev1\endcsname{95.8}
\expandafter\def\csname gz:rare_q0.02_homo:REG_HC1:lev1:nointerval\endcsname{0.0}
\expandafter\def\csname gz:rare_q0.02_homo:REG_HC3:all\endcsname{95.0}
\expandafter\def\csname gz:rare_q0.02_homo:REG_HC3:all:nointerval\endcsname{0.0}
\expandafter\def\csname gz:rare_q0.02_homo:REG_HC3:full\endcsname{95.0}
\expandafter\def\csname gz:rare_q0.02_homo:REG_HC3:full:nointerval\endcsname{0.0}
\expandafter\def\csname gz:rare_q0.02_homo:REG_HC3:lev1\endcsname{95.9}
\expandafter\def\csname gz:rare_q0.02_homo:REG_HC3:lev1:nointerval\endcsname{0.0}
\expandafter\def\csname gz:rare_q0.02_homo:lev1\endcsname{6.6}
\expandafter\def\csname gz:rare_q0.02_homo:medlev\endcsname{0.25}
\expandafter\def\csname gz:rare_q0.02_homo:rankdef\endcsname{1.3}
\expandafter\def\csname gz:rare_q0.02_homo:reps\endcsname{20{,}000}
\expandafter\def\csname gz:rare_q0.02_homo@v2:IREG_HC1:all\endcsname{94.7}
\expandafter\def\csname gz:rare_q0.02_homo@v2:IREG_HC1:full\endcsname{94.9}
\expandafter\def\csname gz:rare_q0.02_homo@v2:IREG_HC1:lev1\endcsname{95.7}
\expandafter\def\csname gz:rare_q0.02_homo@v2:IREG_HC2:all\endcsname{94.7}
\expandafter\def\csname gz:rare_q0.02_homo@v2:IREG_HC2:full\endcsname{94.9}
\expandafter\def\csname gz:rare_q0.02_homo@v2:IREG_HC2:lev1\endcsname{95.7}
\expandafter\def\csname gz:rare_q0.02_homo@v2:IREG_HC3:all\endcsname{94.9}
\expandafter\def\csname gz:rare_q0.02_homo@v2:IREG_HC3:full\endcsname{95.1}
\expandafter\def\csname gz:rare_q0.02_homo@v2:IREG_HC3:lev1\endcsname{97.1}
\expandafter\def\csname gz:rare_q0.02_prop:IREG_HC1:all\endcsname{92.0}
\expandafter\def\csname gz:rare_q0.02_prop:IREG_HC1:all:nointerval\endcsname{0.3}
\expandafter\def\csname gz:rare_q0.02_prop:IREG_HC1:full\endcsname{92.2}
\expandafter\def\csname gz:rare_q0.02_prop:IREG_HC1:full:nointerval\endcsname{0.0}
\expandafter\def\csname gz:rare_q0.02_prop:IREG_HC1:lev1\endcsname{85.6}
\expandafter\def\csname gz:rare_q0.02_prop:IREG_HC1:lev1:nointerval\endcsname{0.0}
\expandafter\def\csname gz:rare_q0.02_prop:IREG_HC2:all\endcsname{90.1}
\expandafter\def\csname gz:rare_q0.02_prop:IREG_HC2:all:nointerval\endcsname{3.7}
\expandafter\def\csname gz:rare_q0.02_prop:IREG_HC2:full\endcsname{90.5}
\expandafter\def\csname gz:rare_q0.02_prop:IREG_HC2:full:nointerval\endcsname{3.3}
\expandafter\def\csname gz:rare_q0.02_prop:IREG_HC2:lev1\endcsname{40.2}
\expandafter\def\csname gz:rare_q0.02_prop:IREG_HC2:lev1:nointerval\endcsname{53.6}
\expandafter\def\csname gz:rare_q0.02_prop:IREG_HC3:all\endcsname{94.8}
\expandafter\def\csname gz:rare_q0.02_prop:IREG_HC3:all:nointerval\endcsname{0.2}
\expandafter\def\csname gz:rare_q0.02_prop:IREG_HC3:full\endcsname{95.1}
\expandafter\def\csname gz:rare_q0.02_prop:IREG_HC3:full:nointerval\endcsname{0.0}
\expandafter\def\csname gz:rare_q0.02_prop:IREG_HC3:lev1\endcsname{94.9}
\expandafter\def\csname gz:rare_q0.02_prop:IREG_HC3:lev1:nointerval\endcsname{0.0}
\expandafter\def\csname gz:rare_q0.02_prop:REG_HC1:all\endcsname{94.9}
\expandafter\def\csname gz:rare_q0.02_prop:REG_HC1:all:nointerval\endcsname{0.0}
\expandafter\def\csname gz:rare_q0.02_prop:REG_HC1:full\endcsname{94.9}
\expandafter\def\csname gz:rare_q0.02_prop:REG_HC1:full:nointerval\endcsname{0.0}
\expandafter\def\csname gz:rare_q0.02_prop:REG_HC1:lev1\endcsname{95.7}
\expandafter\def\csname gz:rare_q0.02_prop:REG_HC1:lev1:nointerval\endcsname{0.0}
\expandafter\def\csname gz:rare_q0.02_prop:REG_HC3:all\endcsname{96.0}
\expandafter\def\csname gz:rare_q0.02_prop:REG_HC3:all:nointerval\endcsname{0.0}
\expandafter\def\csname gz:rare_q0.02_prop:REG_HC3:full\endcsname{96.0}
\expandafter\def\csname gz:rare_q0.02_prop:REG_HC3:full:nointerval\endcsname{0.0}
\expandafter\def\csname gz:rare_q0.02_prop:REG_HC3:lev1\endcsname{96.6}
\expandafter\def\csname gz:rare_q0.02_prop:REG_HC3:lev1:nointerval\endcsname{0.0}
\expandafter\def\csname gz:rare_q0.02_prop:lev1\endcsname{6.1}
\expandafter\def\csname gz:rare_q0.02_prop:medlev\endcsname{0.25}
\expandafter\def\csname gz:rare_q0.02_prop:rankdef\endcsname{1.4}
\expandafter\def\csname gz:rare_q0.02_prop:reps\endcsname{20{,}000}
\expandafter\def\csname gz:rare_q0.02_prop@v2:IREG_HC1:all\endcsname{92.0}
\expandafter\def\csname gz:rare_q0.02_prop@v2:IREG_HC1:full\endcsname{92.2}
\expandafter\def\csname gz:rare_q0.02_prop@v2:IREG_HC1:lev1\endcsname{85.6}
\expandafter\def\csname gz:rare_q0.02_prop@v2:IREG_HC2:all\endcsname{93.1}
\expandafter\def\csname gz:rare_q0.02_prop@v2:IREG_HC2:full\endcsname{93.3}
\expandafter\def\csname gz:rare_q0.02_prop@v2:IREG_HC2:lev1\endcsname{85.8}
\expandafter\def\csname gz:rare_q0.02_prop@v2:IREG_HC3:all\endcsname{94.6}
\expandafter\def\csname gz:rare_q0.02_prop@v2:IREG_HC3:full\endcsname{94.9}
\expandafter\def\csname gz:rare_q0.02_prop@v2:IREG_HC3:lev1\endcsname{90.5}
\expandafter\def\csname gz:samp:s0.6:mseratio\endcsname{1.01}
\expandafter\def\csname gz:samp:s0.6:mseratio:se\endcsname{0.000}
\expandafter\def\csname gz:samp:s0.6:pred\endcsname{1.01}
\expandafter\def\csname gz:samp:s0.6:sdI\endcsname{0.21}
\expandafter\def\csname gz:samp:s0.6:sdR\endcsname{0.21}
\expandafter\def\csname gz:samp:s0.6:sdpct\endcsname{0}
\expandafter\def\csname gz:samp:s0.6:sdtheory\endcsname{0.21}
\expandafter\def\csname gz:samp:s0.6:seI\endcsname{0.21}
\expandafter\def\csname gz:samp:s0.6:seR\endcsname{0.21}
\expandafter\def\csname gz:samp:s0.6:sepct\endcsname{0}
\expandafter\def\csname gz:samp:s0.8:mseratio\endcsname{1.02}
\expandafter\def\csname gz:samp:s0.8:mseratio:se\endcsname{0.001}
\expandafter\def\csname gz:samp:s0.8:pred\endcsname{1.03}
\expandafter\def\csname gz:samp:s0.8:sdI\endcsname{0.24}
\expandafter\def\csname gz:samp:s0.8:sdR\endcsname{0.24}
\expandafter\def\csname gz:samp:s0.8:sdpct\endcsname{1}
\expandafter\def\csname gz:samp:s0.8:sdtheory\endcsname{0.24}
\expandafter\def\csname gz:samp:s0.8:seI\endcsname{0.24}
\expandafter\def\csname gz:samp:s0.8:seR\endcsname{0.24}
\expandafter\def\csname gz:samp:s0.8:sepct\endcsname{0}
\expandafter\def\csname gz:samp:s1.0:mseratio\endcsname{1.04}
\expandafter\def\csname gz:samp:s1.0:mseratio:se\endcsname{0.002}
\expandafter\def\csname gz:samp:s1.0:pred\endcsname{1.09}
\expandafter\def\csname gz:samp:s1.0:sdI\endcsname{0.30}
\expandafter\def\csname gz:samp:s1.0:sdR\endcsname{0.30}
\expandafter\def\csname gz:samp:s1.0:sdpct\endcsname{2}
\expandafter\def\csname gz:samp:s1.0:sdtheory\endcsname{0.31}
\expandafter\def\csname gz:samp:s1.0:seI\endcsname{0.29}
\expandafter\def\csname gz:samp:s1.0:seR\endcsname{0.29}
\expandafter\def\csname gz:samp:s1.0:sepct\endcsname{0}
\expandafter\def\csname gz:samp:s1.2:mseratio\endcsname{1.10}
\expandafter\def\csname gz:samp:s1.2:mseratio:se\endcsname{0.003}
\expandafter\def\csname gz:samp:s1.2:pred\endcsname{1.34}
\expandafter\def\csname gz:samp:s1.2:sdI\endcsname{0.42}
\expandafter\def\csname gz:samp:s1.2:sdR\endcsname{0.40}
\expandafter\def\csname gz:samp:s1.2:sdpct\endcsname{5}
\expandafter\def\csname gz:samp:s1.2:sdtheory\endcsname{0.43}
\expandafter\def\csname gz:samp:s1.2:seI\endcsname{0.37}
\expandafter\def\csname gz:samp:s1.2:seR\endcsname{0.37}
\expandafter\def\csname gz:samp:s1.2:sepct\endcsname{1}
\expandafter\def\csname gz:samp:s1.5:mseratio\endcsname{1.27}
\expandafter\def\csname gz:samp:s1.5:mseratio:se\endcsname{0.014}
\expandafter\def\csname gz:samp:s1.5:pred\endcsname{4.81}
\expandafter\def\csname gz:samp:s1.5:sdI\endcsname{0.81}
\expandafter\def\csname gz:samp:s1.5:sdR\endcsname{0.72}
\expandafter\def\csname gz:samp:s1.5:sdpct\endcsname{13}
\expandafter\def\csname gz:samp:s1.5:sdtheory\endcsname{0.88}
\expandafter\def\csname gz:samp:s1.5:seI\endcsname{0.62}
\expandafter\def\csname gz:samp:s1.5:seR\endcsname{0.61}
\expandafter\def\csname gz:samp:s1.5:sepct\endcsname{1}
\expandafter\def\csname gz:samp:s1.8:mseratio\endcsname{1.62}
\expandafter\def\csname gz:samp:s1.8:mseratio:se\endcsname{0.039}
\expandafter\def\csname gz:samp:s1.8:pred\endcsname{72}
\expandafter\def\csname gz:samp:s1.8:sdI\endcsname{1.86}
\expandafter\def\csname gz:samp:s1.8:sdR\endcsname{1.46}
\expandafter\def\csname gz:samp:s1.8:sdpct\endcsname{27}
\expandafter\def\csname gz:samp:s1.8:sdtheory\endcsname{2.30}
\expandafter\def\csname gz:samp:s1.8:seI\endcsname{1.13}
\expandafter\def\csname gz:samp:s1.8:seR\endcsname{1.12}
\expandafter\def\csname gz:samp:s1.8:sepct\endcsname{1}
\expandafter\def\csname gz:samp:s2.0:mseratio\endcsname{2.00}
\expandafter\def\csname gz:samp:s2.0:mseratio:se\endcsname{0.056}
\expandafter\def\csname gz:samp:s2.0:pred\endcsname{672}
\expandafter\def\csname gz:samp:s2.0:sdI\endcsname{3.48}
\expandafter\def\csname gz:samp:s2.0:sdR\endcsname{2.46}
\expandafter\def\csname gz:samp:s2.0:sdpct\endcsname{41}
\expandafter\def\csname gz:samp:s2.0:sdtheory\endcsname{4.90}
\expandafter\def\csname gz:samp:s2.0:seI\endcsname{1.75}
\expandafter\def\csname gz:samp:s2.0:seR\endcsname{1.73}
\expandafter\def\csname gz:samp:s2.0:sepct\endcsname{1}
\expandafter\def\csname gz:tail:g0:4\endcsname{2.45}
\expandafter\def\csname gz:tail:g0:6\endcsname{17.3}
\expandafter\def\csname gz:tail:g0:8\endcsname{97.4}
\expandafter\def\csname gz:tail:g0:full\endcsname{672.1}
\expandafter\def\csname gz:tail:g0:sim\endcsname{2.00}
\expandafter\def\csname gz:tail:heavy:4\endcsname{1.90}
\expandafter\def\csname gz:tail:heavy:6\endcsname{7.94}
\expandafter\def\csname gz:tail:heavy:8\endcsname{26.6}
\expandafter\def\csname gz:tail:heavy:full\endcsname{78.5}
\expandafter\def\csname gz:tail:heavy:sim\endcsname{1.68}
\expandafter\def\csname gz:tail:t9:4\endcsname{1.36}
\expandafter\def\csname gz:tail:t9:6\endcsname{2.44}
\expandafter\def\csname gz:tail:t9:8\endcsname{3.81}
\expandafter\def\csname gz:tail:t9:full\endcsname{5.05}
\expandafter\def\csname gz:tail:t9:sim\endcsname{1.25}
\expandafter\def\csname gz:tr:X:IREG_HC1:cov\endcsname{83.0}
\expandafter\def\csname gz:tr:X:IREG_HC1:width\endcsname{6.87}
\expandafter\def\csname gz:tr:X:IREG_HC3:cov\endcsname{93.9}
\expandafter\def\csname gz:tr:X:IREG_HC3:width\endcsname{9.14}
\expandafter\def\csname gz:tr:X:REG_HC1:cov\endcsname{95.3}
\expandafter\def\csname gz:tr:X:REG_HC1:width\endcsname{6.80}
\expandafter\def\csname gz:tr:X:REG_HC3:cov\endcsname{97.9}
\expandafter\def\csname gz:tr:X:REG_HC3:width\endcsname{7.57}
\expandafter\def\csname gz:tr:X:lev\endcsname{0.65}
\expandafter\def\csname gz:tr:log1pX:IREG_HC1:cov\endcsname{95.6}
\expandafter\def\csname gz:tr:log1pX:IREG_HC1:width\endcsname{9.92}
\expandafter\def\csname gz:tr:log1pX:IREG_HC3:cov\endcsname{96.8}
\expandafter\def\csname gz:tr:log1pX:IREG_HC3:width\endcsname{10.48}
\expandafter\def\csname gz:tr:log1pX:REG_HC1:cov\endcsname{96.7}
\expandafter\def\csname gz:tr:log1pX:REG_HC1:width\endcsname{9.92}
\expandafter\def\csname gz:tr:log1pX:REG_HC3:cov\endcsname{97.2}
\expandafter\def\csname gz:tr:log1pX:REG_HC3:width\endcsname{10.21}
\expandafter\def\csname gz:tr:log1pX:lev\endcsname{0.08}
\expandafter\def\csname gz:val:crossgap\endcsname{\ensuremath{4\times10^{-14}}}
\expandafter\def\csname gz:val:crossn\endcsname{20}
\expandafter\def\csname gz:val:datasets\endcsname{100}
\expandafter\def\csname gz:val:formulagap\endcsname{0}
\expandafter\def\csname gz:val:full:dbern\endcsname{41.32}
\expandafter\def\csname gz:val:full:extrap\endcsname{45}
\expandafter\def\csname gz:val:full:inc1\endcsname{6.4}
\expandafter\def\csname gz:val:full:inc2\endcsname{3.1}
\expandafter\def\csname gz:val:handgap\endcsname{\ensuremath{3\times10^{-14}}}
\expandafter\def\csname gz:val:lev:hcthree:v1\endcsname{0.195}
\expandafter\def\csname gz:val:lev:hcthree:v2\endcsname{0.132}
\expandafter\def\csname gz:val:lev:hcthree:zero\endcsname{0.132}
\expandafter\def\csname gz:val:lev:hctwo\endcsname{0.129}
\expandafter\def\csname gz:val:lev:onem\endcsname{\ensuremath{1\times10^{-16}}}
\expandafter\def\csname gz:val:lev:resid\endcsname{\ensuremath{1\times10^{-14}}}
\expandafter\def\csname gz:val:p1gap\endcsname{\ensuremath{9\times10^{-16}}}
\expandafter\def\csname gz:val:plugin\endcsname{0.0076}
\expandafter\def\csname gz:valall:maxabsz\endcsname{2.60}
\expandafter\def\csname gz:valall:n\endcsname{18}
\expandafter\def\csname gz:valall:within196\endcsname{16}
\expandafter\def\csname gz:valid:s0.6:n2000:frac\endcsname{94}
\expandafter\def\csname gz:valid:s0.6:n500:frac\endcsname{91}
\expandafter\def\csname gz:valid:s0.8:n2000:frac\endcsname{85}
\expandafter\def\csname gz:valid:s0.8:n500:frac\endcsname{75}

%% file: tables/jk_battery_body.tex
\multicolumn{12}{l}{\emph{Small or no heterogeneity}} \\
Gaussian & 0.022 & 0.091 & 0.36 & 95.0 & 0.36 & 95.0 & 0.091 & 0.36 & 94.8 & 0.36 & 95.0 \\
Misspecified, $G=0$ & 0 & 0.098 & 0.38 & 94.9 & 0.38 & 95.0 & 0.098 & 0.38 & 94.8 & 0.38 & 95.0 \\
Skewed, het., $G=0$ & 0 & 0.167 & 0.65 & 95.0 & 0.66 & 95.4 & 0.165 & 0.64 & 95.0 & 0.65 & 95.5 \\
Heavy-tailed, het. & 0.016 & 0.175 & 0.64 & 94.7 & 0.66 & 95.5 & 0.168 & 0.63 & 94.7 & 0.64 & 95.5 \\
Lognormal $s=1.2$, $g=0.3$ & 0.013 & 0.425 & 1.52 & 94.6 & 1.57 & 95.6 & 0.408 & 1.47 & 94.6 & 1.52 & 95.7 \\
Lognormal $s=1.5$, $g=0.1$ & 0.002 & 0.798 & 2.55 & 94.7 & 2.73 & 96.2 & 0.715 & 2.40 & 94.9 & 2.55 & 96.7 \\
Lognormal $s=1.5$, $g=0.5$ & 0.049 & 0.839 & 2.62 & 93.7 & 2.79 & 95.3 & 0.753 & 2.46 & 93.2 & 2.61 & 95.1 \\
Lognormal $s=1.8$, $g=0.3$, $\sigma=1+1.5X$ & 0.009 & 2.826 & 7.26 & 94.3 & 8.04 & 96.4 & 2.196 & 6.57 & 94.3 & 7.18 & 96.6 \\
Lognormal $s=2$, $g=0.1$ & 0.002 & 3.433 & 7.76 & 94.8 & 8.78 & 97.0 & 2.458 & 6.80 & 95.0 & 7.55 & 97.6 \\
\addlinespace
\multicolumn{12}{l}{\emph{Moderate heterogeneity}} \\
Lognormal $s=1.5$, $g=1$ & 0.172 & 0.901 & 2.81 & 92.0 & 2.97 & 93.7 & 0.826 & 2.64 & 90.2 & 2.81 & 91.9 \\
Rare extreme values & 0.2 & 0.100 & 0.40 & 94.0 & 0.40 & 94.0 & 0.102 & 0.40 & 94.7 & 0.41 & 95.5 \\
\addlinespace
\multicolumn{12}{l}{\emph{Large heterogeneity}} \\
Heavy-tailed, het. & 0.365 & 0.220 & 0.75 & 94.2 & 0.77 & 94.6 & 0.212 & 0.73 & 91.9 & 0.76 & 93.0 \\
Lognormal $s=1.5$, $g=3$ & 0.651 & 1.451 & 4.46 & 89.8 & 4.57 & 90.4 & 1.380 & 4.17 & 84.1 & 4.51 & 86.0 \\
Gaussian & 0.692 & 0.161 & 0.64 & 95.1 & 0.64 & 95.1 & 0.162 & 0.64 & 95.0 & 0.64 & 95.2 \\

%% file: tables/jk_calibration_body.tex
Gaussian & 0.022 & 80.1 & 95.0 & 99.1 & 80.2 & 95.0 & 99.1 & 79.9 & 94.8 & 99.0 & 80.2 & 95.0 & 99.0 \\
Misspecified, $G=0$ & 0 & 79.8 & 94.9 & 99.0 & 79.8 & 95.0 & 99.0 & 79.6 & 94.8 & 98.9 & 79.9 & 95.0 & 99.0 \\
Skewed, het., $G=0$ & 0 & 80.1 & 95.0 & 99.0 & 80.8 & 95.4 & 99.1 & 80.0 & 95.0 & 98.9 & 80.8 & 95.5 & 99.1 \\
Heavy-tailed, het. & 0.016 & 79.1 & 94.7 & 98.9 & 80.8 & 95.5 & 99.1 & 79.0 & 94.7 & 98.9 & 80.9 & 95.5 & 99.1 \\
Lognormal $s=1.2$, $g=0.3$ & 0.013 & 78.4 & 94.6 & 98.9 & 80.6 & 95.6 & 99.2 & 78.6 & 94.6 & 98.9 & 80.8 & 95.7 & 99.2 \\
Lognormal $s=1.5$, $g=0.1$ & 0.002 & 77.0 & 94.7 & 99.3 & 81.2 & 96.2 & 99.6 & 77.8 & 94.9 & 99.3 & 82.5 & 96.7 & 99.6 \\
Lognormal $s=1.5$, $g=0.5$ & 0.049 & 76.6 & 93.7 & 98.5 & 80.4 & 95.3 & 98.9 & 76.4 & 93.2 & 98.3 & 80.6 & 95.1 & 98.8 \\
Lognormal $s=1.8$, $g=0.3$, $\sigma=1+1.5X$ & 0.009 & 75.5 & 94.3 & 99.1 & 81.5 & 96.4 & 99.5 & 76.4 & 94.3 & 99.0 & 83.1 & 96.6 & 99.5 \\
Lognormal $s=2$, $g=0.1$ & 0.002 & 75.1 & 94.8 & 99.3 & 82.3 & 97.0 & 99.6 & 76.1 & 95.0 & 99.3 & 84.2 & 97.6 & 99.7 \\

%% file: tables/validation_body.tex
Gaussian, homoskedastic, $p=3$ ($F$) & $-5.00$ & $-4.70\ (0.09)$ & $-4.99\ (0.09)$ & $-5.17\ (0.16)$ \\
Exponential covariates, homoskedastic ($F,Q,R$) & $-20.25$ & $-17.17\ (0.21)$ & $-19.60\ (0.22)$ & $-19.75\ (0.39)$ \\
Exponential, heteroskedastic, $\gamma=0$ & $32.00$ & $29.82\ (0.28)$ & $31.73\ (0.29)$ & $33.25\ (0.54)$ \\
Exponential, HTE + arm-specific misspec.\ (all) & $44.52$ & $32.40\ (0.56)$ & $38.82\ (0.55)$ & $41.97\ (0.98)$ \\
Gaussian, arm-specific misspec.\ ($\tr E^2,E_W$) & $-2.32$ & $-2.36\ (0.09)$ & $-2.34\ (0.09)$ & $-2.45\ (0.16)$ \\
Mixed skewed/Gaussian, heteroskedastic & $28.00$ & $25.94\ (0.27)$ & $27.85\ (0.28)$ & $28.27\ (0.53)$ \\

%% file: tables/ablation_body.tex
Gaussian, homoskedastic, $p=3$ ($F$) & $-1.0$ & $+29.8$ & $-1.0$ & $-1.0$ & $-1.0$ & $-1.0$ & $-1.0$ \\
Exponential covariates, homoskedastic ($F,Q,R$) & $+1.3$ & $+14.2$ & $-41.3$ & $-37.4$ & $-21.9$ & $+1.3$ & $+1.3$ \\
Exponential, heteroskedastic, $\gamma=0$ & $+2.3$ & $+2.3$ & $+61.1$ & $+2.3$ & $+2.3$ & $+61.1$ & $+2.3$ \\
Exponential, HTE + arm-specific misspec.\ (all) & $-2.6$ & $+0.7$ & $+45.1$ & $-12.4$ & $-3.2$ & $+65.8$ & $-17.7$ \\
Gaussian, arm-specific misspec.\ ($\tr E^2,E_W$) & $-0.8$ & $+23.5$ & $+3.4$ & $-0.8$ & $-0.8$ & $+3.4$ & $-0.8$ \\
Mixed skewed/Gaussian, heteroskedastic & $+0.5$ & $+4.2$ & $+55.7$ & $-5.1$ & $-1.4$ & $+60.3$ & $+0.5$ \\

%% file: tables/validation_families_body.tex
Gaussian & 2 & $-1.47$ & $-1.31\ (0.20)$ & $+0.82$ & $2.27$ \\
Gaussian & 5 & $-2.94$ & $-2.57\ (0.42)$ & $+0.90$ & $3.22$ \\
Gaussian & 8 & $-4.41$ & $-3.50\ (0.62)$ & $+1.48$ & $4.00$ \\
Scale mixture & 2 & $14.81$ & $13.33\ (1.06)$ & $-1.40$ & $8.14$ \\
Scale mixture & 5 & $25.56$ & $26.51\ (1.83)$ & $+0.52$ & $10.39$ \\
Scale mixture & 8 & $36.30$ & $35.37\ (2.74)$ & $-0.34$ & $12.85$ \\
Rademacher & 2 & $-9.80$ & $-9.51\ (0.29)$ & $+1.02$ & $2.97$ \\
Rademacher & 5 & $-15.68$ & $-16.65\ (0.62)$ & $-1.55$ & $6.20$ \\
Rademacher & 8 & $-21.56$ & $-20.94\ (0.94)$ & $+0.67$ & $8.13$ \\
Quadratic misspec. & 2 & $4.81$ & $4.75\ (0.35)$ & $-0.17$ & $3.12$ \\
Quadratic misspec. & 5 & $13.68$ & $14.10\ (0.96)$ & $+0.44$ & $4.99$ \\
Quadratic misspec. & 8 & $22.56$ & $22.19\ (1.72)$ & $-0.21$ & $6.77$ \\

%% file: tables/validity_body.tex
0.6 & 2.3 & 4 & 1.0072\,(1) & 1.0079 & 1.0019\,(1) & 1.0020 & 1.0004\,(1) & 1.0005 \\
0.8 & 3.7 & 12 & 1.0184\,(3) & 1.0243 & 1.0052\,(2) & 1.0061 & 1.0014\,(1) & 1.0015 \\
1.0 & 6.2 & 41 & 1.0429\,(7) & 1.0827 & 1.0137\,(4) & 1.0207 & 1.0043\,(2) & 1.0052 \\
1.2 & 11.2 & 166 & 1.0923\,(15) & 1.3323 & 1.0347\,(11) & 1.0831 & 1.0127\,(8) & 1.0208 \\
1.5 & 33.5 & 2025 & 1.2517\,(51) & 5.0503 & 1.1163\,(44) & 2.0126 & 1.0534\,(36) & 1.2531 \\

%% file: tables/tail_body.tex
Heavy-tailed example ($s=1.8$) & 1.68\,(3) & 78.5 & 1.90 & 7.94 & 26.6 \\
Table \ref{tab:validity}, $s=1.5$ & 1.25\,(1) & 5.05 & 1.36 & 2.44 & 3.81 \\
$\gamma=0$, $s=2$ & 2.00\,(6) & 672.1 & 2.45 & 17.3 & 97.4 \\

%% file: references.bib
@article{EfronStein1981,
  author  = {Efron, Bradley and Stein, Charles},
  title   = {The Jackknife Estimate of Variance},
  journal = {Annals of Statistics},
  year    = {1981},
  volume  = {9},
  number  = {3},
  pages   = {586--596}
}

@book{imbens2015causal,
  author    = {Imbens, Guido W. and Rubin, Donald B.},
  title     = {Causal Inference for Statistics, Social, and Biomedical Sciences},
  publisher = {Cambridge University Press},
  year      = {2015}
}

@techreport{FDA2023,
  author      = {{U.S. Food and Drug Administration}},
  title       = {Adjusting for Covariates in Randomized Clinical Trials for Drugs and
                 Biological Products: Guidance for Industry},
  institution = {U.S. Food and Drug Administration},
  month       = may,
  year        = {2023}
}

@book{fisher1935design,
  author    = {Fisher, Ronald A.},
  title     = {The Design of Experiments},
  publisher = {Oliver and Boyd},
  address   = {Edinburgh},
  year      = {1935}
}

@article{Freedman2008,
  author  = {Freedman, David A.},
  title   = {On Regression Adjustments to Experimental Data},
  journal = {Advances in Applied Mathematics},
  volume  = {40},
  number  = {2},
  pages   = {180--193},
  year    = {2008}
}

@article{Lin2013,
  author  = {Lin, Winston},
  title   = {Agnostic Notes on Regression Adjustments to Experimental Data:
             Reexamining {F}reedman's Critique},
  journal = {The Annals of Applied Statistics},
  volume  = {7},
  number  = {1},
  pages   = {295--318},
  year    = {2013}
}

@article{li2017general,
  author  = {Li, Xinran and Ding, Peng},
  title   = {General Forms of Finite Population Central Limit Theorems with
             Applications to Causal Inference},
  journal = {Journal of the American Statistical Association},
  volume  = {112},
  number  = {520},
  pages   = {1759--1769},
  year    = {2017}
}

@article{negi2021revisiting,
  author  = {Negi, Akanksha and Wooldridge, Jeffrey M.},
  title   = {Revisiting Regression Adjustment in Experiments with Heterogeneous
             Treatment Effects},
  journal = {Econometric Reviews},
  volume  = {40},
  number  = {5},
  pages   = {504--534},
  year    = {2021}
}

@misc{Wager2025,
  author = {Wager, Stefan},
  title  = {Causal Inference: A Statistical Learning Approach},
  note   = {Stanford University lecture notes},
  year   = {2025}
}

@article{zhao2023covariate,
  author  = {Zhao, Anqi and Ding, Peng},
  title   = {Covariate Adjustment in Multiarmed, Possibly Factorial Experiments},
  journal = {Journal of the Royal Statistical Society Series B: Statistical Methodology},
  volume  = {85},
  number  = {1},
  pages   = {1--23},
  year    = {2023}
}

@article{HodgesLehmann1970,
  author  = {Hodges, Joseph L. and Lehmann, Erich L.},
  title   = {Deficiency},
  journal = {The Annals of Mathematical Statistics},
  volume  = {41},
  number = {3},
  pages   = {783--801},
  year    = {1970}
}

@misc{Kranz2024,
  author = {Kranz, Sebastian},
  title  = {From Replications to Revelations: Heteroskedasticity-Robust Inference},
  note   = {arXiv:2411.14763},
  year   = {2024}
}

@article{imbens2009recent,
  author  = {Imbens, Guido W. and Wooldridge, Jeffrey M.},
  title   = {Recent Developments in the Econometrics of Program Evaluation},
  journal = {Journal of Economic Literature},
  volume  = {47},
  number  = {1},
  pages   = {5--86},
  year    = {2009}
}

@article{ChesherJewitt1987,
  author  = {Chesher, Andrew and Jewitt, Ian},
  title   = {The Bias of a Heteroskedasticity Consistent Covariance Matrix Estimator},
  journal = {Econometrica},
  volume  = {55},
  number  = {5},
  pages   = {1217--1222},
  year    = {1987}
}

@article{MacKinnonWhite1985,
  author  = {MacKinnon, James G. and White, Halbert},
  title   = {Some Heteroskedasticity-Consistent Covariance Matrix Estimators with
             Improved Finite Sample Properties},
  journal = {Journal of Econometrics},
  volume  = {29},
  number  = {3},
  pages   = {305--325},
  year    = {1985}
}

@article{LongErvin2000,
  author  = {Long, J. Scott and Ervin, Laurie H.},
  title   = {Using Heteroscedasticity Consistent Standard Errors in the Linear
             Regression Model},
  journal = {The American Statistician},
  volume  = {54},
  number  = {3},
  pages   = {217--224},
  year    = {2000}
}

@article{imbens2016robust,
  author  = {Imbens, Guido W. and Koles{\'a}r, Michal},
  title   = {Robust Standard Errors in Small Samples: Some Practical Advice},
  journal = {The Review of Economics and Statistics},
  volume  = {98},
  number  = {4},
  pages   = {701--712},
  year    = {2016}
}

@article{lei2021regression,
  author  = {Lei, Lihua and Ding, Peng},
  title   = {Regression Adjustment in Completely Randomized Experiments with a
             Diverging Number of Covariates},
  journal = {Biometrika},
  volume  = {108},
  number  = {4},
  pages   = {815--828},
  year    = {2021}
}

@book{BhattacharyaRao1976,
  author    = {Bhattacharya, Rabindra N. and Rao, R. Ranga},
  title     = {Normal Approximation and Asymptotic Expansions},
  publisher = {Wiley},
  address   = {New York},
  year      = {1976}
}

@book{Hall1992,
  author    = {Hall, Peter},
  title     = {The Bootstrap and Edgeworth Expansion},
  publisher = {Springer-Verlag},
  address   = {New York},
  year      = {1992}
}

@book{Ghosh1994,
  author    = {Ghosh, Jayanta K.},
  title     = {Higher Order Asymptotics},
  series    = {NSF-CBMS Regional Conference Series in Probability and Statistics},
  publisher = {Institute of Mathematical Statistics},
  year      = {1994}
}

@article{Sargan1974,
  author  = {Sargan, J. D.},
  title   = {The Validity of {N}agar's Expansion for the Moments of Econometric
             Estimators},
  journal = {Econometrica},
  volume  = {42},
  number  = {1},
  pages   = {169--176},
  year    = {1974}
}

@article{Nagar1959,
  author  = {Nagar, A. L.},
  title   = {The Bias and Moment Matrix of the General $k$-Class Estimators of the
             Parameters in Simultaneous Equations},
  journal = {Econometrica},
  volume  = {27},
  number  = {4},
  pages   = {575--595},
  year    = {1959}
}

@article{Sargan1976,
  author  = {Sargan, J. D.},
  title   = {Econometric Estimators and the {E}dgeworth Approximation},
  journal = {Econometrica},
  volume  = {44},
  number  = {3},
  pages   = {421--448},
  year    = {1976}
}

@incollection{Rothenberg1984,
  author    = {Rothenberg, Thomas J.},
  title     = {Approximating the Distributions of Econometric Estimators and Test
               Statistics},
  booktitle = {Handbook of Econometrics},
  editor    = {Griliches, Zvi and Intriligator, Michael D.},
  volume    = {2},
  chapter   = {15},
  pages     = {881--935},
  publisher = {North-Holland},
  year      = {1984}
}

@article{NeweySmith2004,
  author  = {Newey, Whitney K. and Smith, Richard J.},
  title   = {Higher Order Properties of {GMM} and Generalized Empirical Likelihood
             Estimators},
  journal = {Econometrica},
  volume  = {72},
  number  = {1},
  pages   = {219--255},
  year    = {2004}
}

@article{ChiangMatsushitaOtsu2026,
  author  = {Chiang, Harold D. and Matsushita, Yukitoshi and Otsu, Taisuke},
  title   = {Regression Adjustment in Completely Randomized Experiments with Many
             Covariates},
  journal = {arXiv preprint arXiv:2302.00469},
  year    = {2025}
}

@article{rohatgi1989sharp,
  author  = {Rohatgi, Vijay K. and Sz{\'e}kely, G{\'a}bor J.},
  title   = {Sharp Inequalities Between Skewness and Kurtosis},
  journal = {Statistics \& Probability Letters},
  volume  = {8},
  number  = {4},
  pages   = {297--299},
  year    = {1989}
}

@manual{estimatr,
  author = {Blair, Graeme and Cooper, Jasper and Coppock, Alexander and Humphreys, Macartan and Sonnet, Luke},
  title  = {estimatr: Fast Estimators for Design-Based Inference},
  year   = {2024},
  note   = {R package version 1.0.2},
  url    = {https://CRAN.R-project.org/package=estimatr}
}

@article{sympy,
  title = {{SymPy}: symbolic computing in {Python}},
  author = {Meurer, Aaron and Smith, Christopher P. and Paprocki, Mateusz
    and {\v{C}}ert{\'i}k, Ond{\v{r}}ej and Kirpichev, Sergey B.
    and Rocklin, Matthew and Kumar, AMiT and Ivanov, Sergiu
    and Moore, Jason K. and Singh, Sartaj and Rathnayake, Thilina
    and Vig, Sean and Granger, Brian E. and Muller, Richard P.
    and Bonazzi, Francesco and Gupta, Harsh and Vats, Shivam
    and Johansson, Fredrik and Pedregosa, Fabian and Curry, Matthew J.
    and Terrel, Andy R. and Rou{\v{c}}ka, {\v{S}}t{\v{e}}p{\'a}n
    and Saboo, Ashutosh and Fernando, Isuru and Kulal, Sumith
    and Cimrman, Robert and Scopatz, Anthony},
  journal = {PeerJ Computer Science},
  volume = {3},
  pages = {e103},
  year = {2017},
  doi = {10.7717/peerj-cs.103},
  url = {https://doi.org/10.7717/peerj-cs.103}
}

@article{Tan2014,
  author  = {Tan, Zhiqiang},
  title   = {Second-Order Asymptotic Theory for Calibration Estimators in Sampling and Missing-Data Problems},
  journal = {Journal of Multivariate Analysis},
  volume  = {131},
  pages   = {240--253},
  year    = {2014},
  doi     = {10.1016/j.jmva.2014.07.003}
}
